\documentclass[12pt]{article}
\usepackage{amsmath,amsthm,amssymb,amsfonts}
\usepackage{graphicx, float}
\usepackage{enumerate, enumitem, bigints}
\usepackage{tikz}
\usetikzlibrary{patterns}
\usepackage{pgfplots}
\usepackage{color} 

\usepackage{chngcntr}
\usepackage{apptools}

\newcommand{\blind}{1}

\newtheorem{proposition}{Proposition}[section]
\newtheorem{theorem}{Theorem}[section]
\newtheorem{lemma}{Lemma}[section]
\newtheorem{corollary}{Corollary}[section]
\newtheorem{remark}{Remark}[section]

\newcommand{\tabby}{\hspace{10pt}}

\newcommand{\vmat}[2]{\begin{pmatrix} #1 \\ #2 \end{pmatrix}}
\newcommand{\hmat}[2]{\begin{pmatrix} #1 & #2 \end{pmatrix}}

\newcommand{\barr}{\operatorname{Barr}}
\newcommand{\argmin}{\operatorname{argmin}}
\newcommand{\logdet}{\log\det}
\newcommand{\mappy}[1]{\overset{#1}{\longmapsto}}
\newcommand{\pdev}[2]{\frac{\partial #1}{\partial #2}}

\newcommand{\tp}{\intercal}
\newcommand*{\Scale}[2][4]{\scalebox{#1}{$#2$}}
\newcommand{\pp}{\mathrm{p}}
\usepackage{natbib}

\newcommand\numberthis{\addtocounter{equation}{1}\tag{\theequation}}

\usepackage[caption = false]{subfig} 

\usepackage{algorithm}
\usepackage{algorithmic}

\newlength\myindent
\usepackage{setspace}
\usepackage{hyperref}
\hypersetup{colorlinks,
		citecolor=blue,
		linkcolor=blue,
		urlcolor=black}

\usepackage{xr}

\usepackage{titling}
\predate{}
\postdate{}

\begin{document}

\def\spacingset#1{\renewcommand{\baselinestretch}%
{#1}\small\normalsize} \spacingset{1}


\if1\blind
{
  \title{\bf Approximate Post-Selective Inference for Regression with the Group LASSO}
  \author{
    Snigdha Panigrahi\\
    Department of Statistics, University of Michigan\\
    and  \\
    Peter W. MacDonald \\
    Department of Statistics, University of Michigan\\
    and \\
    Daniel Kessler\\
    Departments of Statistics and Psychiatry, University of Michigan
      }
    \date{}  
  \maketitle
} \fi

\if0\blind
{
  \bigskip
  \bigskip
  \bigskip
  \begin{center}
    {\LARGE\bf Approximate Post-Selective Inference for Regression with the Group LASSO}
\end{center}
  \medskip
} \fi

\bigskip
\begin{abstract}
After selection with the Group LASSO (or generalized variants such as the overlapping, sparse, or standardized Group LASSO), inference for the selected parameters is unreliable in the absence of adjustments for selection bias.
In the penalized Gaussian regression setup, existing approaches provide adjustments for selection events that can be expressed as linear inequalities in the data variables.
Such a representation, however, fails to hold for selection with the Group LASSO and substantially obstructs the scope of subsequent post-selective inference.
Key questions of inferential interest---for example, inference for the effects of selected variables on the outcome---remain unanswered.
In the present paper, we develop a consistent, post-selective, Bayesian method to address the existing gaps by deriving a likelihood adjustment factor and an approximation thereof that eliminates bias from the selection of groups.
Experiments on simulated data and data from the Human Connectome Project demonstrate that our method recovers the effects of parameters within the selected groups while paying only a small price for bias adjustment.
\end{abstract}


\newpage
\spacingset{1.3} 

\section{Introduction}
\label{sec:intro}

Modern statistical analysis of complex data does not always fit into the classical inferential framework.
Instead, analysis splits into two distinct stages: a {\em selection} stage, in which we formulate a model and hypotheses of interest; and an {\em inference} stage, in which we estimate parameters, quantify uncertainties, and test hypotheses under our selected model.
However, classical coverage guarantees for credible and confidence intervals fail dramatically when data used for selection is naively re-used for inference; see \cite{berk2013valid, exact_lasso, benjamini2020selective} and references therein.
Simple procedures like data splitting preserve validity of post-selective inference if two subsets of independent data are used for the selection and inference stages.
However, discarding all the data used in the selection stage is inefficient, and there is potential for methodology which can safely reuse a portion of the information from selection for valid inference.
By adopting a conditional approach, recent tools in selective inference reduce this wastefulness when selection algorithms are applied to data prior to statistical modeling and inference.
As examples, conditional methods by \cite{suzumura2017selective, zhao2019selective, gao2020selective, tanizaki2020computing} provide adjustments for selection bias in different post-selective inference tasks.

To briefly outline the essence of the conditional approach, consider a variable selection algorithm applied to data $Y$ with $p$ (fixed) predictors $X$.
Suppose the algorithm returns as output $\widehat{E}(Y)$, a subset of $\{1,2,\ldots,p\}$ such that each index represents a variable (column of $X$), and therefore $\widehat{E}(Y)$ is associated with a model selected from $2^p$ possibilities.
After selecting a given (nonempty) subset of variables $E$, our interest lies in inference for a set of post-selective parameters
$$\Theta_E = \left\{\theta^{(j)}_{E} \in \mathbb{R},\  j\in E\right\}$$
using the observed data $\left\{ Y = y \right\}$.
Post-selective inference for $\Theta_E$ proceeds by conditioning on the selection event 
$$\left\{ \widehat{E}(Y)=E \right\},$$
which is motivated by the fact that conditional coverage implies unconditional coverage under selection.
That is, for a set $C\left(\widehat{E}\left(Y\right),Y\right) \subseteq \mathbb{R}$ that depends on both the output of selection and the data, \cite{exact_lasso} note
\[
  \mathbb{P}\left(\theta^{(j)}_{\widehat{E}\left(Y\right)} \in C\left(\widehat{E}\left(Y\right),Y\right) \;\lvert \; \widehat{E}\left(Y\right)=E \right) \geq 1 - \alpha \ \Rightarrow  \ \mathbb{P}\left(\theta^{(j)}_{\widehat{E}\left(Y\right)} \in C\left(\widehat{E}\left(Y\right),Y\right) \right)\geq 1-\alpha.
\]
As it turns out, in many problems, it may be more convenient to instead condition on $\mathcal{A}_E$, where $\mathcal{A}_E \subseteq \left\{ \widehat{E}(Y) = E \right\}$.
Inference remains valid by the argument above even when conditioning on a proper subset of the selection event.

After conditioning on $\mathcal{A}_E$, post-selective inference may be carried out either via a frequentist or Bayesian framework.
A Bayesian framework in \cite{yekutieli2012adjusted, selective_bayesian} relies on a conditional, {\em selection-informed} likelihood to facilitate posterior sampling.
The Bayesian approach is especially useful for inferring about vector-valued parameters or functions thereof, and permits flexible inference in different models informed by selection, for instance, models with unknown noise variance.
In the remainder of this paper, we develop an approximate Bayesian method for post-selective inference with the Group LASSO and several of its variants.
The setup for our problem is the following: (i) the covariates act naturally in groups known a priori in the analysis; (ii) only a few of these groups of covariates affect the outcome, captured effectively by a parsimonious model.
A well-developed class of algorithms in \cite{group_lasso, jacobGroupLassoOverlap2009, simon2013sparse} among others exploits this knowledge about the covariate space in order to select regression models with grouped covariates.
Post-selective inference in the resulting selection-informed models is a natural next step that is addressed by our method.

We structure our paper as follows.
We begin by situating the contributions of our method in the post-selective literature in Section \ref{sec:relatedwork}.
In Section \ref{sec:select-corr-framework}, we present a selection-informed posterior that serves as the methodological centerpiece of our Bayesian framework.
In Section \ref{sec:methodology}, we obtain an exact value for a likelihood adjustment factor in our selection-informed posterior to eliminate bias from the selection of groups.
We then apply a generalized version of Laplace-type approximations to obtain feasible sampling updates from an approximate version of the posterior.
In Section \ref{sec:generalizations}, we generalize our method to models informed by different forms of grouped covariates.
We establish large-sample theory for our approximate Bayesian methods in Section \ref{sec:large-sample-theory}.
We demonstrate the potential of our methods in numerical experiments and in a human neuroimaging application in Section \ref{sec:empirical-analysis}.
Proofs for our technical results and further supporting information are included in the appendices.

\section{Related Work and Contributions}
\label{sec:relatedwork}
Below, we identify the challenges that preclude the use of existing methods and their immediate modifications for the Group LASSO.
Fixing some notation, suppose we observe $n$ independent instances of a scalar response variable $Y_i$ and a $p$-dimensional vector of covariates $X_i$ for $i=1, \ldots, n$.
We denote the response vector by $y = \begin{bmatrix} Y_1 & \ldots & Y_n \end{bmatrix}^{\intercal} \in \mathbb{R}^n$ and the corresponding covariate matrix by $X= \begin{bmatrix} X_1 & \ldots & X_n \end{bmatrix}^{\intercal} \in \mathbb{R}^{n \times p}$.
Let $\mathcal{G}$ be a prespecified partition of our $p$ covariates into $G$ groups.
We refer to a group in $\mathcal{G}$ by lowercase $g$, and use $\lvert g \rvert \in \mathbb{N}$ to denote the number of covariates within group $g$.

For now, we consider non-overlapping groups defined by the partition $\mathcal{G}$.
Suppose, we solve the familiar Group LASSO objective in \cite{group_lasso}:
\begin{equation}
\label{glasso:wor}
	\widehat{\beta}^{(\mathcal{G})} \in \operatorname*{argmin}_{\beta} \frac{1}{2} \lVert y - X\beta \rVert_2^2 + \sum_{g \in \mathcal{G}} \lambda_g \lVert \beta_g \rVert_2.
\end{equation}
For each group $g \in \mathcal{G}$, $\beta_g \in \mathbb{R}^{\lvert g \rvert}$ is a vector with entries corresponding to the covariates in group $g$, and $\lambda_g \geq 0$ is a tuning parameter for this group.
The solution of \eqref{glasso:wor} returns a subset of the covariates
\begin{equation*}
	\widehat{E}(y)= \operatorname{supp}(\widehat{\beta}^{(\mathcal{G})}),
\end{equation*}
where the support of the Group LASSO estimator respects the prespecified groups.
Specifically, the selected set of covariates can be written as a union of selected groups in $\mathcal{G}$, which we denote by $\mathcal{G}_{\widehat{E}}$ in the paper.

\subsection{From Atoms to Groups}
\label{sec:atoms-to-groups}

Consider the special case when each covariate forms an atomic group of size $1$, simply called an atom.
In this case, the objective in \eqref{glasso:wor} agrees with the widely studied LASSO.
Established in \cite{exact_lasso}, the selection of atoms is a polyhedral event, which means that the event is expressible as a union of linear inequalities in the response vector $y$.
Existing methods for post-selective inference in \cite{exact_lasso, suzumura2017selective, liu2018more} readily adjust for bias from selection by reducing the polyhedral conditioning event to univariate truncations.
However, when we transition from atoms to nontrivial groups, the selection of promising groups no longer results in polyhedral events.
We visualize this fact through Figure \ref{fig:select-viz} in a simple example, when the sample size and the number of predictors are both equal to $2$.
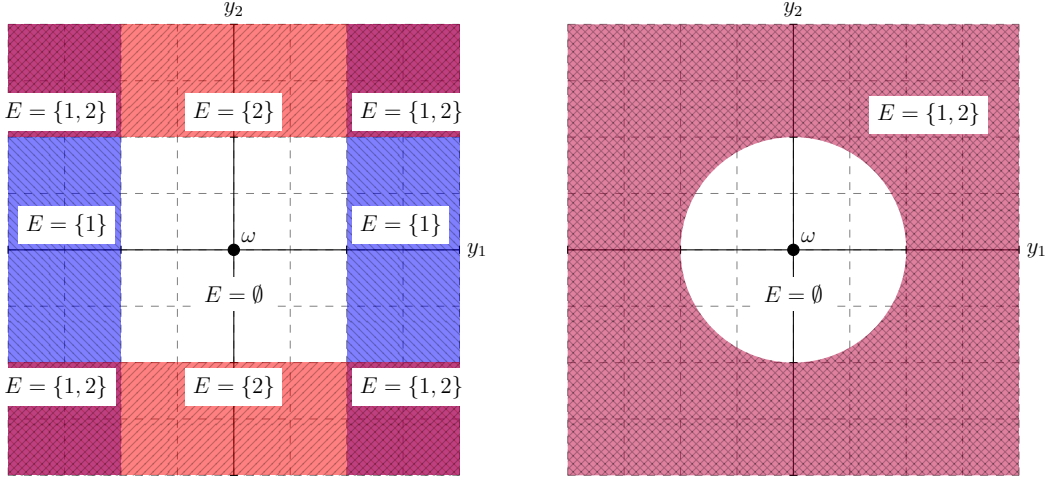
\begin{figure}[h]
  \centering
  \begin{tikzpicture}[scale=1.5, every node/.style={scale=0.72}]
    \draw[step=0.5, gray, dashed, help lines]  (-2, -2) grid (2, 2) ; 
    \foreach \x in {-2, ..., 2}{
      \draw (\x, -1pt) -- (\x, 1pt);
    }
    \foreach \y in {-2, ..., 2}{
      \draw (-1pt, \y) -- (1pt, \y);
    }

    \draw  (-2, 0) -- (2, 0) node[right]{$y_1$} ;
    \draw  (0, -2) -- (0, 2) node[above]{$y_2$} ;

    \filldraw (0, 0) circle[radius=0.05,color=black,fill=black] node [above right] {$\omega$};

    \begin{scope}[shift={(0, 0)}]

      \begin{scope}[opacity=0.5]
        \fill[preaction={fill=blue},pattern=north west lines] (1, -2) rectangle (2, 2) ;
        \fill[preaction={fill=blue},pattern=north west lines] (-2, -2) rectangle (-1, 2) ;
        \fill[preaction={fill=red},pattern=north east lines] (-2, -2) rectangle (2, -1) ;
        \fill[preaction={fill=red},pattern=north east lines] (-2, 1) rectangle (2, 2) ;
      \end{scope}

      \node[below=0.5cm, fill=white] at (0,0) {$E = \emptyset$} ;

      \node[above right=0.1cm, fill=white] at (1,0) {$E = \left\{ 1 \right\}$} ;
      \node[above left=0.1cm, fill=white] at (-1, 0) {$E = \left\{ 1 \right\}$} ;

      \node[above=0.1cm, fill=white] at (0,1) {$E = \left\{ 2 \right\}$} ;
      \node[below=0.1cm, fill=white] at (0, -1) {$E = \left\{ 2 \right\}$} ;

      \node[below left=0.1cm, fill=white] at (-1,-1) {$E = \left\{ 1, 2 \right\}$} ;
      \node[above left=0.1cm, fill=white] at (-1, 1) {$E = \left\{ 1, 2 \right\}$} ;
      \node[above right=0.1cm, fill=white] at (1,1) {$E = \left\{ 1, 2 \right\}$} ;
      \node[below right=0.1cm, fill=white] at (1,-1) {$E = \left\{ 1, 2 \right\}$} ;
    \end{scope}
  \end{tikzpicture} \qquad
  \begin{tikzpicture}[scale=1.5,every node/.style={scale=0.72}]
    \draw[step=0.5, gray, dashed, help lines]  (-2, -2) grid (2, 2) ; 
    \foreach \x in {-2, ..., 2}{
      \draw (\x, -1pt) -- (\x, 1pt);
    }
    \foreach \y in {-2, ..., 2}{
      \draw (-1pt, \y) -- (1pt, \y);
    }

    \draw  (-2, 0) -- (2, 0) node[right]{$y_1$} ;
    \draw  (0, -2) -- (0, 2) node[above]{$y_2$} ;

    \filldraw (0, 0) circle[radius=0.05,color=black,fill=black] node [above right] {$\omega$};

    \begin{scope}[shift={(0, 0)}]

      \begin{scope}[opacity=0.5]
        \fill[preaction={fill=purple},pattern=crosshatch,even odd rule] (-2, -2) rectangle (2, 2) (0, 0) circle [radius=1] ;
      \end{scope}

      \node[fill=white, below=0.5cm] at (0,0) {$E = \emptyset$} ;
      \node[fill=white] at (1.2,1.2) {$E = \left\{ 1, 2 \right\}$} ;
    \end{scope}
  \end{tikzpicture}
  \caption{
    Geometry of selection events for the LASSO (left) and the Group LASSO (right) as a function of $(y_1, y_2)$.
    In the case of no randomization, the origin $\omega$ is the point $\left( 0, 0 \right)$, but see text surrounding \eqref{glasso} for discussion of how randomization affects the origin.
    The LASSO can select any of $\emptyset, \left\{ 1 \right\}, \left\{ 2 \right\}$, or $\left\{ 1, 2 \right\}$ as predictors; the Group LASSO can select $\emptyset$ or $\left\{ 1, 2 \right\}$.}
  \label{fig:select-viz}
\end{figure}

In Figure \ref{fig:select-viz}, we contrast the geometry of the selection event for the LASSO and the Group LASSO.
Our covariates are the columns of an identity matrix and the tuning parameters are set to be $1$.
Under the grouped scenario, the two orthogonal covariates comprise a single group, whereas in the LASSO each covariate is an atom.
For the LASSO, the event leading to the selection of the active set $E$ is a union of rectangular regions in the plane that are highlighted by the same color.
A proper subset of this event is obtained by further restricting the signs of selected covariates to match the observed signs.
This proper subset leads to one of the rectangular regions in the plane; see left panel.
In contrast, the selection of an active group for the Group LASSO is depicted as the complement of a ball in the right panel, which can no longer be characterized as a union of polyhedral events.

\subsection{Post-selective Inference for Overall Group Effects} \label{sec:overall-effects}

We now turn to recent results by \cite{loftus2015selective,yang2016selective} which provide post-selective inference for overall effects of groups after solving the Group LASSO.
Introducing some more notation, let $\mathcal{U}$ represent an operator that maps the vector $v$ to the unit vector $(\|v\|_2)^{-1} \cdot v$.
For a linear subspace $\mathbb{S} \subseteq \mathbb{R}^n$ and its orthogonal complement $\mathbb{S}^{\perp}$, let $\mathcal{P}_{\mathbb{S}}$ and $\mathcal{P}_{\mathbb{S}^\perp}$ denote the projection operators onto the subspaces $\mathbb{S}$ and $\mathbb{S}^{\perp}$ respectively.

Consider solving the Group LASSO in \eqref{glasso:wor}.
Let $E$ be the realized value of $\widehat{E}$.
Suppose we assume the simple model: $y \sim \mathcal{N}_n(\mu, \sigma^2 I_n)$ for inference.
For $g\in \mathcal{G}_{E}$ and the subspace $\mathbb{S}_{g, E} = \text{span}\left(\mathcal{P}_{X^{\perp}_{\mathcal{G}_{E}\setminus g}}(X_g)\right)$, consider the post-selective parameter
$$\mu_g= \|\mathcal{P}_{\mathbb{S}_{g,E}}(\mu)\|_2 \in \mathbb{R}$$
after selection with the Group LASSO.
A significant $p$-value under the null hypothesis $H_{0,g}: \mu_g=0$ confirms the presence of the selected group $g$ in the estimated support; confidence bounds for $\mu_g$ measure the overall effect of the selected group $g$.
The main result by \cite{yang2016selective} allows post-selective inference for $\mu_g$ through a conditional distribution for $\|\mathcal{P}_{\mathbb{S}_{g,\widehat{E}}}(y)\|_2$, which we revisit in the following lemma.
\begin{lemma}\emph{\cite{yang2016selective}.}
\label{yang}
Conditional upon the event
\begin{equation}
\label{cond:event:yang}
\left\{y: \widehat{E}(y)=E, \; \mathcal{U}\left(\mathcal{P}_{\mathbb{S}_{g,\widehat{E}}}(y)\right) = U_g, \; \mathcal{P}_{\mathbb{S}^\perp_{g,\widehat{E}}}(y) = W_g\right\},
\end{equation}
the density for $\left\lVert {\mathcal{P}_{\mathbb{S}_{g,\widehat{E}}}(y)} \right\rVert_2$ at $\gamma_g$ is proportional to:
\begin{equation*}
\gamma_g^{|g|-1}\cdot \exp\left(-\frac{1}{2\sigma^2}\left(\gamma_g^2 - 2\gamma_g \cdot U_g^{\intercal}\mu\right)\right)\cdot 1_{\mathcal{R}_{E}}(\gamma_g),
\end{equation*}
where $\mathcal{R}_{E}= \left\{\gamma_g \in \mathbb{R}^{+} : \widehat{E}\left(U_g \gamma_g+ W_g\right)=E\right\}$.
\end{lemma}
Applying a probability integral transform to the conditional law in Lemma \ref{yang} produces a pivot for
$$\widetilde\mu_g =U_g^{\intercal}\mu$$
conditional upon \eqref{cond:event:yang}.
In particular, $\widetilde\mu_g$ agrees with $\mu_g$ under the null $H_{0,g}$.
As a result, a pivot for the former parameter yields a valid $p$-value for testing $H_{0,g}$ and coincides with the $p$-value in \cite{loftus2015selective}.
The two parameters, however, do not coincide in general.
Instead, the following relation holds by Cauchy-Schwarz:
$$\mu_g \geq \widetilde\mu_g,$$
and inverting the pivot thus provides a conservative, lower confidence bound for $\mu_g$.

As emphasized in the preceding discussion in Section \ref{sec:relatedwork}, the selection of groups is no longer a polyhedral event.
Indeed, the difficulties posed by the non-polyhedral geometry for the Group LASSO continue to persist; we note that the truncating region $\mathcal{R}_{E}$ in Lemma~\ref{yang} lacks a closed-form description.
The outlined approach overcomes this barrier to some extent by narrowing down the scope of inferential targets to conducting inference on overall group effects, in which case one only needs to explore a positive half-line to approximately compute $\mathcal{R}_{E}$.
Besides lacking an upper confidence bound for overall group effects, the existing approach does not yield interval estimates for the effects of the individual variables in the selected groups, nor does it identify a joint distribution for the individual effects.

\subsection{Our method} \label{sec:our-method}

Closing existing gaps, we develop a Bayesian method for post-selective inference after conducting a randomized selection of groups.
Our method accounts for the non-polyhedral selection of groups via a likelihood adjustment factor and characterizes a {\em selection-informed} posterior distribution based on the likelihood adjustment.
Working with a selection-informed posterior grants us the flexibility to estimate the individual effects within selected groups and functions thereof through credible regions and general posterior expectations.
At the same time, a randomized selection of groups permits us a very simple and exact characterization for the truncating region in the conditional likelihood that makes subsequent inference easily feasible.

The randomizing variable, or {\em randomization}, used for the selection of groups is a Gaussian variable throughout the remainder of the paper and is hereafter termed Gaussian randomization.
Our methods based on Gaussian randomization are closely related to data carving proposals in \cite{optimal_inference, selective_bayesian, panigrahi2019carving, schultheiss2021multicarving}, wherein selection operates only on a subset of the samples, but subsequent post-selective estimation uses the full data.
The variance of the Gaussian randomization is a tuning parameter analogous to the split proportion in data splitting, providing us control of the relative amount of information used in selecting a group-sparse model and estimating the post-selective parameters.
The information borrowed by our approach from selection yields credible intervals which are shorter than the corresponding interval estimates for data splitting with roughly the same information split.

In the following sections, we develop our method in two steps.
First, we account for the selection of groups, a non-polyhedral event, via an exact likelihood adjustment factor.
Rather than characterizing the non-polyhedral event in the space of the data and randomization variables, we develop a change of variables in the polar coordinate system that is motivated by ideas in \cite{harris2016selective}.
Our choice of conditioning event is characterized by simple sign constraints in the new variables which yields us a selection-informed posterior distribution.
In the next step, we propose a computationally feasible surrogate for this posterior distribution with a (generalized) Laplace approximation.
Our Bayesian method delivers statistically consistent estimates using a selection-informed posterior distribution for the group-sparse parameter vector.
Continuing with our simple grouped example introduced earlier and depicted in the right panel of Figure \ref{fig:select-viz}, Figure \ref{fig:posterior-concentration} serves to preview the distribution of samples from our surrogate selection-informed posterior by varying the number of observations $n=25,50,100,250,500,1000$.
Assuredly, as $n$ increases, the support of the posterior concentrates around the true bivariate parameter, suggesting the statistical consistency of our method which we justify theoretically in Section~\ref{sec:large-sample-theory}.

\begin{figure}[h]
  \centering
  \includegraphics[width=\textwidth]{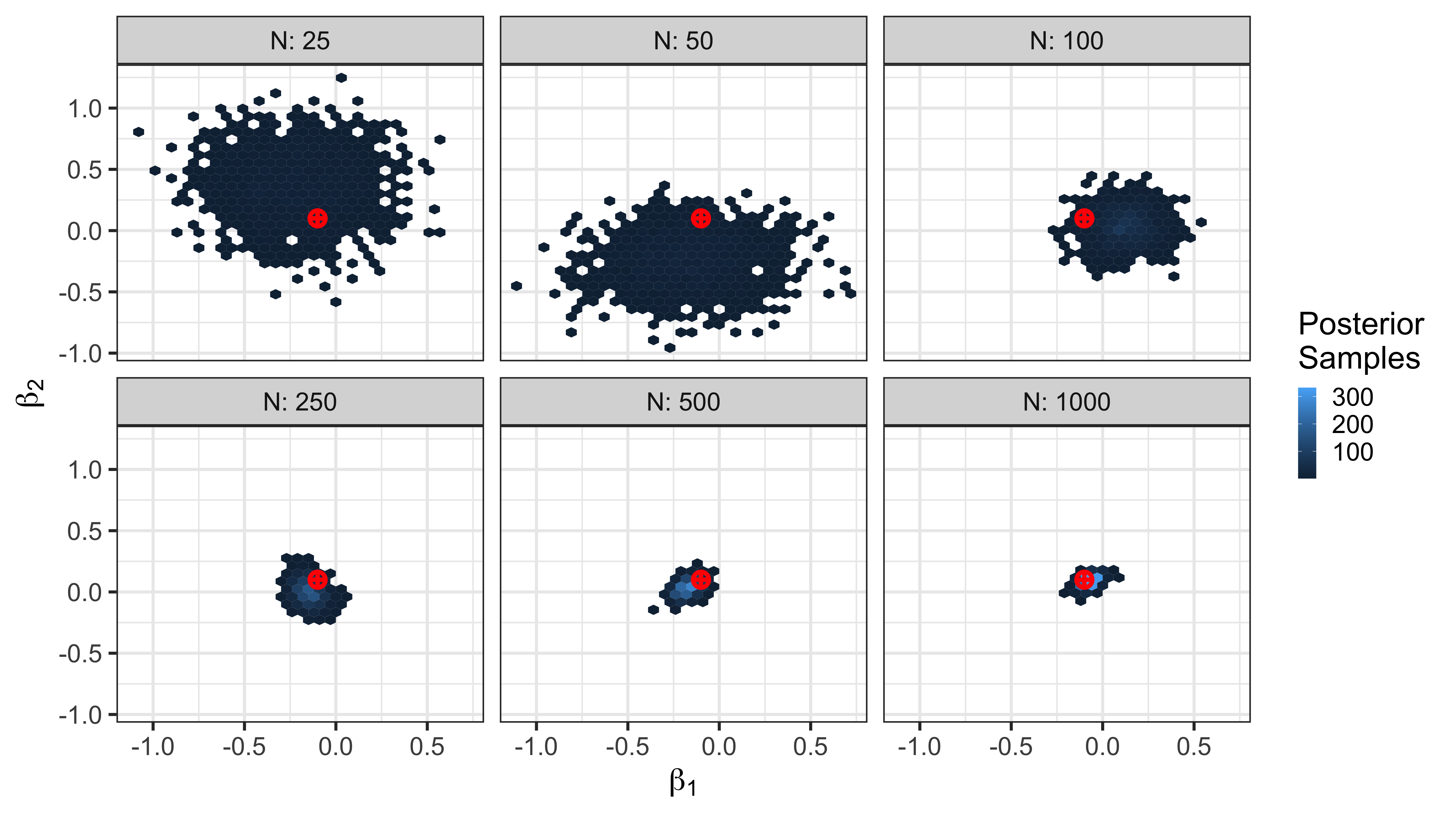}
  \caption{
    Distribution of samples from a surrogate of our selection-informed posterior based on a bivariate model with a single group of orthogonal covariates for varying $n$.
    The color of each hex depicts the number of posterior samples drawn in that region, while a red crosshair indicates the location of the true value of the parameter $\beta =
    \protect\begin{bmatrix}
      -0.1 & 0.1
    \protect\end{bmatrix}
    ^{\intercal}$.
    }
  \label{fig:posterior-concentration}
\end{figure}

\section{Framework for Selection-informed Inference}
\label{sec:select-corr-framework}

\subsection{Basic Setup}
\label{sec:basic-setup}

Consistent with a post-selective setting, under a fixed $X$ regression, our problem proceeds in two stages: first, we select promising groups by optimizing an objective inducing grouped sparsity; then, we specify a group-sparse linear model informed by the groups of covariates learned from the previous stage.
We begin describing our methods for non-overlapping groups, based on a prespecified partition $\mathcal{G}$ of $p$ covariates into $G$ groups.
Later in Section~\ref{sec:generalizations}, we present a larger category of grouped sparsities that our methods successfully encompass.

Using notation defined in Section~\ref{sec:relatedwork}, we consider the Group LASSO objective in \eqref{glasso:wor} with an added randomization term:
\begin{equation}
\label{glasso}
	\widehat{\beta}^{(\mathcal{G})} \in \operatorname*{argmin}_{\beta} \left\{ \frac{1}{2} \lVert y - X\beta \rVert_2^2 + \sum_{g \in \mathcal{G}} \lambda_g \lVert \beta_g \rVert_2 - \omega^{\intercal}\beta \right\}.
\end{equation}
In the final term of this objective, $\omega \sim \mathcal{N}_p(0,\Omega)$ is a Gaussian randomization variable independent of the data.
As indicated previously, perturbing the optimization problem with a Gaussian randomization variable $\omega$ introduces a tradeoff between selection and inference, giving the user the ability to reserve some information from the selection stage to perform inference.
Additional discussion of the role of randomization in \eqref{glasso} and the relation of randomization variance with data splitting is given in Section~\ref{sec:empirical-analysis}.
Hereafter, focusing on the solution of \eqref{glasso}, we let
$$\widehat{E}= \operatorname{supp}(\widehat{\beta}^{(\mathcal{G})})$$
be the support of the randomized Group LASSO estimator and let $\mathcal{G}_{\widehat{E}}$ be the selected groups of covariates according to the estimated support.

Revisiting the example in Section \ref{sec:relatedwork} and the related Figure \ref{fig:select-viz}, we note that the selection regions have a similar geometry with the added randomization: the randomization instance $\omega$ merely shifts the origin in both panels of the figure.
Elaborating on the example, suppose that $n=p=2$, $X= I_n$, the identity matrix and $\lambda_g=\lambda=1$.
Let $\omega\sim \mathcal{N}_2(0,I_2)$.
The stationary mapping for the optimization in (\ref{glasso}) is given by:
\begin{equation} \label{kkt_cond}
	y + \omega = \widehat{\beta}^{(\mathcal{G})} +  z,
\end{equation}
where the final term is the subgradient of the Group LASSO penalty evaluated at the solution.
In the case that the single group of two covariates is not selected, $\widehat{\beta}^{(\mathcal{G})}=0$ and $\|z\|_2 <1$.
For any fixed $\omega$, the collection of $y$ that leads to no selection is equivalent to a ball centered at $\omega$.
The complement of this region characterizes the selection of the group of size $2$.
Instead, when we have two atoms (groups with size $1$ each), i.e., we solve a randomized version of the LASSO in \cite{randomized_response}, the selection of a subset of covariates with fixed signs is equivalent to linear inequalities in $y$ and $\omega$.
Once again, shifting the origin in the left panel of Figure \ref{fig:select-viz} to $\omega$ depicts the polyhedral selection event for the randomized LASSO.
Recent work by \cite{selective_bayesian, panigrahi2019approximate} provide a likelihood after the randomized LASSO; but, these methods are not applicable to the present problem, because the selection of groups with size greater than $1$ does not admit a polyhedral form.

In the next stage, we specify a model after selection.
Letting $E$ be the realized value of $\widehat{E}$, we model our response as
\begin{equation} \label{sel_model}
	y \sim \mathcal{N}_n(X_E \beta_E, \sigma^2 I_n).
\end{equation}
The selected model in \eqref{sel_model}, using the solution of \eqref{glasso}, may indeed be misspecified.
Suppose, the true distribution for our response is
\begin{equation*}
  y \sim \mathcal{N}_n(X\bar{\beta},\bar\sigma^2 I_n),
\end{equation*}
for some $\bar{\beta} \in \mathbb{R}^{p}$, $\bar\sigma^2 \in \mathbb{R}^+$.
Under the true distribution, our method delivers inference for the best linear representation of the response mean using the selected covariates $X_E$, regardless of model misspecification.
We elaborate further on this point when we turn to a selection-informed likelihood based on the model in \eqref{sel_model}.

By analogy with \cite{yekutieli2012adjusted, panigrahi2018scalable}, we pose a selection-informed prior for our post-selective parameter
\begin{equation} \label{prior}
	\beta_E \sim \pi_E
\end{equation}
to invoke a Bayesian framework after selection.
Both the selected model and the selection-informed prior depend on the observed data.
However, they do so only through the selection event accounted for by conditioning.

Two additional comments are in order here to highlight the flexibilities our framework offers in terms of defining models post selection.
One, without loss of generality, we are able to assume that the variance parameter $\sigma$ is known.
Following the lines of \cite{selective_bayesian}, the Bayesian approach we take easily accommodates the case of unknown variance by treating it as a parameter and posing a joint selection-informed prior on $\beta_E$ and $\sigma$.
Two, the model in \eqref{sel_model} can be more general.
For instance, our model may be parameterized by a realization for $\widehat{E}'$ specified through an arbitrary function of $\widehat{E}$, as is pursued in \citet{panigrahi2020integrative}.
The adjustment for selection in any case must account for the non-polyhedral selection of promising groups.
We proceed with the selected model \eqref{sel_model} to simplify the development.

\subsection{Selection-informed Posterior} \label{sec:selinf-posterior}

In this section we define a selection-informed posterior using the model for $y$ in \eqref{sel_model} and the prior in \eqref{prior}.
Through the remaining paper, we use the notation $\pp(\mu, \Sigma; b)$ for a normal density function with mean $\mu$ and covariance $\Sigma$ evaluated at $b$.
To lay out the selection-informed posterior, we define the data variables involved in selection: (i) the randomization variable $\omega$;
(ii) the least squares estimate based on $(X_E, y)$
\begin{equation*}
	\widehat{\beta}_E = \left( X_E^{\intercal}X_E \right)^{-1} X_E^{\intercal} y;
\end{equation*}
(iii) the orthogonal projection $N_E = X^{\intercal} \left( I_n - X_E \left(X_E^{\intercal}X_E \right)^{-1}X_E^{\intercal} \right) y$, assuming $X_E$ is full rank.
Under the selected model, $\widehat{\beta}_E$ has mean $\beta_E$, and $N_E$ has mean $0$.
Denote the covariance of $\widehat{\beta}_E$ by $\Sigma_E = \sigma^2 \left( X_E^{\intercal}X_E \right)^{-1}$ and let $\Psi_E$ be the covariance of $N_E$.
Ignoring selection, the usual joint likelihood for these three variables is given by:
\begin{equation} \label{like_unadjust}
	\pp(\beta_E, \Sigma_E ; \widehat{\beta}_E)\cdot  \pp(0, \Psi_E; N_E) \cdot \pp(0, \Omega; \omega).
\end{equation}
The factorization follows directly from the independence between $\widehat{\beta}_E$ and $N_E$ and their independence with the randomization variable $\omega$.

Accounting for the selection-informed nature of our model, the likelihood we work with conditions upon an event:
\begin{equation} \label{sel:event:gen}
\mathcal{A}_E \subseteq \{(\widehat{\beta}_E, N_E, \omega): \widehat{E} = E\}.
\end{equation}
The conditioning event $\mathcal{A}_E$ for the group-sparse problem is a proper subset of the selection event $\{\widehat{E} = E\}$ based on the KKT conditions for \eqref{glasso}, which we define precisely in Theorem~\ref{thm1}.
After truncating realizations to the event $\mathcal{A}_E$, the corresponding conditional likelihood is proportional to
\begin{equation*}
	 \pp(\beta_E, \Sigma_E ; \widehat{\beta}_E)\cdot  \pp(0, \Psi_E; N_E) \cdot \pp(0, \Omega; \omega)  \cdot \mathbf{1}_{\mathcal{A}_E}(\widehat{\beta}_E, N_E,\omega).
\end{equation*}

Now we state our selection-informed likelihood, derived after conditioning further upon the ancillary statistic $N_E$ and integrating out the randomization variable $\omega$.
Up to proportionality in $\beta_E$, the expression for this likelihood agrees with
\begin{equation}\label{sel:lik}
         \left\{ \mathbb{P}(\mathcal{A}_{E; N_E} \  | \ \beta_E) \right\}^{-1} \cdot \pp(\beta_E, \Sigma_E ; \widehat{\beta}_E)\cdot \int \pp(0, \Omega; \omega)  \cdot \mathbf{1}_{\mathcal{A}_{E; N_E}}(\widehat{\beta}_E,\omega) d\omega;
\end{equation}
$ \mathcal{A}_{E;N_E}$ is the set of $\widehat{\beta}_E$, $\omega$ that result in the event $\mathcal{A}_E$ for the fixed instance $N_E$ and
$$\mathbb{P}(\mathcal{A}_{E; N_E} \  | \ \beta_E)=  \int  \pp(\beta_E, \Sigma_E ; \widehat{\beta}_E) \cdot \pp(0, \Omega; \omega)  \cdot \mathbf{1}_{\mathcal{A}_{E; N_E}}(\widehat{\beta}_E,\omega) d\omega d\widehat{\beta}_E,$$
where $\mathbb{P}(\cdot \  | \ \beta_E)$ highlights the dependence of the probability for the event $ \mathcal{A}_{E;N_E}$ on the post-selective parameters $\beta_E$.

More generally, the selection-informed likelihood in \eqref{sel:lik} yields us inference for the best linear representation of the response mean in terms of the selected covariates.
To note this generality, say, our response is generated from the linear model: $y \sim \mathcal{N}_n(X\bar\beta,\bar\sigma^2 I_n)$.
For any fixed set $E$ with size $\left\lvert E \right\rvert$, we have $\widehat{\beta}_E \sim \mathcal{N}_{|E|} (\beta_E, \Sigma_E)$ and $N_E \sim \mathcal{N}_{p}(\xi_E, \Psi_E)$ where
$$\beta_E = \underset{b\in \mathbb{R}^{|E|}}{\text{argmin}} \|X\bar\beta - X_E b\|_2^2,$$
and $\xi_E =X^{\intercal}  \mathcal{P}_{\mathbb{S}_E^{\perp}}(X\bar\beta)$ for $\mathbb{S}_E = \text{span}(X_E)$.
The likelihood for $\widehat{\beta}_E$, $N_E$ and $\omega$ factorizes as
\begin{equation} \label{like_unadjust_full}
  \pp(\beta_E, \Sigma_E ; \widehat{\beta}_E)\cdot  \pp(\xi_E, \Psi_E; N_E) \cdot \pp(0, \Omega; \omega).
\end{equation}
Treating $\xi_E$ as nuisance parameters post selection, we condition on $N_E$ to obtain a likelihood function of $\beta_E$, free from nuisance parameters.
It is easy to see that our selection-informed likelihood assumes the expression in \eqref{sel:lik} and inference proceeds identically, regardless of model misspecification.

Using \eqref{sel:lik} in conjunction with our selection-informed prior \eqref{prior} ultimately yields us our selection-informed posterior distribution for $\beta_E$:
\begin{equation} \label{select_posterior}
	   \pi_E(\beta_E\  | \ \widehat{\beta}_E, N_E)\propto (\mathbb{P}(\mathcal{A}_{E; N_E} \  | \ \beta_E))^{-1}\cdot\pi_E(\beta_E)\cdot  \pp(\beta_E, \Sigma_E ; \widehat{\beta}_E).
\end{equation}
In contrast to frequentist approaches that rely on a truncated law (for example, the approach taken by \cite{harris2016selective}), the selection-informed posterior is fully supported on $\mathbb{R}^d$ where $d$ is the number of parameters within our selection-informed model.
For this reason, Bayesian post-selective inference successfully avoids some of the difficulties faced when sampling from truncated laws with complicated support sets.

Evaluating $\mathbb{P}(\mathcal{A}_{E; N_E}\  | \ \beta_E)$, called the likelihood {\em adjustment factor} in \cite{selective_bayesian}, is rightly recognized as the prime technical hurdle in carrying out selection-informed Bayesian inference.
Through a careful choice for the conditioning event $\mathcal{A}_{E; N_E}$ after applying a change of variables, we develop mathematical expressions for the adjustment factor and update estimates in a feasible analytic form for the group-sparse problem.
We take this up in the next section.

\section{Selection-informed Bayesian Methods}
\label{sec:methodology}

\subsection{An Exact Adjustment Factor}
\label{subsec:adjustmentfactor}

We begin by identifying an exact theoretical value for the likelihood adjustment factor in our selection-informed posterior \eqref{select_posterior}.
With a slight abuse of notation, hereon, we denote the event $\mathcal{A}_{E; N_E}$ by $\mathcal{A}_E$ and the associated adjustment factor by $\mathbb{P}(\mathcal{A}_{E} \  | \ \beta_E)$.

Back to our primary case study of the (non-overlapping) Group LASSO, we express the non-zero solution for a selected group, $g \in \mathcal{G}_E$, in the polar coordinate system as
$$\widehat{\beta}^{(\mathcal{G})}_g= \|\widehat{\beta}^{(\mathcal{G})}_g\|_2\cdot  \frac{\widehat{\beta}^{(\mathcal{G})}_g}{\|\widehat{\beta}^{(\mathcal{G})}_g\|_2}= \gamma_g u_g,
$$
where $\gamma_g= \|\widehat{\beta}^{(\mathcal{G})}_g\|_2 > 0$ is a scalar representing the size of the selected group and $u_g= \frac{\widehat{\beta}^{(\mathcal{G})}_g}{\|\widehat{\beta}^{(\mathcal{G})}_g\|_2}$ is a vector in $\mathbb{R}^{\lvert g \rvert}$ satisfying $\lVert u_g \rVert_2 = 1$.
The stationary mapping for \eqref{glasso} at the solution is given by:
$$
\omega = X^{\intercal} X
    \begin{pmatrix}
      \left( \gamma_g u_g \right)_{g \in \mathcal{G}_E} \\
      0
    \end{pmatrix}
    - X^{\intercal} X_E \widehat{\beta}_E - N_E +
    \begin{pmatrix}
      \left( \lambda_g u_g \right)_{g \in \mathcal{G}_E} \\
      \left( \lambda_g z_g \right)_{g \in -\mathcal{G}_E}
    \end{pmatrix}.
$$
$\begin{pmatrix}\left( \lambda_g u_g \right)^{\intercal}_{g \in \mathcal{G}_E} & \left( \lambda_g z_g \right)^{\intercal}_{g \in -\mathcal{G}_E}\end{pmatrix}^{\intercal}$ is the subgradient of the $\ell_2$-norm Group LASSO penalty at the solution,
where $z_g$ is a vector in $\mathbb{R}^{\lvert g \rvert}$ satisfying $\lVert z_g \rVert_2 < 1$ for each non-selected group $g \in -\mathcal{G}_E$.
We collect the following optimization variables:
\begin{equation*}
 \begin{aligned}
	\widehat{\gamma} &= (\gamma_g : g \in \mathcal{G}_E)^{\intercal} \in \mathbb{R}^{\lvert \mathcal{G}_E \rvert}, \
	\widehat{\mathcal{U}} = \{ u_g : g \in \mathcal{G}_E \}, \
	\widehat{\mathcal{Z}} = \{ z_g : g \in -\mathcal{G}_E \},
 \end{aligned}
\end{equation*}
calling their respective realizations $\gamma$, $\mathcal{U}$ and $\mathcal{Z}$.
Letting $\operatorname{diag}(\cdot)$ operate on an ordered collection of matrices and return the corresponding block diagonal matrix, we fix $U = \operatorname{diag}\left( \left(u_g\right)_{g \in \mathcal{G}_E}\right)$.
Then, based on the stationary mapping from the Group LASSO, define
\begin{equation} \label{KKT_map}
  \phi_{\widehat{\beta}_E}(\widehat{\gamma}, \widehat{\mathcal{U}}, \widehat{\mathcal{Z}})
  = A \widehat{\beta}_E + B(\widehat{\mathcal{U}})\widehat{\gamma}  + c(\widehat{\mathcal{U}}, \widehat{\mathcal{Z}}),
\end{equation}
where
$$
A = -X^{\intercal}X_E,\;  B= X^{\intercal}X_E U, \; c = - N_E +
	    \begin{pmatrix}
	      \left( \lambda_g u_g \right)^{\intercal}_{g \in \mathcal{G}_E} &
	      \left( \lambda_g z_g \right)^{\intercal}_{g \in -\mathcal{G}_E}
	    \end{pmatrix}^{\intercal}.
$$

Theorem \ref{thm1} gives us the expression for the adjustment factor after applying the change of variables:
\begin{equation}
\label{CoV}
\omega \to (\widehat{\gamma}, \widehat{\mathcal{U}}, \widehat{\mathcal{Z}}), \ \text{ where } \ (\widehat{\gamma}, \widehat{\mathcal{U}}, \widehat{\mathcal{Z}})= \phi_{\widehat{\beta}_E}^{-1}(\omega).
\end{equation}
For each $g \in \mathcal{G}_E$, we construct the orthonormal basis completion for $u_g$ that we denote by $\bar{U}_g\in \mathbb{R}^{\lvert g \rvert \times \lvert g \rvert-1}$.
Further, $x>t$ for $x\in \mathbb{R}^k$ and $t\in \mathbb{R}$ simply means that the inequality holds in a coordinate-wise sense.

\begin{theorem} \label{thm1}
Consider the conditioning event
\begin{equation*}
\label{cond:event:GL}
\mathcal{A}_E = \{\widehat{E}=E, \ \widehat{\mathcal{U}}= \mathcal{U}, \ \widehat{\mathcal{Z}}= \mathcal{Z}\}.
\end{equation*}
Define the following matrices
	\begin{equation*}
	       \begin{aligned}
			&\;\;\;\;\;\;\;\;\;\;\bar{U} = \operatorname{diag}\left( \left(\bar{U}_g\right)_{g \in \mathcal{G}_E}\right), \Gamma = \operatorname{diag}\left( \left( \widehat{\gamma}_g I_{\lvert g \rvert - 1} \right)_{g \in \mathcal{G}_E} \right),  \Lambda = \operatorname{diag}\left( \left(\lambda_g I_{|g|}\right)_{g \in \mathcal{G}_E}\right).
	     \end{aligned}
     \end{equation*}
Then, we have
	\begin{equation*}
	        \begin{aligned}
		\mathbb{P}(\mathcal{A}_E \  | \  \beta_E) &\propto\int \int \pp(\beta_E, \Sigma_E ; \widehat{\beta}_E) \cdot \exp \left\{ -\frac{1}{2}( A \widehat{\beta}_E + B\widehat{\gamma}  + c)^{\intercal} \Omega^{-1} (A \widehat{\beta}_E + B\widehat{\gamma}  + c) \right\} \\
		&\;\;\;\;\;\;\;\;\;\;\;\;\;\;\;\;\;\;\;\;\;\;\;\;\;\;\;\;\;\;\;\;\;\;\;\;\;\;\;\;\;\;\;\;\;\;\;\;\;\;\;\;\;\;\;\;\;\;\;\;\;\;\;\;\; \times J_{\phi_{\widehat{\beta}_E}}(\widehat{\gamma};\mathcal{U}, \mathcal{Z}) \cdot \mathbf{1}(\widehat{\gamma} > 0) d\widehat{\beta}_E d\widehat{\gamma},
		\end{aligned}
	\end{equation*}
	where
	\begin{equation} \label{jacobian_final}
		J_{\phi_{\widehat{\beta}_E}}(\widehat{\gamma};\mathcal{U}, \mathcal{Z}) = \operatorname{det} \left( \Gamma + \bar{U}^{\intercal}(X_E^{\intercal}X_E)^{-1}\Lambda \bar{U} \right).
	\end{equation}
\end{theorem}

The non-polyhedral event we set out to analyze is characterized exactly through the adjustment factor in Theorem \ref{cond:event:GL}.
This exact characterization is possible due to the choice of conditioning event as well as the specific form of randomization.
Drawing an analogy to the conditioning event for the LASSO in \cite{exact_lasso}, conditioning on $\widehat{\mathcal{U}}= \mathcal{U}$ is similar to their required conditioning on the sign of each selected coefficient, where we interpret the sign as the univariate special case of vector direction in multiple dimensions.
By conditioning further upon $\widehat{\mathcal{Z}}= \mathcal{Z}$, we avoid an integration over $p-|E|$ variables.
Furthermore, the specific form of randomization in \eqref{glasso} allows us a characterization of our conditioning event $\mathcal{A}_E$ in terms of simple sign constraints on $\widehat{\gamma}$, the sizes ($\ell_2$ norms) of the selected groups.
We note that randomization in other forms will enable a tradeoff in the relative amount of information between selection and inference, but, may not yield a computationally feasible likelihood as obtained above with a linear, additive randomization term in \eqref{glasso}.

In our likelihood adjustment, $J_{\phi_{\widehat{\beta}_E}}(\widehat{\gamma}; \mathcal{U}, \mathcal{Z})$ represents the Jacobian associated with the change of variables $\phi_{\widehat{\beta}_E}(\cdot)$.
Noticing the dependence of this function on simply $\widehat{\gamma}$ and the observed $\mathcal{U}$, we call this function
$J_{\phi}(\cdot, \mathcal{U})$.
In the special case where the design matrix of the selected model is orthogonal, the Jacobian takes a much simpler form which is given in Corollary \ref{cor1}.
\begin{corollary} \label{cor1}
	Suppose $X_E^{\intercal}X_E = I_{|E|}$. Then
	\begin{equation*}
		J_{\phi}(\gamma, \mathcal{U}) = \textstyle\prod_{g \in \mathcal{G}_E} (\gamma_g + \lambda_g)^{(|g|-1)}.
	\end{equation*}
\end{corollary}
\noindent The proof of Corollary~\ref{cor1} is a direct calculation based on \eqref{jacobian_final} and is omitted.

In comparison with the adjustment for polyhedral selection events in \cite{selective_bayesian}, the selection of groups leads to a nontrivial Jacobian function $J_{\phi}(\cdot, \mathcal{U})$ in our likelihood adjustment.
It is easy to note that the Jacobian dissolves as a constant when all the selected groups are atoms with sizes are exactly equal to $1$.
Based on the event $\mathcal{A}_E$, our likelihood adjustment in this special case (with a constant Jacobian) gives an adjustment for the randomized LASSO.

\subsection{Surrogate Selection-informed Posterior}
\label{sec:approx-posterior}

Plugging in the adjustment factor from Theorem \ref{thm1} into \eqref{select_posterior} gives us the selection-informed posterior.
Proposition \ref{prop:1} simplifies the expression for this posterior further expressing it in terms of Gaussian densities.
We defer the details for matrices $\bar{A}$, $\bar{R}$, $\bar{b}$, $\bar{s}$, $\bar{\Theta}$ and $\bar{\Omega}$, which do not depend on $\widehat{\beta}_E$ or $\widehat{\gamma}$, to the appendices.
\begin{proposition}
\label{prop:1}
Conditioning upon the event $\mathcal{A}_E$ in Theorem \ref{thm1}, the selection-informed posterior in \eqref{select_posterior} agrees with
\begin{equation*}
\begin{aligned}
& \left(\int J_{\phi}(\widehat{\gamma}, \mathcal{U})  \cdot \pp(\bar{R}\beta_E+ \bar{s}, \bar{\Theta} ; \widehat{\beta}_E)\cdot  \pp(\bar{A}\widehat{\beta}_E+ \bar{b}, \bar{\Omega} ; \widehat{\gamma})\cdot \mathbf{1}(\widehat{\gamma} > 0) d\widehat{\gamma} d\widehat{\beta}_E\right)^{-1} \\
&\;\;\;\;\;\;\;\;\;\;\;\;\;\;\;\;\;\;\;\;\;\;\;\;\;\;\;\;\;\;\;\;\;\;\;\;\;\;\;\;\;\;\;\;\;\;\;\;\;\;\;\;\;\;\;\;\;\;\;\;\;\;\;\;\;\;\;\;\;\;\;\;\;\;\;\;\;\;\;\;\;\;\;\;\;\;\;\;\;\;\;\;\;\times \pi_E(\beta_E) \cdot \pp(\bar{R}\beta_E+ \bar{s}, \bar{\Theta};\widehat{\beta}_E).
\end{aligned}
\end{equation*}
\end{proposition}

Bypassing integrations, we provide easy-to-implement deterministic expressions for a surrogate selection-informed posterior and the corresponding gradient in Theorem \ref{thm2}.
Let
\begin{equation*}
\begin{aligned}
&\beta_E^{\star}, \gamma^{\star}=  \operatorname*{argmin}_{\widetilde{\beta}_E, \widetilde{\gamma}}\ \Big\{\dfrac{1}{2}(\widetilde{\beta}_E-\bar{R} \beta_E- \bar{s} )^{\intercal} (\bar{\Theta})^{-1} (\widetilde{\beta}_E-\bar{R} \beta_E- \bar{s})  \\
& \;\;\;\;\;\;\;\;\;\;\;\;\;\;\;\;\;\;\;\;\;\;\;\;\;\;\;\;\;\; + \dfrac{1}{2} (\widetilde{\gamma}-\bar{A}\widetilde{\beta}_E -\bar{b})^{\intercal} (\bar{\Omega})^{-1} (\widetilde{\gamma}-\bar{A}\widetilde{\beta}_E -\bar{b}) + \barr (\widetilde\gamma)\Big\},
\end{aligned}
\end{equation*}
where $\barr ( \cdot )$ is a barrier penalty \citep{auslender1999penalty} that takes the value $\infty$ when the support constraints are violated and imposes a smaller penalty for values farther away from the boundary of the positive orthant.
We use $C$ to denote a constant free of $\beta_E$.
A generalized version of the Laplace approximation \citep{wong2001asymptotic, inglot2014simple} for the normalizing constant in Proposition \ref{prop:1}:
\begin{equation*}
\int J_{\phi}(\widehat{\gamma}, \mathcal{U})  \cdot \pp(\bar{R}\beta_E+ \bar{s}, \bar{\Theta} ; \widehat{\beta}_E)\cdot  \pp(\bar{A}\widehat{\beta}_E+ \bar{b}, \bar{\Omega} ; \widehat{\gamma})\cdot \mathbf{1}(\widehat{\gamma} > 0) d\widehat{\gamma} d\widehat{\beta}_E \ \ \ \ \ \ \ \ \ \ \ \  \ \ \ \ \
\end{equation*}
\begin{align}
&\approx C\cdot J_{\phi}(\gamma^{\star}; \mathcal{U}) \cdot \exp\Big(-\dfrac{1}{2}(\beta^{\star}_E-\bar{R} \beta_E- \bar{s} )^{\intercal} (\bar{\Theta})^{-1} (\beta^{\star}_E-\bar{R} \beta_E- \bar{s}) \nonumber \\
&\;\;\;\;\;\;\;\;\;\;\;\;\;\;\;\;\;\;\;\;\;\;\;\;\;\;\;\;\;\;\;\;\;\;- \dfrac{1}{2} (\gamma^{\star}-\bar{A}\beta^{\star}_E -\bar{b})^{\intercal} (\bar{\Omega})^{-1} (\gamma^{\star}-\bar{A}\beta^{\star}_E -\bar{b}) - \barr (\gamma^{\star})\Big) \numberthis
\label{sel:post:gen:Laplace}
\end{align}
is the basis of our surrogate posterior.
Using convex analysis \citep{rockafellar2015convex}, we detail the surrogate selection-informed posterior and gradient in the following theorem.

\begin{theorem}
  \label{thm2}
  Fixing $\bar{\Sigma}= \bar{\Omega} + \bar{A} \bar{\Theta} (\bar{A})^{\intercal}$, $\bar{P}= \bar{A}\bar{R}$ and $\bar{q}= \bar{A}\bar{s} + \bar{b}$, let
  \begin{equation}
  \label{optimizer:main}
  \gamma^{\star} = \operatorname*{argmin}_{\gamma\in \mathbb{R}^{|\mathcal{G}_E|}}\ \dfrac{1}{2}(\gamma-\bar{P}\beta_E -\bar{q})^{\intercal} (\bar{\Sigma})^{-1} (\gamma-\bar{P}\beta_E -\bar{q}) + \barr (\gamma).
  \end{equation}
Letting $\Gamma^{\star} = \operatorname{diag}\left( \left( \gamma^{\star}_g I_{\lvert g \rvert - 1} \right)_{g \in \mathcal{G}_E} \right)$
 and $M_g$ be the set of $\lvert g \rvert - 1$ diagonal indices of $\Gamma^{\star}$ for group $g$, define
$$J^{\star}= \left(\sum_{i \in M_g} [(\Gamma^{\star} + \bar{U}^{\intercal}(X_E^{\intercal}X_E)^{-1}\Lambda \bar{U})^{-1}]_{ii}\right)_{g\in \mathcal{G}_E} \in \mathbb{R}^{\lvert \mathcal{G}_E \rvert}.$$
Using the approximation in \eqref{sel:post:gen:Laplace}, the logarithm of our surrogate selection-informed posterior is given by:
\begin{equation*}\label{log:posterior}
\log \pi_E(\beta_E) +  \log \pp(\bar{R}\beta_E+ \bar{s}, \bar{\Theta} ; \widehat{\beta}_E) -\log \pp(\bar{P}\beta_E + \bar{q}, \bar{\Sigma} ; \gamma^{\star})  + \barr (\gamma^{\star})- \log J_{\phi}(\gamma^{\star}; \mathcal{U}),
\end{equation*}
with the following gradient:
\begin{equation*}\label{grad:log:posterior}
\nabla  \log \pi_E(\beta_E) + \bar{R}^{\intercal}(\bar{\Theta})^{-1}(\widehat{\beta}_E -\bar{R} \beta_E-\bar{s}) +\bar{P}^{\intercal} (\bar{\Sigma})^{-1}  \Big( \bar{P}\beta_E + \bar{q} -\gamma^{\star}- \left( \bar{\Sigma}^{-1} +  \nabla^2 \barr(\gamma^{\star}) \right)^{-1} J^{\star}\Big).
\end{equation*}
\end{theorem}

Recall, $|\mathcal{G}_E|$ is the number of selected groups after solving \eqref{glasso}.
Evident from Theorem \ref{thm2}, gradient-based sampling from the surrogate posterior requires us to solve a $\mathbb{R}^{|\mathcal{G}_E|}$-dimensional optimization problem in every update, without carrying out integrations for the theoretical adjustment.
Algorithm \ref{alg:updates} outlines a prototype implementation of our methods to generate estimates for the group-sparse parameters using the surrogate posterior.

\begin{algorithm}[H]
 \caption{A Prototype Implementation of our Selection-informed Bayesian Method}
 \label{alg:updates}
\begin{spacing}{2.}
\begin{algorithmic}
\REQUIRE ($y$, $X$, $\Omega$, $\mathcal{G}$, $\{\lambda_g\}_{g\in \mathcal{G}}$) $\overset{\text{Optimize} \;  \eqref{glasso}}{\longrightarrow}$ $\widehat{E}=E, \ \widehat{\mathcal{U}}= \mathcal{U}, \ \widehat{\mathcal{Z}}= \mathcal{Z}$
\ENSURE \text{\textbf{Set up parameters for (Laplace)}}
\STATE {\text{(Orthonormal completion)}}
Calculate $\bar{U}$ (see Theorem~\ref{thm1})
\STATE {\text{(Parameters)}}
Calculate $\bar{R}, \bar{s}, \bar{\Theta}, \bar{P}, \bar{q}, \bar{\Sigma}$ (see Theorem~\ref{thm2})
\ENSURE \text{\textbf{Implementation for a generic gradient-based sampler}}
\STATE (Initialize) Sample: $\beta_E^{(1)} = \widehat{\beta}_E$, Step Size: $\eta$, Proposal Scale: $\chi$, Number of Samples: $K$
\FOR {$k= 1,2,\cdots, K-1$}
 \STATE \textbf{(Laplace)} Solve $\gamma^{\star (k)} = \operatorname*{argmin}_{\gamma\in \mathbb{R}^{|\mathcal{G}_E|}}\ \Big\{\dfrac{1}{2}(\gamma-\bar{P}\beta_E^{(k)} -\bar{q})^{\intercal} (\bar{\Sigma})^{-1} (\gamma-\bar{P}\beta_E^{(k)} -\bar{q})$ \\ \hspace{70mm}$+ \barr (\gamma)\Big\}$
 \STATE (Jacobian)
 Calculate  $J^{\star (k)}$ (see  Theorem~\ref{thm2})
 \STATE (Gradient)  Calculate $\nabla \log \pi_E(\beta_E^{(k)}\  | \ \widehat{\beta}_E, N_E)$ (see Theorem~\ref{thm2})
 \STATE  (Update) $\beta_E^{(k+1)} \longleftarrow \beta_E^{(k)} + \eta \chi \nabla \log \pi_E(\beta_E^{(k)}\  | \ \widehat{\beta}_E, N_E) + \sqrt{2\eta} \epsilon^{(k)}$, $\epsilon^{(k)}\sim \mathcal{N} (0, \chi)$.
\ENDFOR
\end{algorithmic}
\end{spacing}
\end{algorithm}

We revisit our simple running example in Section \ref{sec:select-corr-framework} to instantiate Algorithm \ref{alg:updates}.
In the selection stage, we solve \eqref{glasso} with $\omega \sim \mathcal{N}_2(0, \tau^2 I_2)$ where $\tau^2$ is the randomization variance.
Before noting the updates from the surrogate posterior, we assess in Figure \ref{fig:rel:approx} the relative accuracy of the (generalized) Laplace approximation in \eqref{sel:post:gen:Laplace} with respect to the exact normalizing constant.
Because we have exactly one selected group of covariates, the exact normalizer is a one-dimensional integral that can be computed numerically.
Especially, we observe how the relative accuracy of the approximation varies with sample size $n$ and randomization variance $\tau^2$.
For any fixed sample size, the accuracy of approximation for the integral with the mode decreases as the randomization variance increases or, equivalently the concentration of probability mass in the integrand has a greater spread.
As expected, we observe that the relative accuracy converges to $0$ with growing sample size for all values of randomization variance.
\begin{figure}[H]
  \centering
  \includegraphics[height=10cm, width=\textwidth]{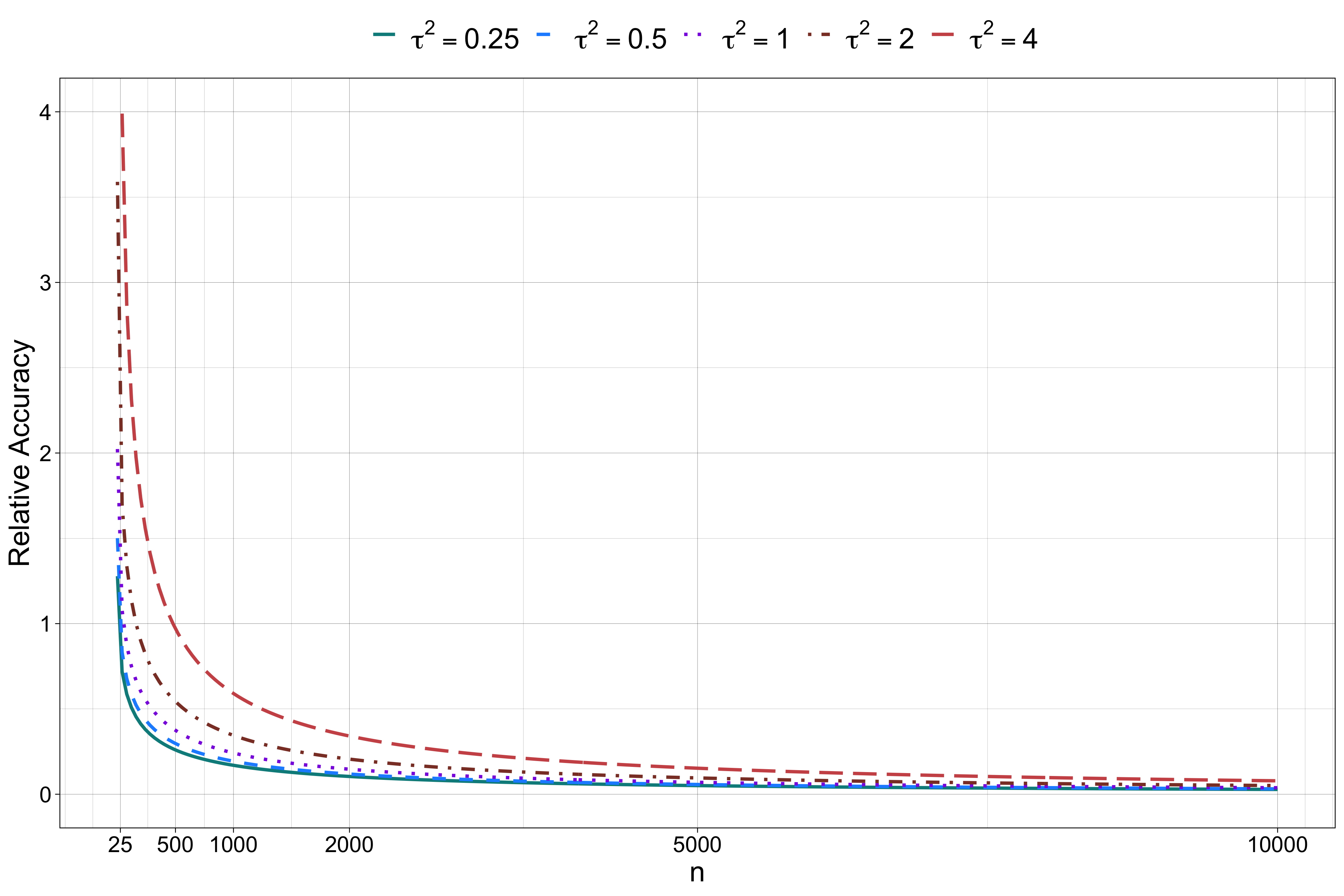}
  \vspace{-1cm}
  \caption{Plot for relative accuracy of the (generalized) Laplace approximation with respect to the exact normalizing constant.}
  \label{fig:rel:approx}
\end{figure}

Now, we exemplify Algorithm \ref{alg:updates} for the simple example.
Applying Theorems \ref{thm1} and \ref{thm2}, \textbf{(Laplace)} in this instance solves the one-dimensional optimization
$$\gamma^{\star (k)} = \operatorname*{argmin}_{\gamma\in \mathbb{R}}\ \Big\{\dfrac{1}{2}(\gamma-\bar{P}\beta_E^{(k)} -\bar{q})^{\intercal} (\bar{\Sigma})^{-1} (\gamma-\bar{P}\beta_E^{(k)} -\bar{q})+ \barr (\gamma)\Big\}$$
where $\bar{P}= U^{\tp}(2\cdot I_2- UU^{\tp})^{-1}$, $\bar{q}=-1$, $\bar{\Sigma}=1+ U^{\tp}(2\cdot I_2- UU^{\tp})^{-1} U$ and $U$ is the unit vector in $\mathbb{R}^2$ from writing the Group LASSO solution in the polar coordinate system.
The Jacobian function is given by $J^{\star (k)}= (1+\gamma^{\star (k)})^{-1}$.
We generate our prototype update for $\beta_E^{(k+1)}$ that depends on the expression for the gradient of the surrogate in Theorem \ref{thm2}.

At last, we briefly comment on the case when $\sigma$, the error variance in the data, is treated as an unknown parameter.
Using a joint prior on $(\beta_E, \sigma)$ in conjunction with our selection-informed likelihood gives us Bayesian inference in this situation.
As prescribed in \cite{selective_bayesian}, one could run a Gibbs sampler that alternates between drawing (i) an update for $\beta_E$ given $\sigma$ and (ii) an update for $\sigma$ given $\beta_E$.
Note, the updates for $\beta_E$ based on a gradient-based sampler will assume the expression in Theorem \ref{thm2}, except now the gradient of the prior is with respect to the parameter vector $\beta_E$.

\section{Generalization to other Grouped Sparsities}
\label{sec:generalizations}

We now generalize our selection-informed methods to learning algorithms which target other forms of grouped sparsities and covariate structures.

\subsection{Overlapping Group LASSO}
\label{sec:overl-group-lasso}

One approach for the overlapping Group LASSO recovers superposed groups by augmenting the covariate space with duplicated predictors \citep{jacobGroupLassoOverlap2009}.
We proceed by implementing the solver \eqref{glasso} with a randomization variable $\omega^* \sim \mathcal{N}_{p^{*}} (0,\Omega)$ such that $p^*$ is the number of covariates after duplication.
To formalize the setting, we let  $X^* \in \mathbb{R}^{n \times p^*}$ denote the augmented matrix of covariates constructed from $X \in \mathbb{R}^{n \times p}$, given that these overlapping groups are determined before selection.
Each set of selected covariates $E^*$ in the augmented space maps to a set of selected variables $E$ in the original space by reversing the duplication.
The stationary mapping for the overlapping Group LASSO with augmented matrix $X^*$ and randomization $\omega^*$, which we call $\phi^*$, is given by
\begin{equation}
\label{KKT:oglasso}
  \omega^*
  = X^{*\intercal} X^*
    \begin{pmatrix}
      \left( \gamma^*_g u^*_g \right)_{g \in \mathcal{G}_{E^*}} \\
      0
    \end{pmatrix}
    - X^{*\intercal} X_E \hat{\beta}_E -N_E +
    \begin{pmatrix}
      \left( \lambda_g u^*_g \right)_{g \in \mathcal{G}_{E^*}} \\
      \left( \lambda_h z^*_h \right)_{h \in -\mathcal{G}_{E^*}}
    \end{pmatrix}
\end{equation}
where $\hat{\beta}_E$ and $N_E= (X^*)^{\intercal}(y-X_E \hat{\beta}_E)$ are the refitted and the ancillary statistics in our selected model (cf. (\ref{sel_model})).
Following our notation,
\begin{equation*}
	\widehat\gamma^* = (\gamma^*_g : g \in \mathcal{G}_{E^*})^{\intercal}, \
	\widehat{\mathcal{U}}^* = \{ u^*_g : g \in \mathcal{G}_{E^*}\}, 
\end{equation*}
represent a polar decomposition of the overlapping Group LASSO solution; let $U^* = \operatorname{diag}\left( \left(u^*_g\right)_{g \in \mathcal{G}_{E^*}}\right)$.
Finally,
\begin{equation*}
	\widehat{\mathcal{Z}}^* = \{z^*_g : g \in -\mathcal{G}_{E^*}\}
\end{equation*}
are the subgradient variables from the Group LASSO penalty for the non-selected groups in the augmented predictor space.

The conditioning event we study for an analytically feasible selection-informed posterior is given by
\begin{equation*}
	\mathcal{A}_{E^*} = \{\widehat{E}^*=E^*,\ \widehat{\mathcal{U}}^* = \mathcal{U}^*, \ \widehat{\mathcal{Z}}^* = \mathcal{Z}^*\}
\end{equation*}
where $E^*$, $\mathcal{U}^*$ and $\mathcal{Z}^*$ are the corresponding observed instances.
We then recover an expression for the adjustment factor along the lines of Theorem~\ref{thm1} using the matrices
\begin{equation}\label{A:B:c:oglasso}
A = -(X^*)^{\intercal}X_E, \quad B= (X^*)^{\intercal}X^*_{E^*} U^*, \quad c = - N_E +
		    \begin{pmatrix}
		      \left( \lambda_g u^*_g \right)^{\intercal}_{g \in \mathcal{G}_E} &
		      \left( \lambda_g z^*_g \right)^{\intercal}_{g \in -\mathcal{G}_E}
		    \end{pmatrix}^{\intercal}
\end{equation}
from the mapping in \eqref{KKT:oglasso}.

\begin{proposition}
\label{adj:oglasso}
	Define the following matrices
	\begin{equation*}
	       \begin{aligned}
			&\bar{U}^* = \operatorname{diag}\left( \left(\bar{U}^*_g\right)_{g \in \mathcal{G}_{E^*}}\right), \Gamma^* = \operatorname{diag}\left( \left( \widehat\gamma^*_g I_{\lvert g \rvert} \right)_{g \in \mathcal{G}_{E^*}} \right),  \Lambda = \operatorname{diag}\left( \left(\lambda_g I_{|g|}\right)_{g \in \mathcal{G}_{E^*}}\right)\\
	     \end{aligned}
     \end{equation*}
     where $A$, $B$, $c$ are specified in \eqref{A:B:c:oglasso}.
     Then, $\mathbb{P}(\mathcal{A}_{E^*} \  | \  \beta_E)$ is proportional to
		\begin{equation*}
		        \begin{aligned}
			& \int \int \pp(\beta_E, \Sigma_E ; \widehat{\beta}_E) \cdot \exp \left\{ -\frac{1}{2}( A \widehat{\beta}_E + B\widehat{\gamma}^*  + c)^{\intercal} \Omega^{-1} (A \widehat{\beta}_E + B\widehat{\gamma}^*  + c) \right\} \\
		&\;\;\;\;\;\;\;\;\;\;\;\;\;\;\;\;\;\;\;\;\;\;\;\;\;\;\;\;\;\;\;\;\;\;\;\;\;\;\;\;\;\;\;\;\;\;\;\;\;\;\;\;\;\;\;\;\;\;\;\;\;\;\;\;\; \times J_{\phi^*}(\widehat{\gamma}^*;\mathcal{U}^*) \cdot \mathbf{1}(\widehat{\gamma}^* > 0) d\widehat{\beta}_E d\widehat{\gamma}^*
			\end{aligned}
		\end{equation*}
		where
		\begin{equation} \label{jacobian_overlap}
			J_{\phi^*}(\widehat\gamma^*;\mathcal{U}^*) = \operatorname{det} \begin{pmatrix} \left( (X^*_{E^*})^{\tp}X^*_{E^*}\Gamma^* + \Lambda \right)\bar{U}^* & (X^*_{E^*})^{\tp}X^*_{E^*} U^* \end{pmatrix}.
		\end{equation}
\end{proposition}

Notice, since $X^*$ contains overlapping groups, the matrix $(X^*_{E^*})^{\tp}X^*_{E^*}$ may not be invertible.
Thus, in comparison to Theorem~\ref{thm1}, (\ref{jacobian_overlap}) provides a different expression for the Jacobian function involving the sizes of the selected groups of variables.
In solving \eqref{glasso}, one may introduce a ridge penalty $\epsilon \|\beta\|_2^2/2$, where $\epsilon$ is a small positive number. This will in turn lead to \eqref{A:B:c:oglasso} with
$$
B=\begin{bmatrix} (X_E^*)^{\intercal}X^*_{E^*} +\epsilon\cdot  I_{E^*} \\  (X^*_{-E^*})^{\intercal}X^*_{E^*} \end{bmatrix} U^*.
$$
This results in the following Jacobian
$$
J_{\phi^*}(\widehat{\gamma}^*;\mathcal{U}^*) = \operatorname{det} \left( \Gamma^* + \left( \bar{U}^* \right)^{\intercal} \left( \left( X^*_{E^*} \right)^{\tp} X^*_{E^*} +\epsilon\cdot  I_{E^*} \right)^{-1}\Lambda \bar{U}^*  \right)
$$
where the ridge parameter can specifically be used to counter the collinearity in the augmented predictor matrix.

\subsection{Standardized Group LASSO}
\label{sec:std-group-lasso}

An alternate treatment to the Group LASSO objective is popularly applied in problems with correlated covariates when a within-group orthonormality is desired.
The learning algorithm proposed by \citet{simon2012standardization} addresses the selection of groups in the presence of such correlations via a modification to the Group LASSO penalty.
Equivalently, the canonical objective \eqref{glasso} is reparameterized in the standardized formulation; for each submatrix $X_g$ containing the predictors in group $g \in \mathcal{G}$,
the quadratic loss function is now given by
\begin{equation}
    \lVert y - \sum_g X_g \beta_g \rVert_2^2 = \lVert y - \sum_g W_g \theta_g \rVert_2^2,
\end{equation}
$X_g = W_g R_g$ for $W_g$, an orthonormal matrix and $R_g$, an invertible matrix and $\theta_g = R_g \beta_g$, the reparameterized vector.
The standardized Group LASSO optimizes the Group LASSO objective in terms of $\theta$ rather than $\beta$:
\begin{equation}\label{std:glasso}
  \widehat{\theta}^{(\mathcal{G})} = \argmin_{\theta} \left\{ \frac{1}{2}\lVert y - \sum_{g\ \in \mathcal{G}} W_g \theta_g \rVert_2^2 + \sum_{g \in \mathcal{G}} \lambda_g \lVert \theta_g \rVert_2 - \omega^{\tp}\theta \right\}.
\end{equation}
Lastly, the original parameters of interest are estimated by $R_g^{-1}\widehat{\theta}^{(\mathcal{G})}$.
The selected set of groups $E$ comprises groups with non-zero coordinates in $\widehat{\theta}^{(\mathcal{G})}$ after solving \eqref{std:glasso}.


Letting $W \in \mathbb{R}^{n \times p}$ be the column-wise concatentation of the standardized groups of covariates $W_g$, the stationary mapping for the standardized Group LASSO is given by
\begin{equation}
  \omega
  = W^{\intercal} W
    \begin{pmatrix}
      \left( \gamma_g u_g \right)_{g \in \mathcal{G}_E} \\
      0
    \end{pmatrix}
    - W^{\intercal} X_E \widehat{\beta}_E - N_E +
    \begin{pmatrix}
      \left( \lambda_g u_g \right)_{g \in \mathcal{G}_E} \\
      \left( \lambda_h z_h \right)_{h \in -\mathcal{G}_E}
    \end{pmatrix};
\end{equation}
$\widehat{\beta}_E$ is our usual refitted statistic and $N_E = W^{\intercal} (Y- X_E \widehat{\beta}_E)$.
Consistent with our approach, a polar decomposition of the (non-zero) standardized Group LASSO solution $\widehat{\theta}^{(\mathcal{G})}$ is represented via
\begin{equation*}
	\gamma = (\gamma_g : g \in \mathcal{G}_{E}), \ \
	\mathcal{U} = \{ u_g : g \in \mathcal{G}_{E}\}. 
\end{equation*}
Recall,
\begin{equation*}
	\mathcal{Z} = \{z_g : g \in -\mathcal{G}_{E}\} 
\end{equation*}
are the subgradient variables for the non-selected groups. Setting
\begin{equation*}
	\mathcal{A}_E = \{\widehat{E}=E,\ \widehat{\mathcal{U}}=\mathcal{U}, \ \widehat{\mathcal{Z}}=\mathcal{Z} \}, 
\end{equation*}
\begin{equation}\label{A:B:c:stdglasso}
A = -W^{\intercal}X_E, \quad B= W^{\intercal}W_E U,  \quad c = - N_E +
		     \begin{pmatrix}
		      \left( \lambda_g u_g \right)^{\intercal}_{g \in \mathcal{G}_E} &
		      \left( \lambda_g z_g \right)^{\intercal}_{g \in -\mathcal{G}_E}
		    \end{pmatrix}^{\intercal},
\end{equation}
we present the adjustment factor in line with Theorem~\ref{thm1} after a change of variables from inverting the stationary mapping of \eqref{std:glasso}.

\begin{proposition}
\label{adj:stdglasso}
	Consider the following matrices
		\begin{equation*}
		       \begin{aligned}
				&\bar{U} = \operatorname{diag}\left( \left(\bar{U}_g\right)_{g \in \mathcal{G}_E}\right), \Gamma = \operatorname{diag}\left( \left( \widehat\gamma_g I_{\lvert g \rvert - 1} \right)_{g \in \mathcal{G}_E} \right), \Lambda = \operatorname{diag}\left( \left(\lambda_g I_{|g|}\right)_{g \in \mathcal{G}_E}\right),
                        \end{aligned}
		\end{equation*}
		We then have
		\begin{equation*}
		        \begin{aligned}
			\mathbb{P}(\mathcal{A}_E \  | \  \beta_E) &\propto  \int \int \pp(\beta_E, \Sigma_E ; \widehat{\beta}_E) \cdot \exp \left\{ -\frac{1}{2}( A \widehat{\beta}_E + B\widehat{\gamma}  + c)^{\intercal} \Omega^{-1} (A \widehat{\beta}_E + B\widehat{\gamma}  + c) \right\} \\
		&\;\;\;\;\;\;\;\;\;\;\;\;\;\;\;\;\;\;\;\;\;\;\;\;\;\;\;\;\;\;\;\;\;\;\;\;\;\;\;\;\;\;\;\;\;\;\;\;\;\;\;\;\;\;\;\;\;\;\;\;\;\;\;\;\; \times J_{\phi}(\widehat{\gamma};\mathcal{U}) \cdot \mathbf{1}(\widehat{\gamma} > 0) d\widehat{\beta}_E d\widehat{\gamma}
		\end{aligned}
		\end{equation*}
		where
		\begin{equation*}
			J_{\phi}(\widehat\gamma;\mathcal{U}) = \operatorname{det}( \Gamma + \bar{U}^{\intercal}(W_E^{\intercal}W_E)^{-1}\Lambda \bar{U} ).
		\end{equation*}
\end{proposition}

\subsection{Sparse Group LASSO}
\label{sec:sparse-group-lasso}

The sparse Group LASSO \citep{simon2013sparse} produces solutions that are sparse at both the group level and the individual level within selected groups by deploying the Group LASSO penalty along with the usual $\ell_1$ penalty.
A randomized formulation of the sparse Group LASSO is given by
\begin{equation} \label{sglasso}
\operatorname*{argmin}_{\beta}  \left\{\frac{1}{2} \lVert y - X\beta \rVert_2^2 + \textstyle\sum_{g} \lambda_g \lVert \beta_g \rVert_2 + \lambda_0 \lVert \beta \rVert_1 - \omega^{\tp} \beta \right\};
\end{equation}
 the sum over $g$ forms a non-overlapping partition of the predictors.
Notice, this criterion may be viewed as a special case of the overlapping Group LASSO where each predictor $i$ appears in a group $g \in \{1, \ldots ,G\}$, as well as in its own individual group.

Setting up notations, recall that $\mathcal{G}_E$ denotes the set of selected groups and $-\mathcal{G}_E$ its complement; let $T_g$ denote the selected predictors in group $g$, and $-T_g$ the corresponding complement.
Finally, we define
\begin{equation*}
  \breve{E} = \bigcup_{g \in \mathcal{G}_E} T_g, 
\end{equation*}
the set of selected predictors in the selected groups which parameterize our model.
We write the stationary mapping for the sparse Group LASSO below:
\begin{equation} \label{kkt_sgl}
\begin{aligned}
  \omega
  = X^{\intercal} X
    \begin{pmatrix}
      \left( \gamma_g u_g \right)_{g \in \mathcal{G}_E} \\
      0
    \end{pmatrix}
    &- X^{\intercal} X_{\breve{E}} \hat{\beta}_{\breve{E}} - N_{\breve{E}}  +
    \begin{pmatrix}
      \left( \lambda_g u_g \right)_{g \in \mathcal{G}_E} \\
      \left( \lambda_h z_h \right)_{h \in -\mathcal{G}_E}
    \end{pmatrix} + \lambda_0\left( \begin{pmatrix} (s_j)_{j \in T_g} \\ (s_j)_{j \in -T_g} \end{pmatrix}_{g \in \mathcal{G}} \right)
\end{aligned}
\end{equation}
where $\hat{\beta}_{\breve{E}}$ and $N_{\breve{E}}$ are defined as per projections according to the selected model; we denote this mapping by $\breve{\phi}$.
Recall, $\gamma$, $\mathcal{U}$, and $\mathcal{Z}$ are consistent in their definition in terms of the groups we select after solving \eqref{glasso}.
In addition, the subgradient variables from the $\ell_1$ penalty are represented by $s_j$, with $s_j = \text{sign}(\widehat{\beta}^{(\mathcal{G})}_j)$ for $j \in \breve{E}$, and $\lvert s_j \rvert < 1$ for $j \in -\breve{E}$.
We collect these scalar variables into the two sets, $\mathcal{S}_{\mathcal{G}_E}$ and $\mathcal{S}_{-\mathcal{G}_E}$, depending on whether the predictor is in a selected group.
For non-selected predictors in selected groups, the corresponding entry of $u_g$ is zero. That is, we can interpret $u_g$ as a unit vector with the same dimension as the  selected part of group $g$, denoting it by $\breve{u}_g = (u_{g,j} : j \in T_g)^{\tp} \in \mathbb{R}^{\lvert T_g \rvert}$. Set $\breve{U} = \operatorname{diag}\left( \left(\breve{u}_g\right)_{g \in \mathcal{G}_E}\right)$.
Define the selection event
\begin{equation*}
	\mathcal{A}_{\breve{E}} = \{\widehat{\breve{E}}=\breve{E},\ \widehat{\mathcal{U}}=\mathcal{U},\ \widehat{\mathcal{Z}}=\mathcal{Z},\ \widehat{\mathcal{S}}_{\mathcal{G}_E}=\mathcal{S}_{\mathcal{G}_E}, \ \widehat{\mathcal{S}}_{-\mathcal{G}_E}= \mathcal{S}_{-\mathcal{G}_E} \},
\end{equation*}
\begin{equation}\label{A:B:c:spglasso}
{A = - X^{\intercal} X_{\breve{E}}, \quad B= X^{\intercal}X_{\breve{E}} \breve{U},  \quad c = - N_E +  \begin{pmatrix}
		      \left( \lambda_g u_g \right)^{\intercal}_{g \in \mathcal{G}_E} &
		      \left( \lambda_g z_g \right)^{\intercal}_{g \in -\mathcal{G}_E}
		    \end{pmatrix}^{\intercal} + \lambda_0
				\left( \left((s_j)_{j \in g}\right)_{g \in \mathcal{G}} \right)^{\intercal}},
\end{equation}
using the stationary mapping for the learning algorithm under scrutiny.
We then recover the following theoretical expression for the adjustment factor post the sparse Group LASSO.

\begin{proposition}
\label{adj:spglasso}
	For each selected group $g \in \mathcal{G}_E$, construct $\bar{\breve{U}}_g \in \mathbb{R}^{\lvert T_g \rvert \times (\lvert T_g - 1 \rvert)}$ as the orthonormal basis completion of $\breve{u}_g$.
	Define the following matrices
		\begin{equation*}
		       \begin{aligned}
				& \bar{\breve{U}} = \operatorname{diag}\left( \left(\bar{\breve{U}}_g\right)_{g \in \mathcal{G}_E}\right), \Gamma = \operatorname{diag}\left( \left( \gamma_g I_{\lvert T_g \rvert - 1} \right)_{g \in \mathcal{G}_E} \right),  \Lambda = \operatorname{diag}\left( \left(\lambda_g I_{|T_g|}\right)_{g \in \mathcal{G}_E}\right), \\
                 \end{aligned}
		\end{equation*}
		based upon \eqref{A:B:c:spglasso}.
		Then, we have
		\begin{equation*}
		        \begin{aligned}
			\mathbb{P}(\mathcal{A}_E \  | \ \beta_{\breve{E}}) &\propto \int \int \pp(\beta_{\breve{E}}, \Sigma_{\breve{E}} ; \widehat{\beta}_{\breve{E}}) \cdot \exp \left\{ -\frac{1}{2}( A \widehat{\beta}_{\breve{E}} + B\widehat{\gamma}  + c)^{\intercal} \Omega^{-1} (A \widehat{\beta}_{\breve{E}} + B\widehat{\gamma}  + c) \right\} \\
		&\;\;\;\;\;\;\;\;\;\;\;\;\;\;\;\;\;\;\;\;\;\;\;\;\;\;\;\;\;\;\;\;\;\;\;\;\;\;\;\;\;\;\;\;\;\;\;\;\;\;\;\;\;\;\;\;\;\;\;\;\;\;\;\;\; \times J_{\breve{\phi}_{\widehat{\beta}_{\breve{E}}}}(\widehat{\gamma};\mathcal{U}) \cdot \mathbf{1}(\widehat{\gamma} > 0) d\widehat{\beta}_{\breve{E}} d\widehat{\gamma}
		       \end{aligned}
		\end{equation*}
		where
		\begin{equation*}
			J_{\breve{\phi}}(\widehat\gamma;\mathcal{U}) = \operatorname{det}( \Gamma + \bar{\breve{U}}^{\intercal}(X_{\breve{E}}^{\intercal}X_{\breve{E}})^{-1}\Lambda \bar{\breve{U}} ).
		\end{equation*}
\end{proposition}

\section{Large Sample Theory}
\label{sec:large-sample-theory}

In the present section, we establish statistical credibility for our surrogate selection-informed posterior under a fixed $p$ and growing $n$ regime.
We fix $\beta_{n, E}$ to be the sequence of parameters governing our generating model such that
$\sqrt{n} \beta_{n,E} = b_n \bar{\beta}_E$, where $n^{-1/2}b_n = O(1)$, $b_n \rightarrow \infty$ as $n \rightarrow \infty$.
Introducing the dependence on the sample size, let $\ell_{n,E}(\beta_{n,E};  \widehat{\beta}_{n, E} \ \lvert \ N_{n, E})$ represent the surrogate (log) selection-informed likelihood given in Theorem \ref{thm2} after ignoring constants and the prior:
\begin{equation*}\label{log:lik:n}
\begin{aligned}
& \log \pp(\bar{R}\sqrt{n}\beta_{n, E}+ \bar{s}, \bar{\Theta} ; \sqrt{n}\widehat{\beta}_{n, E})  + \dfrac{n}{2}(\gamma_n^{\star}-\bar{P}\beta_{n, E} -n^{-1/2}\bar{q})^{\intercal} (\bar{\Sigma})^{-1} (\gamma_n^{\star} -\bar{P}\beta_{n, E} -n^{-1/2}\bar{q}) \\
&\;\;\;\;\;\;\;\;\;\;\;\;\;\;\;\;\;\;\;\;\;\;\;\;\;\;\;\;\;\;\;\;\;\;\;\;\;\;\;\;\;\;\;\;\; + \barr (\sqrt{n} \gamma^{\star}_n) - \log J_{\phi}(\sqrt{n} \gamma^{\star}_n; \mathcal{U}),
 \end{aligned}
\end{equation*}
based on the optimizer
  \begin{equation*}
  \gamma_n^{\star} = \operatorname*{argmin}_{\gamma}\ \dfrac{n}{2}(\gamma-\bar{P}\beta_{n, E}  - n^{-1/2}\bar{q})^{\intercal} (\bar{\Sigma})^{-1} (\gamma-\bar{P}\beta_{n,E} -n^{-1/2}\bar{q}) + \barr (\sqrt{n} \gamma).
  \end{equation*}
Recall, appending our surrogate selection-informed likelihood to a selection-informed prior $\pi_E(\cdot)$ gives us our selection-informed posterior.
We denote the measure of a set $\mathcal{K}$ with respect to this posterior distribution as follows:
\begin{equation}
\label{surrogate:posterior:prob}
\Pi_{n,E}(\mathcal{K}  \ | \ \widehat{\beta}_{n,E} ; N_{n,E} ) = \dfrac{\int_{\mathcal{K}} \pi_E(z_n) \cdot \exp(\ell_{n,E}(z_n ;  \widehat{\beta}_{n, E}\  \lvert  \ N_{n, E}))\ dz_n}{\int \pi_E(z_n) \cdot \exp(\ell_{n,E}(z_n;  \widehat{\beta}_{n, E}\  \lvert \ N_{n, E})) \ dz_n}.
\end{equation}
We let
	\begin{equation*}
		\mathcal{B}(\beta_{n,E},\delta) = \left\{ z : \lVert z - \beta_{n,E} \rVert_2^2 < \delta \right\} 
	\end{equation*}
denote a ball of radius $\delta$ around our parameter of interest, $\beta_{n,E}$, and $\mathcal{B}^c(\beta_{n,E},\delta)$ denote its complement.
Lastly we use $\mathbb{P}_{n,E} (\cdot)$ to represent the selection-informed probability after conditioning upon our selection event and the ancillary statistic under the generating parameter, $\beta_{n, E}$.

Our main theoretical result, Theorem \ref{thm3}, proves that our surrogate version of the selection-informed posterior concentrates around the true parameter as the sample size grows infinitely large, giving us the rate of contraction.
We begin with two supporting propositions: (i) Proposition \ref{gen_laplace_convergence} proves the convergence of the approximate normalizing constant based on \eqref{sel:post:gen:Laplace} to the exact counterpart when the support constraints for the sizes of the selected groups are restricted to a compact subset; (ii) Proposition \ref{bounds:lik} bounds the curvature of the surrogate (log) selection-informed likelihood around its maximizer. Proofs of these main results are in Appendix \ref{appendix:large-sample-theory}.
To support the claim in Proposition \ref{bounds:lik}, Lemma \ref{lemma_jacobian_derivs} and Lemma \ref{lemma_jacobian_contribution} provide supplementary theory to control the asymptotic orders of the gradient and Hessian of the (log) Jacobian in our surrogate selection-informed posterior. We include both these results in Appendix \ref{appendix:supporting-large-sample-theory}.



\begin{proposition} \label{gen_laplace_convergence}
	Suppose that
	\begin{equation}
	\label{assump2}
		\lim_{n \rightarrow \infty} (b_n)^{-2} \left\{ \log \mathbb{P}\left( \sqrt{n} \gamma_n > 0\right) - \log \mathbb{P} \left( \sqrt{n} \gamma_n > \bar{q}  \right)\right\} = 0.
	\end{equation}
	Let $\bar\gamma_n^{\star}$ be given as
	\begin{equation*}
	 \operatorname*{argmin}_{\bar\gamma < \bar{Q} \cdot 1_{\lvert E \rvert} } \frac{b_n^2}{2} (\bar\gamma - \bar{P}\bar{\beta}_E - (b_n)^{-1}\bar{q})^{\tp} \bar{\Sigma}^{-1} (\bar\gamma - \bar{P}\bar{\beta}_E -(b_n)^{-1} \bar{q}) +  \barr (b_n \bar\gamma),
	\end{equation*}
	where $\bar{Q}>0$.
   Then we have
	\begin{equation*}
	        \begin{aligned}
		& \lim_{n \rightarrow \infty}  (b_n)^{-2} \log \mathbb{P} \left( 0< \sqrt{n} \gamma_n  < b_n \bar{Q} \cdot 1_{\lvert E \rvert}   + \bar{q} \right)+ (b_n)^{-2} \barr (b_n \bar\gamma_n^{\star}) - (b_n)^{-2} \log J_{\phi} (b_n \bar\gamma^{\star}_n ; \mathcal{U}) \\
		&\;\;\;\;\;\;\;\;\;\;\;\;\;\;\;\;\;\;\;\;\;\;\;\;\;\;\;\;\;\;\;\;\;\;\;\;\;\;\;\; \;\;\;\;\; + \dfrac{1}{2} (\bar\gamma_n^{\star} - \bar{P}\bar{\beta}_E - (b_n)^{-1}\bar{q})^{\tp} \bar{\Sigma}^{-1} (\bar\gamma_n^{\star} - \bar{P}\bar{\beta}_E - (b_n)^{-1}\bar{q})= 0.
		\end{aligned}
	\end{equation*}
\end{proposition}

\begin{remark}
For a fixed selection-informed prior $\pi_E(\cdot)$, we may choose $\bar{Q}$ to be a positive constant in order to consider a sufficiently large compact subset of our selection region that would work for all $\bar\beta_{E}$ in a bounded set of probability close to 1 under our prior.
Proposition~\ref{gen_laplace_convergence} now implies that
our surrogate (log) selection-informed likelihood converges to its exact counterpart, obtained by plugging in the exact probability of selection under the parameter sequence $\beta_{n, E}$, as the sample size grows to $\infty$.
\end{remark}

Proposition \ref{gen_laplace_convergence} and the above remark together motivate the following approximation for the likelihood adjustment factor in Proposition \ref{prop:1}:
\begin{equation}
\label{approx:scaled:n}
\begin{aligned}
& \exp\Big(-\dfrac{n}{2} (\gamma_n^{\star} - \bar{P}\bar{\beta}_E - n^{-1/2}\bar{q})^{\tp} \bar{\Sigma}^{-1} (\gamma_n^{\star} - \bar{P}\bar{\beta}_E - n^{-1/2}\bar{q})\\
&\;\;\;\;\;\;\;\;\;\;\;\;\;\;\;\;\;\;\;\;\;\;\;\;\;\;\;\;\;\;\;\;\;\;\;\;\;\;\;\;\;\;- \barr (\sqrt{n} \gamma_n^{\star}) + \log J_{\phi} (\sqrt{n} \gamma^{\star}_n ; \mathcal{U}) \Big),
\end{aligned}
\end{equation}
by substituting $b_n \bar\gamma$ with  $\sqrt{n} \gamma$ in the optimization objective of the above Proposition.




\begin{proposition} \label{bounds:lik}
	Fix $\mathcal{C} \in \mathbb{R}^{\lvert E \rvert}$, a compact set. Define $\widehat{\beta}_{n, E}^{\;\text{max}}$ to be the maximizer of the selection-informed likelihood sequence, $\ell_{n,E}(\cdot;  \widehat{\beta}_{n, E} \ \lvert \ N_{n, E})$. Then there exist positive constants $C_0 \leq C_1$
and $N\in \mathbb{N}$ such that for any $0 < \epsilon_0 < C_0$,
	\begin{equation*}
	       \begin{aligned}
		-\dfrac{n}{2}(C_1 + \epsilon_0) \lVert z_n - \widehat{\beta}_{n,E}^{\;\text{max}} \rVert_2^2 & \leq \ell_{n,E}(z_n ;  \widehat{\beta}_{n, E} \ \lvert \ N_{n, E})- \ell_{n,E}(\widehat{\beta}_{n, E}^{\;\text{max}} ;  \widehat{\beta}_{n, E} \ \lvert \ N_{n, E}) \\
		&\;\;\;\;\;\;\;\;\;\;\;\;\;\;\;\;\;\;\;\;\;\;\;\;\;\;\;\;\;\;\;\;\;\;\;\;\;\;\;\;\;\;\;\;\;\;\;\;\;\;\leq -\dfrac{n}{2}(C_0 - \epsilon_0) \lVert z_n - \widehat{\beta}_{n,E}^{\;\text{max}} \rVert_2^2
		\end{aligned}
	\end{equation*}
	for all $n \geq N$ and $z_n \in \mathcal{C}$.
\end{proposition}


We are now ready to state and prove our main theoretical result on the concentration properties of our selection-informed posterior.


\begin{theorem} \label{thm3}
        Suppose a selection-informed prior $\pi_E(\cdot)$ with compact support $\mathcal{C}$ assigns non-zero probability to $\mathcal{B}(\beta_{n,E},\delta')\subset \mathcal{C}$ for any $\delta' > 0$ such that the inclusion is satisfied. Further, assume for the associated prior measure $\Pi_E(\cdot)$ that
        $$\textstyle\lim_{n\to\infty}\exp(-b_n^2\delta^2 K/2)/\Pi_E(\mathcal{B}(\beta_{n,E},\kappa\delta_n)) = 0 
        $$
         for any $K>0$, $\kappa \in (0,1)$, and $\delta > 0$, where $\delta_n$ is defined by $\sqrt{n}\delta_n = b_n \delta$.
         Then, the following convergence must hold for any $\epsilon>0$:
	\begin{equation*}
		\mathbb{P}_{n,E} \left(  \Pi_{n,E} \left( \mathcal{B}^c(\beta_{n,E},\delta_n) \ | \ \widehat{\beta}_{n,E} ; N_{n,E} \right) \leq \epsilon \right) \rightarrow 1 \text{ as } n \rightarrow \infty. 
	\end{equation*}
\end{theorem}

We visualize in Figure~\ref{fig:posterior-concentration} the concentration theory presented in Theorem~\ref{thm3} for our simple example with a single group of two covariates.

\section{Empirical Investigations}
\label{sec:empirical-analysis}


\subsection{Experimental Design}
\label{sec:canonical-gl-posi}

In all of our experiments with synthetic data, we construct the design $X$ by drawing $n = 500$ rows independently according to $\mathcal{N}_p \left( 0, \Sigma \right)$; $\Sigma$ follows an autoregressive structure with the $(i,j)$-th entry of the covariance matrix $\Sigma_{(i, j)}= 0.2^{ \lvert i -j \rvert}$.
We fix the support and vary the values of $\beta$ according to a variety of schemes described below; in each case, we have a ``Low", ``Medium", and ``High" signal-to-noise ratio (SNR) regime as detailed in Appendix \ref{appendix:empirical-analysis}.
Finally, we draw $Y \sim \mathcal{N}_n \left( X \beta, \sigma^2 I_n \right)$, a Gaussian group-sparse linear model; we fix $\sigma = 3$ in our experiments.
We consider the following settings for our grouped covariates:
\begin{itemize}[leftmargin=*]
\itemsep0em
\item Balanced: In our balanced analysis, we partition (i) $p = 100$ covariates into $25$ disjoint groups each of cardinality four when we solve the canonical Group LASSO and the standardized Group LASSO to learn a group-sparse model; (ii) $p=103$ covariates into $34$ groups of four predictors each, but the last feature of the first group is also the first feature of the second group, and so on when we solve the overlapping Group LASSO.
In the latter case,  the first and last groups each have three features in no other groups, and all other groups have two features in no other groups.
We randomly select three of these candidate groups to be active, and let each coefficient assume a random sign with the same magnitudes that in turn depend on the SNR regime.
\item Heterogeneous:
In the heterogeneous setting, we allow the disjoint groups of covariates to differ in their sizes.
Further, our groups of covariates now display heterogeneity in the signal amplitudes both within and between the active groups.
We have 3 groups with three predictors each, 4 groups with four predictors each, 5 groups with five predictors each, and 5 groups with ten predictors each.
We set one of each of the three-, four-, and five-predictor groups to be active with linearly increasing signal magnitudes and each active coefficient is assigned a random sign.
\end{itemize}

For each realization of the data and each setting under study, we apply three methods: (i) ``Selection-informed'', the selection-informed implementation summarized in Algorithm \ref{alg:updates}; drawing each sample remarkably solves only a $|\mathcal{G}_{E}|$-dimensional optimization problem as can be appreciated by reviewing Step \textbf{(Laplace)} in Algorithm~\ref{alg:updates}; (ii) ``Naive'', the standard inferential tool that
first fits the usual Group LASSO \eqref{glasso} with no randomization to identify the active set $E$, and then fits $\widehat{\beta}_E$ using ordinary least squares restricted to $X_E$ from which we obtain credible intervals ignoring the effects of selection; (iii) ``Split'', the sample splitting method follows the same procedure as  ``Naive" except that this method partitions the data at a prespecified ratio $r$, that is, ``Split" applies the usual Group LASSO to $[rn]$ randomly chosen subsamples without replacement to obtain $E$ and then uses the remaining (holdout) samples to fit a linear model restricted to $E$ for interval estimation.
The nominal level for the interval estimates is set at $90\%$.
Following Algorithm \ref{alg:updates}, a (gradient-based) Langevin sampler takes a noisy step along the gradient of our posterior \eqref{grad:log:posterior} at each draw:
\begin{equation} \label{langevin}
    \beta_E^{(K+1)} = \beta_E^{(K)} + \eta \chi \nabla \log \pi_E(\beta_E^{(k)}\  | \ \widehat{\beta}_E, N_E) + \sqrt{2\eta} \epsilon^{(K)},
\end{equation}
where $\eta > 0$ is a predetermined step size, $\pi_E(\beta_E\  | \ \widehat{\beta}_E, N_E)$ denotes our surrogate selection-informed posterior at the parameter vector $\beta_E$, $\epsilon^{(K)}\sim N(0, \chi)$, our Gaussian proposal \citep{shang2015covariance} and we plug in the expression for the gradient of the surrogate posterior from Theorem \ref{thm2}.
In practice, we set $\eta = 1$ and determine $\chi$ from the inverse of the Hessian of our (negative-log) posterior.
It bears emphasis that this sampler serves as a representative execution of our methods; more generally, other sampling schemes to deliver inference based upon our selection-informed posterior (and its gradient) are clearly possible.
We construct credible intervals from the appropriate quantiles in the posterior sample of the parameter for ``Selection-informed" under a diffuse Gaussian (selection-informed) prior, i.e. $\pi_E= \mathcal{N}_{|E|} \left( 0,  r_0\sigma^2\right)$ and $r_0= 100$.
Our credible intervals for ``Split" and ``Naive" are reported for the same prior.

In our experimental findings, we draw attention to comparisons between ``Split" based on an allocation of $r$ fraction of the samples for selecting a group-sparse linear model and
``Selection-informed'' based on the solution of the randomized Group LASSO \eqref{glasso} with a $p$-dimensional isotropic Gaussian randomization variable independent of our response; that is, $\Omega = \tau^2 \cdot I_p$.
To (approximately) match the amount of information utilized during selection by ``Split" for an honest assessment of inference for the randomized methods, we fix the ratio of randomization variation to the noise level in our response as follows:
\begin{equation}
\label{rand:level}
\frac{\tau^2}{\sigma^2}= \frac{(1-r)}{r},
\end{equation}
where $r$ is the proportion of data reserved by ``Split" for solving the Group LASSO.
The value for randomization variation we set for comparisons is motivated from an asymptotic equivalence between data splitting and a Gaussian randomization scheme proved in \citet[Proposition 4.1]{selective_bayesian}, which notes that a regularized regression objective using $[rn]$ subsamples, under an i.i.d. generative process for the response and covariates, can be formulated as \eqref{glasso} with the randomization covariance $\Omega= \frac{(1-r)}{r}\cdot \mathbb{E}[ \frac{1}{n}X^{\tp} X]$.
In this sense, our choice of $\tau^2$ allows us to mimic data carving during post-selective inference for the group-sparse parameters.
Clearly, there will be a tradeoff between the information used for selection and inference when the same data is utilized for learning a group-sparse model and inferring for these selection-informed parameters.
We remark that ``Naive", deploying all the samples for selection, does not strike a tradeoff between the two intertwined goals.

\subsection{Inferential findings for different Grouped Sparsities}
\label{sec:extended-gl-posi}

Based on the design of experiment in the preceding section, we undertake $100$ rounds of numerical simulations for each level of randomization variation and a category of the three SNR regimes, ``Low", ``Medium" and ``High".
Our experimental findings after solving the canonical Group LASSO \eqref{glasso} are summarized for Balanced and Heterogeneous groups in Figures\footnote{We do not plot outliers in any box plots.} \ref{fig:can-que}, \ref{fig:can-cov}, \ref{fig:can-len}.
On the x-axis of these figures, we vary the level of randomization with decreasing levels from left to right.
The randomization level is determined by the ratio of data $r$ allocated for model selection and reserved for inference in the case of ``Split" and the level of variation for the corresponding Gaussian randomization scheme  $\tau^2$ is set according to \eqref{rand:level}  in the case of ``Selection-informed"; this value is denoted by the label $x:y$ such that $(x+y)^{-1}x = r$ on the x-axis.
Note, because ``Naive" deploys no randomization, we found it instructive to assign it a label ``0" for the randomization level on the x-axis.
We then highlight how our method can be adapted for extensions to the overlapping and standardized Group LASSO through Figures \ref{fig:other-que}, \ref{fig:other-cov}, \ref{fig:other-len}.
To implement the overlapping Group LASSO, we take the approach of \citet{jacobGroupLassoOverlap2009} that duplicates overlapping features to obtain an augmented design matrix $X^{\star}$ with no overlaps.
Because columns are duplicated, the expanded design matrix is rank deficient; we therefore incorporate a (small) ridge term as per the prescription in Section \ref{sec:generalizations}.
After selecting the active set $E^{\star}$, we map back to the set of selected variables $E$ in the original space to define our group-sparse model and perform inference for $\beta_E$.

In terms of our findings, Figures \ref{fig:can-que} and \ref{fig:other-que} first depict an assessment of the model selection accuracy which we measure in terms of the $\text{F}1$ score:
$$\text{F}1 \; \text{score}= \dfrac{\text{True Positives}}{\text{True Positives} + \frac{1}{2}(\text{False Positives} +\text{False Negatives}) }.$$
These plots corroborate the approximate correspondence in the amount of information used for model selection by the two randomized methods and subsequently in the quality of models selected by them.
We note that the selection accuracy for the randomized methods increases with decreased levels of randomization and is bounded above by the ``Naive" selection based on all the data.
Figures \ref{fig:can-cov} and \ref{fig:other-cov} plot the distribution of coverages of the credible intervals for all three methods grouped by the level of randomization.
Consistent with expectations, ``Naive" does not yield honest interval estimates with the shortfall in coverage understandably more severe for the lower SNR regimes.
Evident from the empirical distribution for the coverage of interval estimates,  ``Selection-informed" and ``Split" discard information from model selection to counteract the bias in uncertainty estimation.
Figures \ref{fig:can-len} and \ref{fig:other-len}  anchor our motivation in the paper to borrow residual information from selection in order to construct more efficient inferential procedures than the benchmark offered by sample splitting.
The gains in efficiency for ``Selection-informed" are noticeable from the clear separation in the distributions of the lengths of the interval estimates produced by the randomized approaches at a fixed level of randomization in diverse grouped settings, for example, we may compare the third quartile for the ``Selection-informed" distribution with the first quartile for the corresponding ``Split".
The ``Naive" intervals in the figure for lengths highlight the price paid in terms of efficiency for constructing honest estimates of uncertainty post selection.
We note a marginal loss in inferential accuracy for the interval estimates based on our method as the level of randomization increases.
This is attributed to the Laplace-type approximation we apply to replace the probability of selection with the mode of the associated integrand.
That is, the quality of approximation under a fixed sample size deteriorates fractionally as the concentration of probability mass is more spread out in the case of randomization level ``1:2" with four times the level of randomization variation than ``2:1".
The above intuitive explanation is corroborated by Figure \ref{fig:rel:approx} for the simple, running example in the paper.

\subsection{Application to neuroimaging data}
\label{sec:appl-neur-data}

We apply our method to a subset ($n = 785$) of human neuroimaging data from the Human Connectome Project (HCP) \citep{vanessenWUMinnHumanConnectome2013}, a landmark study undertaken by a consortium involving Washington University, the University of Minnesota, and Oxford University.
The HCP has led to a substantial advancement of human neuroimaging methodology and included the collection of several corpora of data which are available to researchers interested in studying brain function and connectivity.

\begin{figure}[H]
  \centering
  \includegraphics[width=0.9\textwidth]{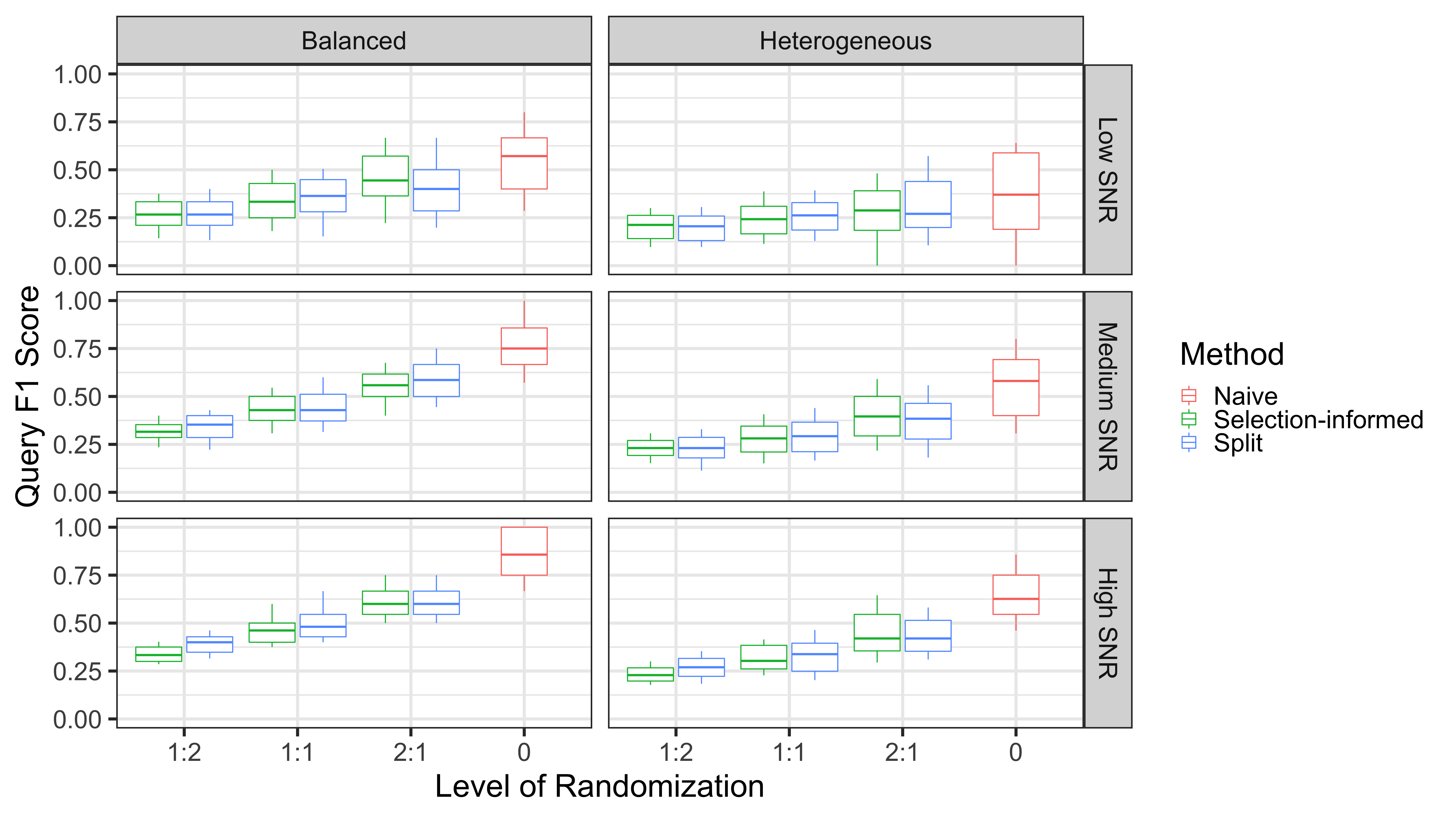}
  \vspace{-0.5cm}
  \caption{Box plots for model selection accuracy under the Group LASSO.}
  \label{fig:can-que}
\end{figure}

\begin{figure}[H]
  \centering
  \includegraphics[width=0.9\textwidth]{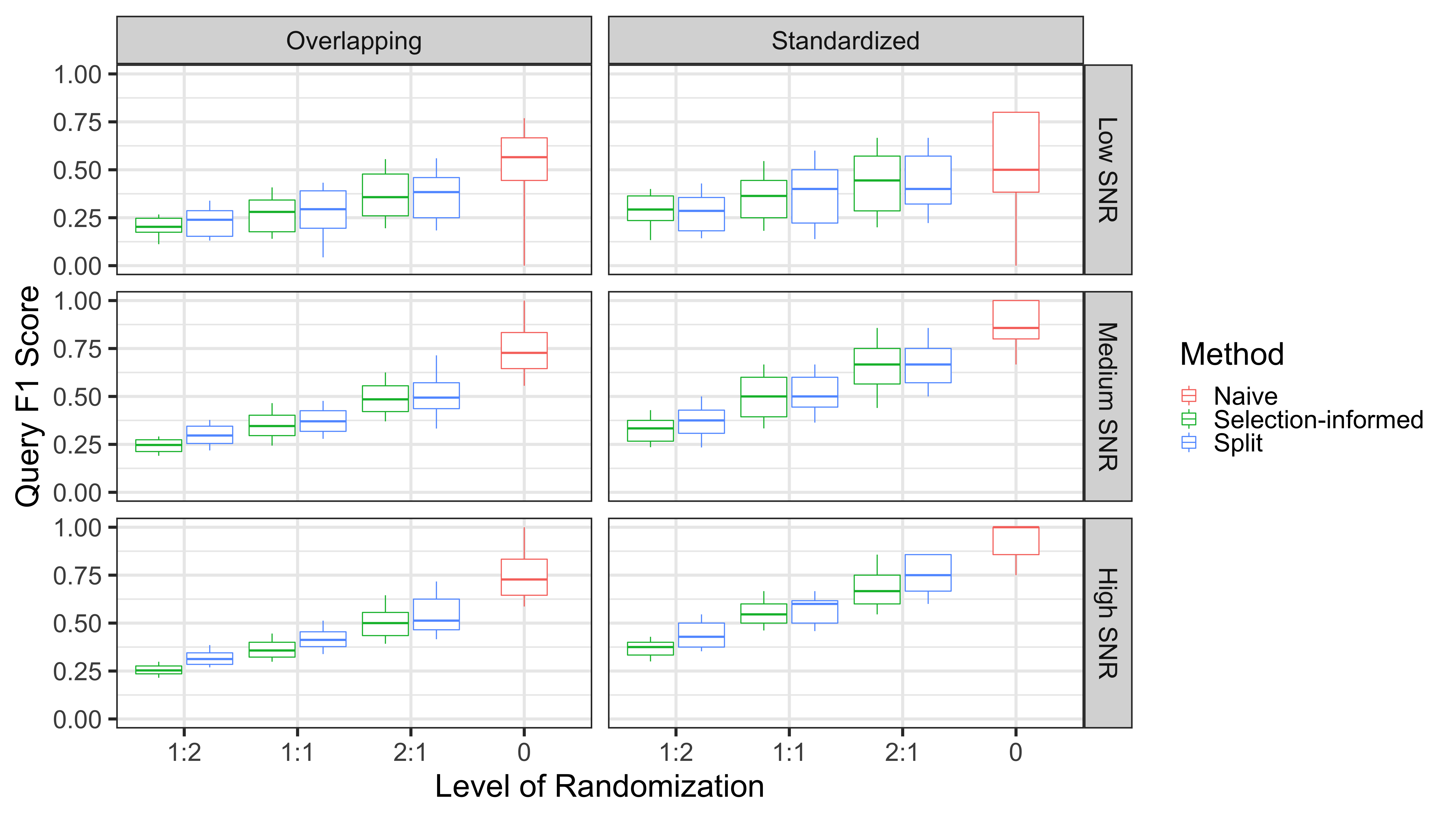}
  \vspace{-0.5cm}
  \caption{Box plots for model selection accuracy under extensions of the Group LASSO.}
  \label{fig:other-que}
\end{figure}

\begin{figure}[H]
  \centering
  \includegraphics[width=0.9\textwidth]{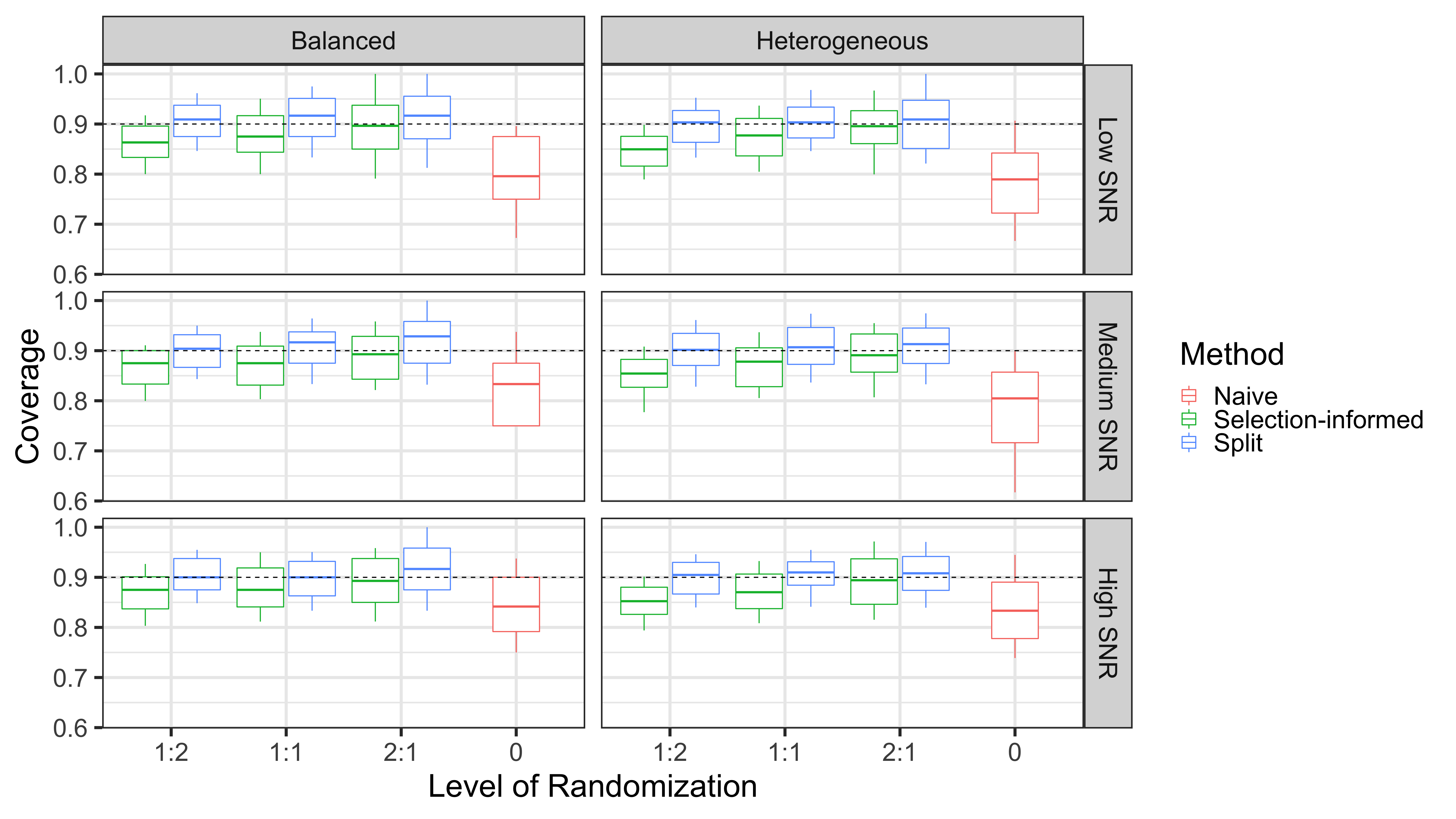}
  \vspace{-0.5cm}
  \caption{Box plots for coverage of credible intervals post the Group LASSO.}
  \label{fig:can-cov}
\end{figure}

\begin{figure}[H]
  \centering
  \includegraphics[width=0.9\textwidth]{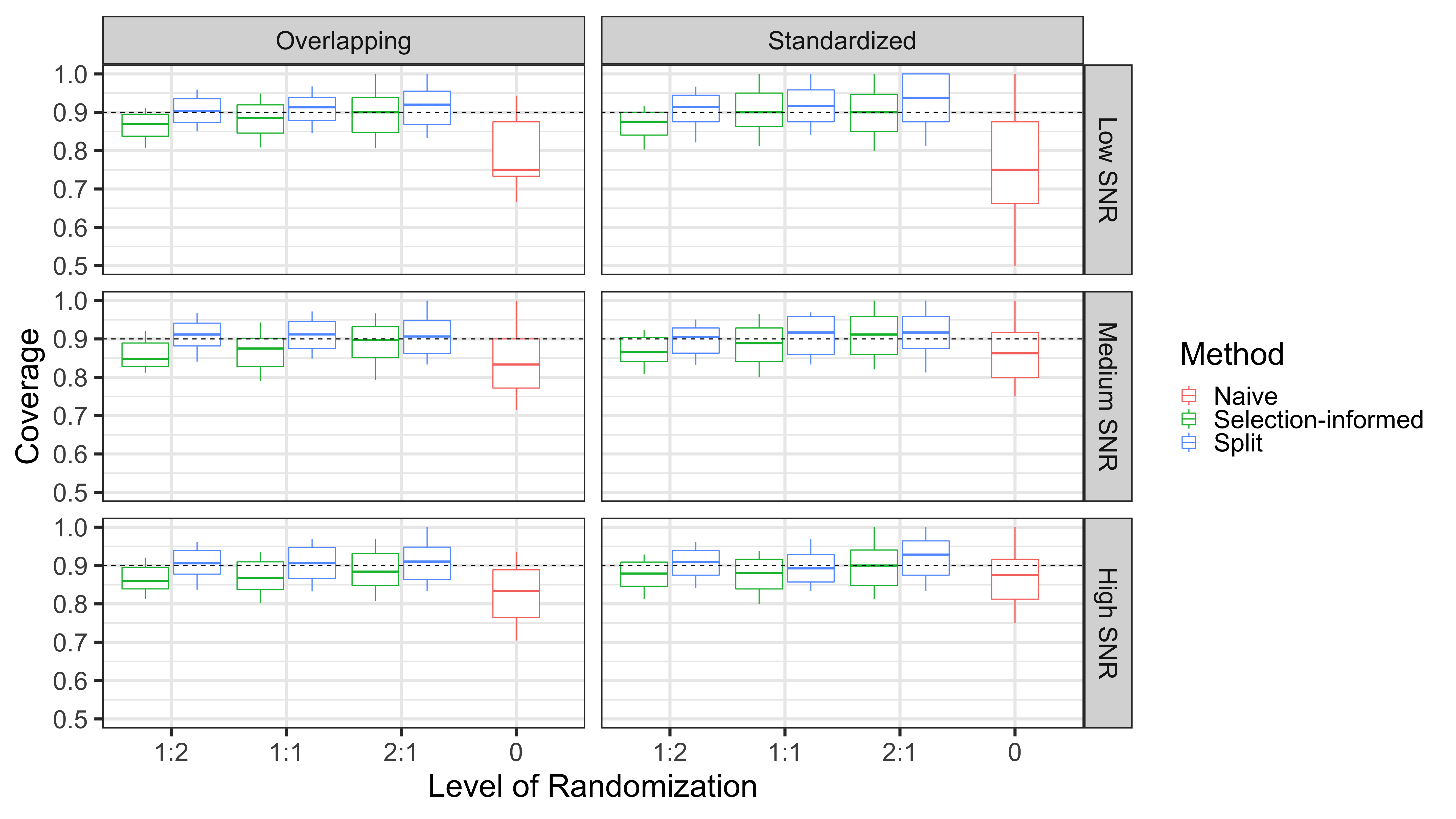}
  \vspace{-0.5cm}
  \caption{Box plots for coverage of credible intervals post the extensions of the Group LASSO.}
  \label{fig:other-cov}
\end{figure}

\begin{figure}[H]
  \centering
  \includegraphics[width=0.9\textwidth]{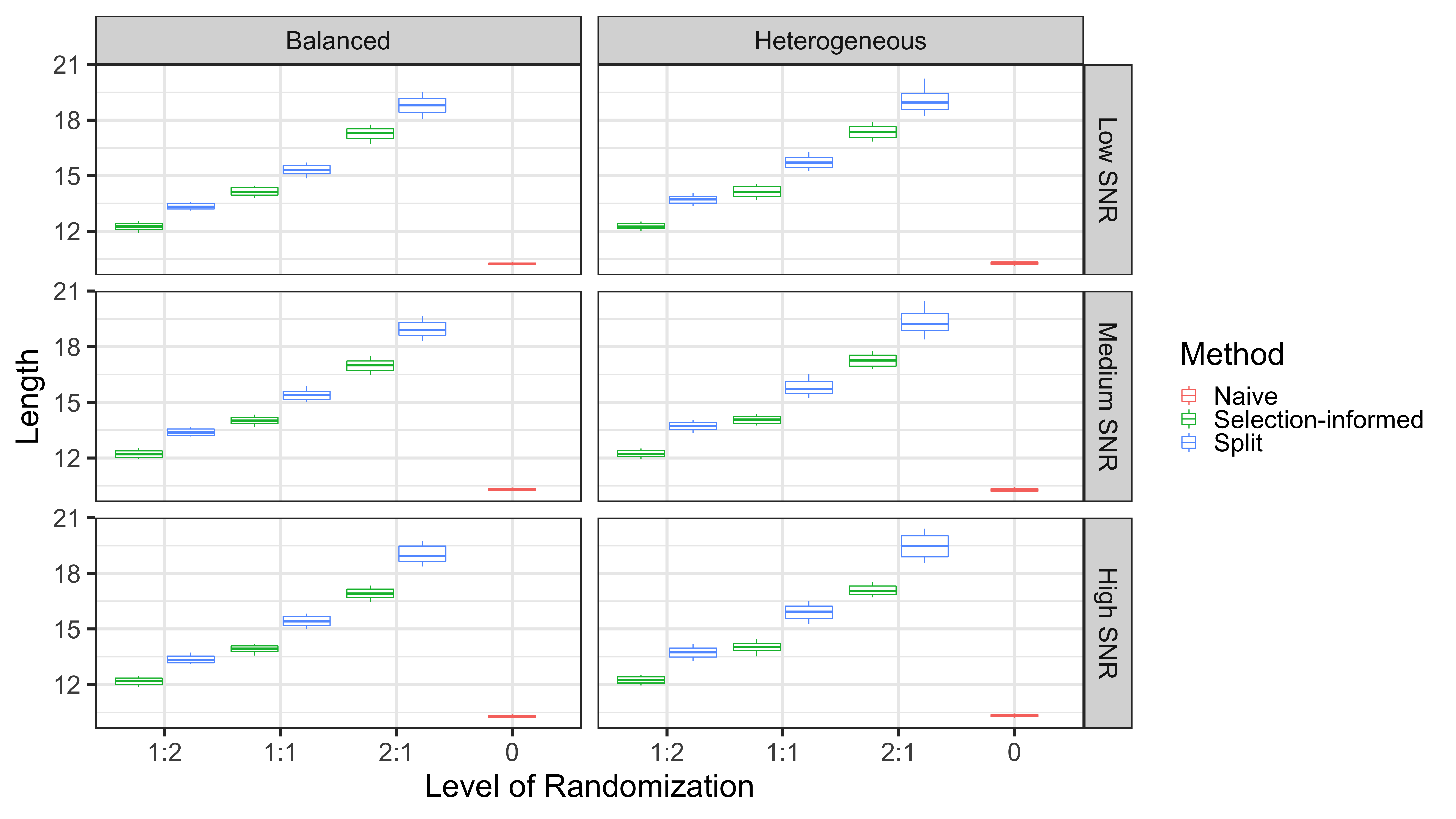}
  \vspace{-0.5cm}
  \caption{
    Box plots for lengths of credible intervals post the canonical Group LASSO.}
  \label{fig:can-len}
\end{figure}

\begin{figure}[H]
  \centering
  \includegraphics[width=0.9\textwidth]{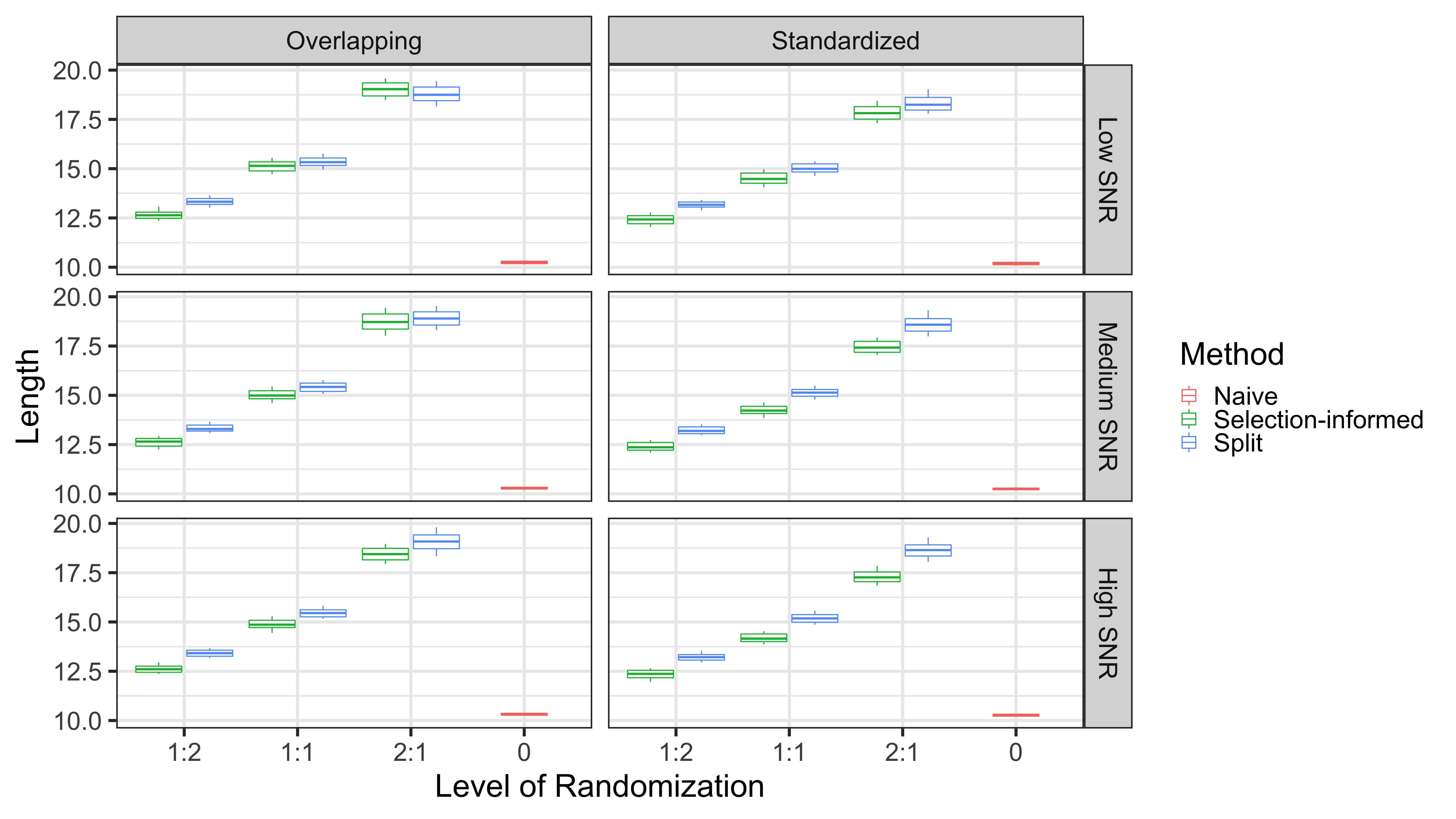}
  \vspace{-0.5cm}
  \caption{
    Box plots for lengths of credible intervals post the extensions of the Group LASSO.}
  \label{fig:other-len}
\end{figure}

We consider below a linear model to understand participant accuracy on a working memory task using brain activity measured at $p = 236$ locations in the brain during performance of the task, using data graciously processed by the lab of our collaborator (see Acknowledgements).
Using both behavioral and functional magnetic resonance imaging (fMRI) measurements recorded from a cognitive task, a standardized measure of accuracy for each participant during this task will be our response $y$ and contrasts relying upon brain activation records during the task form our covariates $X$.
We provide a summary of these details in Appendix \ref{appendix:HCP}, accompanied by a description for the preprocessing steps and parameter settings.
Our analysis here groups the covariates by brain system and applies our selection-informed method to calibrate interval estimates for coefficients within any selected systems.

We apply both the randomized Group LASSO and the Group LASSO to a random split of this dataset.
We consider the level of isotropic Gaussian randomization to be $1:1$, $2:1$, and $9:1$, by setting the variance parameter $\tau^2$ according to \eqref{rand:level} after replacing $\sigma^2$ with $$\hat{\sigma}^2 = \left( n - p \right)^{-1} \left\lVert y - X \left( X^{\intercal} X \right)^{-1} X^{\intercal} y \right\rVert_2^2,$$
for $r=1/2, 2/3, 9/10$, respectively.
In all cases, we select one group: the ``Fronto-parietal Task Control'' (FP) system.
To restore inferential validity, we reuse  data via our ``Selection-informed" method to draw samples from the surrogate selection-adjusted posterior.
Given the overlap between datasets employed, \citet{sripadaTreadmillTestCognition2020} also found that activation in the Fronto-parietal Task Control system could be used to predict ``General Cognitive Ability'' (GCA): this general pattern that included activation in the fronto-parietal system and deactivation in the default mode system under general cognitive demands (including working memory) is discussed and reviewed in \citet{sripadaTreadmillTestCognition2020}.
We depict $90\%$ interval estimates for each of the locations in the brain within the selected FP system under both the ``Selection-informed" and ``Split" methods in Figure \ref{fig:hcp-node-CIs} at varying levels of randomization.
Corroborating the general pattern from our numerical experiments, the intervals from our ``Selection-informed" methods are roughly 8\% shorter than those from ``Split" on an average in the application.

\begin{figure}[h]
  \includegraphics[width=\textwidth]{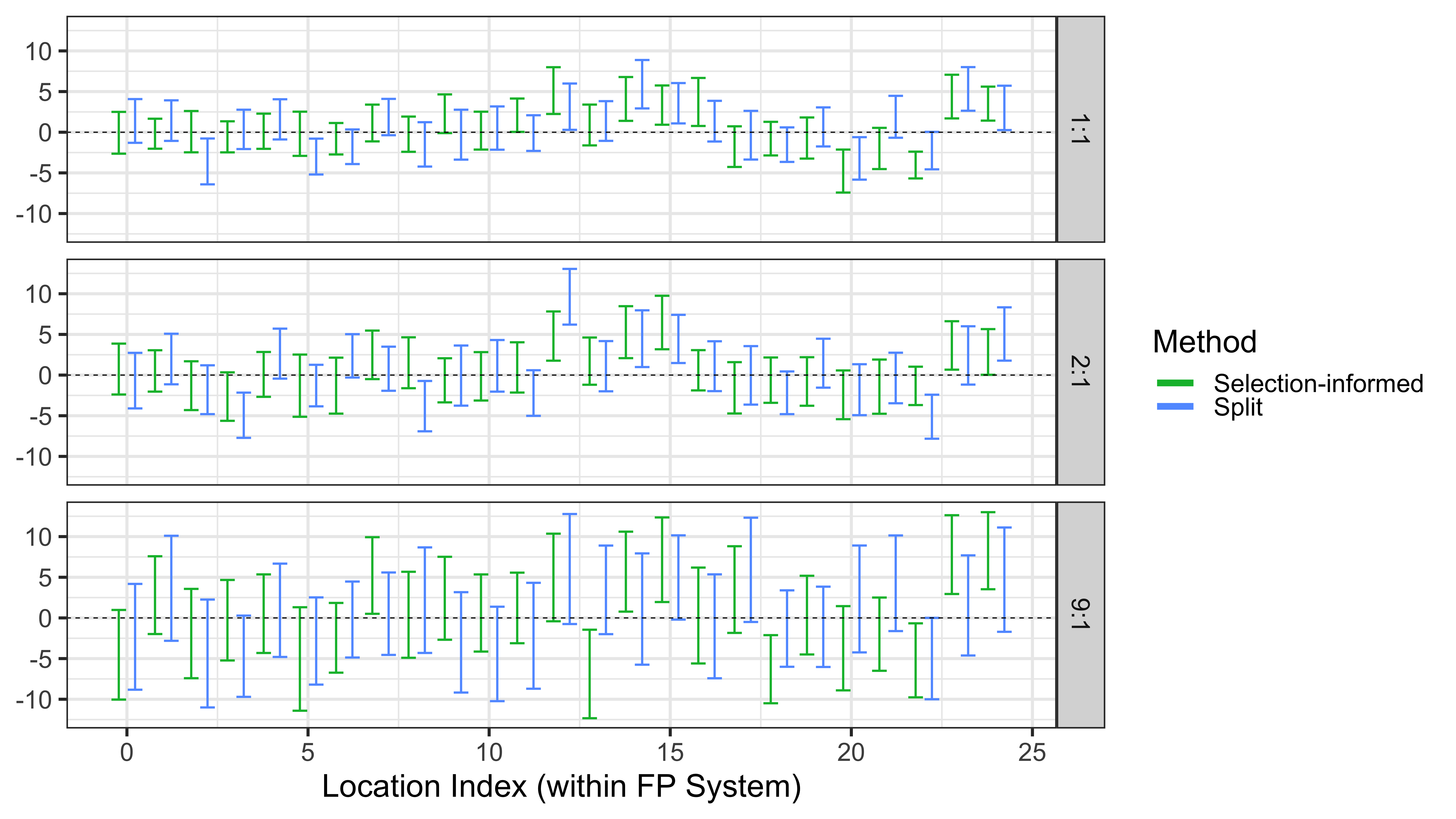}
  \caption{Interval estimates for coefficients at each location within the selected FP brain system for human neuroimaging application.
    Intervals are provided for both our ``Selection-informed" method and ``Split" at a variety of randomization levels.
   At the levels of randomization $1:1$, $2:1$, and $9:1$, average interval widths for $(\text{``Selection-informed"}, \text{``Split"})$ are $(4.72, 5.10), (5.99, 6.17), (10.03, 11.77)$.
  }
  \label{fig:hcp-node-CIs}
\end{figure}

\section{Conclusion}
\label{sec:conclusion}

In this paper, we provide methods to account for the selection-informed nature of models after solving group-sparse learning algorithms.
Deriving conditional inference in these settings is particularly challenging due to a breakdown of a polyhedral representation for the selection event.
Formally cast into a Bayesian framework, we successfully characterize an exact adjustment factor to account for selections of grouped variables and importantly, provide a computationally feasible solution to bridge the gap between theory and practice for a general class of group-sparse models.
Appealingly amenable to a large class of grouped sparsities and context-relevant targets, the efficiency of our methods is evident from the minimal price we pay to correct for selection in comparison to naive inference.

Our work leaves room for promising directions of research which we hope to take on as future investigations.
The performance of our methods serve as an encouraging direction to tackle non-affine geometries in general, seen often with penalties disparate in behavior from $\ell_1$-sparsity imposing algorithms.
These developments do not preclude an asymptotic framework for valid inference after the Group LASSO, when we deviate from Gaussian distributions.
Lastly, grouped selections form the first step of many hierarchical exploratory pipelines (for example, in genomic studies) to select predictors with better meaning and accuracy.
Our solutions for reusing data after a class of grouped selection rules hold the potential to infer using automated models from these complicated selection pipelines.

 \section{Acknowledgments}

S.P. was supported by NSF-DMS 1951980 and NSF-DMS 2113342. D.K.\ was supported by NSF-DMS 1646108 and a Rackham Predoctoral Fellowship from the University of Michigan. P.W.M.\ was supported by NSF-DMS 1916222.
Data were provided in part by the Human Connectome Project, WU-Minn Consortium (Principal Investigators: David Van Essen and Kamil Ugurbil; 1U54MH091657) funded by the 16 NIH Institutes and Centers that support the NIH Blueprint for Neuroscience Research; and by the McDonnell Center for Systems Neuroscience at Washington University. 
We thank Dr.~Chandra Sripada and his research group (with particular thanks to Saige Rutherford) for providing a processed version of this data as well as helpful comments.
This research was supported in part through computational resources and services provided by Advanced Research Computing (ARC), a division of Information and Technology Services (ITS) at the University of Michigan, Ann Arbor.
In addition, this work used the Extreme Science and Engineering Discovery Environment (XSEDE), which is supported by National Science Foundation grant number ACI-1548562.

\bibliographystyle{plainnat}
\bibliography{references}

\appendix

\section{Proofs for technical developments (Section \ref{sec:methodology})}
\label{appendix:methodology}

\begin{proof}[Proof of Theorem~\ref{thm1}.]
        We write our adjustment factor as follows:
        \begin{equation*}
        \begin{aligned}
        & \mathbb{P}(\mathcal{A}_E \  | \ \beta_E)= \int \int \mathrm{p}(\beta_E, \Sigma_E ; \widehat{\beta}_E) \cdot \mathrm{p}(0, \Omega ; \omega) \cdot \mathbf{1}_{\mathcal{A}_{E}}(\widehat{\beta}_E,\omega)\ d\omega  d\widehat{\beta}_E.
        \end{aligned}
        \end{equation*}
        Next, the change of variables in \eqref{CoV} together with the conditioning upon $\widehat{\mathcal{U}}= \mathcal{U}$ and $\widehat{\mathcal{Z}}= \mathcal{Z}$ results in the below simplification
        \begin{equation} \label{adj:factor}
        \begin{aligned}
        &\mathbb{P}(\mathcal{A}_E \  | \  \beta_E) =  \int \int J_{\phi_{\widehat{\beta}_E}}(\widehat\gamma; \mathcal{U}, \mathcal{Z}) \cdot \mathrm{p}(\beta_E, \Sigma_E ; \widehat{\beta}_E) \cdot\mathrm{p}(0, \Omega ; \phi_{\widehat{\beta}_E}(\widehat{\gamma}, \mathcal{U}, \mathcal{Z})) \cdot \mathbf{1}(\widehat{\gamma} > 0) \ d\widehat{\gamma}  d\widehat{\beta}_E\\
        &\propto \int \int J_{\phi_{\widehat{\beta}_E}}(\widehat\gamma; \mathcal{U}, \mathcal{Z}) \cdot \mathrm{p}(\beta_E, \Sigma_E ; \widehat{\beta}_E) \cdot\exp \Big\{ - \frac{1}{2} (A\widehat{\beta}_E + B(\mathcal{U})\widehat{\gamma} +c(\mathcal{U}, \mathcal{Z}))^{\intercal}\Omega^{-1} \\
        & \;\;\;\;\; \;\;\;\;\; \;\;\;\;\; \;\;\;\;\; \;\;\;\;\; \;\;\;\;\; \;\;\; \;\;\;  \;\;\;\;\; \;\;\;\;\; \;\;\;\;\; \;\;\;\;\;\;\;\;\;\; \;\;\;\;\;\;\;\;  (A\widehat{\beta}_E + B(\mathcal{U})\widehat{\gamma} +c(\mathcal{U}, \mathcal{Z})) \Big\}\cdot \mathbf{1}(\widehat{\gamma} > 0) \ d\widehat{\gamma}  d\widehat{\beta}_E
        \end{aligned}
        \end{equation}
        where $J_{\phi_{\widehat{\beta}_E}}(\gamma; \mathcal{U}, \mathcal{Z})$ is the Jacobian associated with the change of variables $\phi_{\widehat{\beta}_E}(\cdot)$.
         To compute the non trivial Jacobian, we write the change of variables map in \eqref{CoV} as follows:
         \begin{equation*}
	  (\widehat{\gamma},\mathcal{U},\mathcal{Z}) \overset{\phi}{\longmapsto} (\phi_1,\phi_2);
	\end{equation*}
         \begin{equation*}
         \begin{aligned}
	  &\phi_1(\widehat{\beta}_E,\widehat{\gamma},\mathcal{U},\mathcal{Z}) = Q U \widehat{\gamma} - Q\widehat{\beta}_E - ((N_{E,j})_{j \in E}) + ((\lambda_g u_g)_{g \in \mathcal{G}_E}),\\
	  &\;\;\;\;\;\;\phi_2(\widehat{\beta}_E,\widehat{\gamma},\mathcal{U},\mathcal{Z}) = X_{-E}^{\intercal}X_E(U\widehat{\gamma} - \widehat{\beta}_E) - ((N_{E,j})_{j \in -E}) + ((\lambda_g z_g)_{g \in -\mathcal{G}_E}).
	 \end{aligned}
	\end{equation*}
	such that $Q = X_E^{\intercal} X_E$.
	Then, the derivative matrix is given by:
	\begin{equation*}
	  D_{\phi_{\widehat{\beta}_E}} =
	  \begin{pmatrix}
	     \pdev{\phi_1}{\mathcal{U}} & \pdev{\phi_1}{\widehat{\gamma}} & \pdev{\phi_1}{\mathcal{Z}} \\
	     \pdev{\phi_2}{\mathcal{U}} & \pdev{\phi_2}{\widehat{\gamma}} & \pdev{\phi_2}{\mathcal{Z}}
	  \end{pmatrix},
	\end{equation*}
         where $\pdev{}{\mathcal{U}}(\cdot)$ refers to differentiation with respect to each $u_g$ in the coordinates of its tangent space.
         Note that the block $\pdev{\phi_1}{\mathcal{Z}}$ above the diagonal is zero and $\det\left(\pdev{\phi_2}{\mathcal{Z}} \right) \propto 1$.
	Thus, it follows that
	\begin{equation*}
	 J_{\phi_{\widehat{\beta}_E}}(\widehat{\gamma};\mathcal{U}, \mathcal{Z})=  \det(D_{\phi_{\widehat{\beta}_E}}) \propto \det \left( \hmat{\pdev{\phi_1}{\mathcal{U}}}{\pdev{\phi_1}{\widehat{\gamma}}} \right).
	\end{equation*}
	First, it is easy to see $\pdev{\phi_1}{\widehat{\gamma}} = QU$.
	For computing the other block $\pdev{\phi_1}{\mathcal{U}}$, let $u_g \in \mathcal{S}^{|g|-1}$ be associated with the tangent space:
	$T_{u_g}\mathcal{S}^{|g|-1} = \{v : v^{\intercal}u_g = 0\}$,
	the orthogonal complement of $\operatorname{span} \{u_g\}$; $\bar{U}_g$ is a fixed orthonormal basis for this tangent space.
	For a vector of coordinates $y_g \in \mathbb{R}^{|g|-1}$ and a general function $h$,
	\begin{equation*}
	  \pdev{h(u_g)}{u_g} := \pdev{h(u_g + \bar{U}_g y_g)}{y_g}.
	\end{equation*}
	Writing this more compactly with a stacked vector $y = (y_1, \ldots ,y_{|\mathcal{G}_E|})^{\intercal}$, it follows that for fixed $g$,
	\begin{equation*}
	\begin{aligned}
	  \pdev{\phi_1}{u_g} &= \pdev{}{y_g} \left\{ Q(U + \bar{U}y)\widehat{\gamma} + ((\lambda_g(u_g + \bar{U}_gy_g))_{g \in E}) \right\} \\
	                     &= Q
	                         \begin{pmatrix}
	                           0 &
	                           \cdots &
	                           (\widehat{\gamma}_g \bar{U}_g)^{\intercal} &
	                           \cdots &
	                           0
	                         \end{pmatrix}^{\intercal} +
	  \begin{pmatrix}
	    0 &
	    \cdots &
	    (\lambda_g \bar{U}_g)^{\intercal} &
	    \cdots &
	    0
	  \end{pmatrix}^{\intercal},
	 \end{aligned}
	\end{equation*}
	and combining these column-wise we obtain the full derivative matrix
	\begin{equation*}
	\begin{aligned}
	  \pdev{\phi_1}{\mathcal{U}} =
	                       \begin{pmatrix}
	                         \pdev{\phi_1}{u_1} & \cdots & \pdev{\phi_1}{u_{|\mathcal{G}_E|}}
	                       \end{pmatrix} &= Q \cdot \operatorname{diag}( (\widehat{\gamma}_g \bar{U}_g)_{g \in E}) + \operatorname{diag}( (\lambda_g \bar{U}_g)_{g \in E} ) \\
	                   &= Q \bar{\Gamma} \bar{U} + \Lambda \bar{U} ; \ \text{ where }\  \bar\Gamma = \operatorname{diag} \left( \left( \widehat{\gamma}_g I_{|g|} \right)_{g \in \mathcal{G}_E} \right),\\
	                   &= (Q\bar{\Gamma} + \Lambda)\bar{U}.
	 \end{aligned}
	\end{equation*}
	This gives us
	\begin{equation}
	\label{alt:exp:det}
	\det( D_{\phi_{\widehat{\beta}_E}}) \propto \det \hmat{(Q\bar{\Gamma} + \Lambda)\bar{U}}{QU}.
	\end{equation}
	 Simplifying this expression further
	\begin{equation*} \label{jacobian}
	\begin{aligned}
	  \det( D_{\phi_{\widehat{\beta}_E}}) &\propto
	                 \det \left( \vmat{\bar{U}^{\intercal}}{U^{\intercal}} \hmat{\bar{\Gamma} \bar{U} + Q^{-1}\Lambda \bar{U}}{U} \right) = \det \left(
	                    \begin{pmatrix}
	                      \Gamma + \bar{U}^{\intercal}Q^{-1}\Lambda \bar{U} & 0 \\
	                      \ldots & I_{|E|}
	                    \end{pmatrix} \right) \nonumber \\
	                &= \det( \Gamma + \bar{U}^{\intercal}Q^{-1}\Lambda \bar{U} ),
	\end{aligned}
	\end{equation*}
      This follows since $\bar{U}^{\intercal} \bar\Gamma \bar{U} = \Gamma$ by block orthogonality, and the final equality is deduced using block triangularity. This proves our claim in the Theorem.
	\end{proof}

\begin{proof}[Proof of Proposition \ref{prop:1}.]
Ignoring the selection-informed prior for now, our likelihood after conditioning upon the event $\mathcal{A}_E$ is given by
\begin{equation*}
\begin{aligned}
& \dfrac{\pp(\beta_E, \Sigma_E ; \widehat{\beta}_E)\cdot \bigintssss J_{\phi}(\widehat{\gamma}, \mathcal{U}) \exp\Big\{-\frac{1}{2}(A\widehat{\beta}_E + B\widehat{\gamma} +c)^{\intercal}\Omega^{-1} (A\widehat{\beta}_E + B\widehat{\gamma} +c)\Big\}\cdot \mathbf{1}(\widehat{\gamma} > 0)d\widehat{\gamma}}{\bigintssss J_{\phi}(\widehat{\gamma}, \mathcal{U})  \cdot \pp(\beta_E, \Sigma_E ; \widehat{\beta}_E)\cdot \exp\Big\{-\frac{1}{2}(A\widehat{\beta}_E + B\widehat{\gamma} +c)^{\intercal}\Omega^{-1} (A\widehat{\beta}_E + B\widehat{\gamma} +c)\Big\}\cdot \mathbf{1}(\widehat{\gamma} > 0) d\widehat{\gamma} d\widehat{\beta}_E}.
 \end{aligned}
\end{equation*}
We note that this expression is proportional to:
\begin{equation*}
\begin{aligned}
&
\left(\int J_{\phi}(\widehat{\gamma}, \mathcal{U})  \mathrm{p}(\bar{R}\beta_E+ \bar{s}, \bar{\Theta} ; \widehat{\beta}_E)\cdot  \mathrm{p}(\bar{A}\widehat{\beta}_E+ \bar{b}, \bar{\Omega} ; \widehat{\gamma})\cdot \mathbf{1}(\widehat{\gamma} > 0) d\widehat{\gamma} d\widehat{\beta}_E\right)^{-1} \\
&\;\;\;\;\;\;\;\;\;\;\;\;\;\;\;\;\;\;\;\;\;\;\;\;\;\;\;\;\;\;\;\;\;\;\;\;\;\;\;\;\;\;\;\;\;\;\;\;\;\;\;\;\;\;\;\;\;\;\;\;\;\;\;\;\;\;\;\;\;\;\;\;\;\;\;\;\;\;\;\;\;\;\;\;\;\;\;\;\;\;\;\;\;\;\times\mathrm{p}(\bar{R}\beta_E+ \bar{s}, \bar{\Theta}; \widehat{\beta}_E)
 \end{aligned}
\end{equation*}
leaving out the constants in $\beta_E$,
where
$$\bar{\Omega} = (B^{\intercal}\Omega^{-1}B)^{-1}, \; \bar{A} = -\bar{\Omega} B^{\intercal} \Omega^{-1} A, \;\bar{b} =  -\bar{\Omega} B^{\intercal} \Omega^{-1} c,$$
$$\bar{\Theta} =\left( \Sigma_E^{-1} - (\bar{A})^{\intercal}  (\bar{\Omega})^{-1} \bar{A} + A^{\intercal} \Omega^{-1} A\right)^{-1} , \ \bar{R} = \bar{\Theta}\Sigma_E^{-1}, \ \bar{s} = \bar{\Theta}\left((\bar{A})^{\intercal}  (\bar{\Omega})^{-1} \bar{b} - A^{\intercal}\Omega^{-1}c\right).$$
This display relies on the observation:
\begin{equation*}
\begin{aligned}
&\pp(\beta_E, \Sigma_E ; \widehat{\beta}_E) \exp\Big\{-\frac{1}{2}(A\widehat{\beta}_E + B\widehat{\gamma} +c)^{\intercal}\Omega^{-1} (A\widehat{\beta}_E + B\widehat{\gamma} +c)\Big\}\\
&\;\;\;\;\;\;\;\;\;\;\;\;\;\;\;\;\;\;\;\;\;\;\;\;\;\;\;\;\;\;\;\;\;\;\;\;\;\;\;\;\;\;\;\;\;\;\;\;\;\;\;\;\;\;\;\;= K(\beta_E) \cdot  \mathrm{p}(\bar{R}\beta_E+ \bar{s}, \bar{\Theta} ; \widehat{\beta}_E)\cdot  \mathrm{p}(\bar{A}\widehat{\beta}_E+ \bar{b}, \bar{\Omega} ; \widehat{\gamma}),
 \end{aligned}
\end{equation*}
such that $K(\beta_E)$ involves $\beta_E$ alone.
\end{proof}

\begin{proof}[Proof of Theorem~\ref{thm2}.]
To derive the expression in \eqref{log:posterior}, we note that optimizing over $\widetilde{\beta}_E$ in the problem:
\begin{equation*}
\begin{aligned}
& \text{minimize}_{\widetilde{\beta}_E, \widetilde{\gamma}}\ \Big\{\dfrac{1}{2}(\widetilde{\beta}_E-\bar{R}\beta_E-\bar{s})^{\intercal} \bar\Theta^{-1} (\widetilde{\beta}_E-\bar{R}\beta_E-\bar{s})   \\
&\;\;\;\;\;\;\;\;\;\;\;\;\;\;\;\;\;\;\;\;\;\;\;\;\;\;\;\;\;\;+ \dfrac{1}{2} (\widetilde{\gamma}-\bar{A}\widetilde{\beta}_E -\bar{b})^{\intercal} (\bar{\Omega})^{-1}(\widetilde{\gamma}-\bar{A}\widetilde{\beta}_E -\bar{b})+ \barr (\widetilde\gamma)\Big\}
\end{aligned}
\end{equation*}
gives us
$$
\text{minimize}_{\gamma}\ \dfrac{1}{2}(\gamma-\bar{P}\beta_E -\bar{q})^{\intercal} (\bar{\Sigma})^{-1} (\gamma-\bar{P}\beta_E -\bar{q}) + \barr (\gamma).
$$
Plugging in the value of this optimization into \eqref{sel:post:gen:Laplace}, the logarithm of the surrogate posterior is given by the expression
\begin{equation*}
\begin{aligned}
\log \pi_E(\beta_E) +  \log \mathrm{p}(\bar{R}\beta_E+ \bar{s}, \bar{\Theta}; \widehat{\beta}_E) + \dfrac{1}{2}(\gamma^{\star}-\bar{P}\beta_E -\bar{q})^{\intercal} (\bar{\Sigma})^{-1} & (\gamma^{\star} -\bar{P}\beta_E -\bar{q})  \\
&+ \barr (\gamma^{\star})- \log J_{\phi}(\gamma^{\star}; \mathcal{U}),
\end{aligned}
\end{equation*}
ignoring additive constants.
To compute the gradient of the (log) surrogate posterior, define
$\zeta^{\star}(\cdot)$ to be the convex conjugate for the function
\begin{equation*} \label{zeta_func}
\zeta(\gamma) = \dfrac{1}{2}\gamma^{\intercal}(\bar{\Sigma})^{-1}\gamma + \barr (\gamma).
\end{equation*}
This allows us to write \eqref{log:posterior} in the below form
\begin{equation} \label{log_posterior_conjugate}
\begin{aligned}
&\log \pi_E(\beta_E) - \dfrac{1}{2}(\widehat{\beta}_E -\bar{R}\beta_E- \bar{s})^{\intercal}(\bar{\Theta})^{-1}(\widehat{\beta}_E -\bar{R}\beta_E- \bar{s}) \\
&\;\;\;\;\;\;\;\;\;\;\;\;\;\;+\dfrac{1}{2}(\bar{P}\beta_E +\bar{q})^{\intercal} (\bar{\Sigma})^{-1} (\bar{P}\beta_E +\bar{q})-\zeta^{\star}((\bar{\Sigma})^{-1} (\bar{P}\beta_E +\bar{q})) - \log J_{\phi}(\gamma^{\star}; \mathcal{U}).
\end{aligned}
\end{equation}
Denoting $L(\beta_E) = (\bar{\Sigma})^{-1} (\bar{P}\beta_E +\bar{q})$ and taking the derivative of \eqref{log_posterior_conjugate} with respect to $\beta_E$ gives us:
\begin{equation*}
\begin{aligned}
& \nabla\log \pi_E(\beta_E) +(\bar{R})^{\intercal}(\bar{\Theta})^{-1}(\widehat{\beta}_E -\bar{R}\beta_E- \bar{s}) + \bar{P}^{\intercal}(\bar{\Sigma})^{-1} (\bar{P}\beta_E +\bar{q})\\
&\Scale[0.95]{-\bar{P}^{\intercal}(\bar{\Sigma})^{-1} \nabla_{L(\beta_E)}\zeta^{\star}((\bar{\Sigma})^{-1} (\bar{P}\beta_E +\bar{q})) -\bar{P}^{\intercal}(\bar{\Sigma})^{-1} \nabla_{L(\beta_E)} \gamma^{\star}((\bar{\Sigma})^{-1} (\bar{P}\beta_E +\bar{q})) \nabla_{\gamma^{\star}} \log J_{\phi}(\gamma^{\star}; \mathcal{U})}\\
&= \nabla\log \pi_E(\beta_E) + (\bar{R})^{\intercal}(\bar{\Theta})^{-1}(\widehat{\beta}_E -\bar{R}\beta_E- \bar{s}) + \bar{P}^{\intercal}(\bar{\Sigma})^{-1} (\bar{P}\beta_E +\bar{q})\\
& -\bar{P}^{\intercal}(\bar{\Sigma})^{-1} (\nabla_{L(\beta_E)}\zeta)^{-1}(L(\beta_E)) -\bar{P}^{\intercal}(\bar{\Sigma})^{-1} \nabla_{L(\beta_E)}\gamma^{\star}(L(\beta_E))\nabla_{\gamma^{\star}} \log J_{\phi}(\gamma^{\star}; \mathcal{U})
\end{aligned}
\end{equation*}
\begin{equation*}
\begin{aligned}
&= \nabla\log \pi_E(\beta_E) + (\bar{R})^{\intercal}(\bar{\Theta})^{-1}(\widehat{\beta}_E -\bar{R}\beta_E- \bar{s})+ \bar{P}^{\intercal}(\bar{\Sigma})^{-1} (\bar{P}\beta_E +\bar{q})\;\;\;\;\;\;\;\;\;\;\;\;\;\;\;\;\;\;\;\;\;\;\;\;\;\;\;\;\;\;\;\;\;\\
&-\bar{P}^{\intercal}(\bar{\Sigma})^{-1} \gamma^{\star}(L(\beta_E)) -\bar{P}^{\intercal}(\bar{\Sigma})^{-1} \nabla_{L(\beta_E)}\gamma^{\star}(L(\beta_E))\nabla_{\gamma^{\star}} \log J_{\phi}(\gamma^{\star}; \mathcal{U}).
\end{aligned}
\end{equation*}
Note, we use the fact: $\nabla\zeta^{\star}(\cdot) = (\nabla\zeta)^{-1}(\cdot)$ in the second display and the third display follows by observing: $(\nabla\zeta)^{-1}(L(\beta_E))= \gamma^{\star}(L(\beta_E))$ where
$$\gamma^{\star} = \operatorname*{argmax}_{\gamma}\ \gamma^{\intercal}L(\beta_E)- \zeta(\gamma),
$$
further equal to the optimizer defined in \eqref{optimizer:main}.
Computing the two pieces in the final term: $\nabla_{L(\beta_E)}\gamma^{\star}(L(\beta_E))$ and $\nabla_{\gamma^{\star}} \log J_{\phi}(\gamma^{\star}; \mathcal{U})$, we first have
\[
\nabla \gamma^{\star}(\cdot) = \nabla^2 \zeta^{\star}(\cdot) = \left[ \nabla^2 \zeta (\gamma^{\star}(\cdot)) \right]^{-1}.
\]
Since $\nabla^2 \zeta(\gamma) = \bar{\Sigma}^{-1} + \nabla^2 \barr(\gamma)$, we have
$$
\nabla_{L(\beta_E)}\gamma^{\star}(L(\beta_E)) =  \left( \bar{\Sigma}^{-1} +  \nabla^2 \barr(\gamma^{\star}(L(\beta_E)) \right)^{-1}. 
$$
As for $\nabla \log J_{\phi}(\gamma^{\star}; \mathcal{U})$, recall from Theorem \ref{thm1}
\begin{equation*}
J_{\phi}(\gamma;\mathcal{U}) = \operatorname{det}( \Gamma + \bar{U}^{\intercal}(X_E^{\intercal}X_E)^{-1}\Lambda \bar{U} ).
\end{equation*}
We have the following matrix derivative identities for a square matrix $X$ \citep{petersen2008matrix}:
\begin{equation} \label{partial_first}
  \frac{\partial \logdet(X)}{\partial X_{ij}} = (X^{-1})_{ji},   \frac{\partial (X^{-1})_{kl}}{\partial X_{ij}} = - (X^{-1})_{ki}(X^{-1})_{jl}.
\end{equation}
Using the chain rule:
\[
  \gamma \mappy{J_1} (\Gamma +  \bar{U}^{\intercal}(X_E^{\intercal}X_E)^{-1}\Lambda \bar{U}) \mappy{J_2} \logdet(\Gamma +  \bar{U}^{\intercal}(X_E^{\intercal}X_E)^{-1}\Lambda \bar{U}),
\]
we partition the indices into $\{M_g\}_{g \in E}$ such that $M_g$ is the set of $|g|-1$ indices along the diagonal of $\Gamma$ corresponding to group $g$.
Then
\begin{equation*} \label{partial_j1}
  \left[ \frac{\partial J_1}{\partial \gamma_g} \right]_{ij} = \begin{cases} 1, \tabby i=j, i \in M_g, \\
    0, \tabby \text{otherwise.} \end{cases}
\end{equation*}
By (\ref{partial_first}), the partial derivatives of $J_2$ are the entries of $(\Gamma + \bar{U}^{\intercal}(X_E^{\intercal}X_E)^{-1}\Lambda \bar{U})^{-1}$.
Putting these together,
\begin{equation*} \label{first_deriv}
  \frac{\partial \log J_{\phi}(\cdot; \mathcal{U})}{\partial \gamma_g} = \sum_{i \in M_g} [(\Gamma +  \bar{U}^{\intercal}(X_E^{\intercal}X_E)^{-1}\Lambda \bar{U})^{-1}]_{ii},
\end{equation*}
which gives us the expression for $\nabla_{\gamma^{\star}} \log J_{\phi}(\gamma^{\star}; \mathcal{U})$.
\end{proof}

We conclude with a remark highlighting the distinction from the usual Laplace-type approximation, which we adopt for tractable calculations of the adjustment factor. An alternate approximation for \eqref{sel:post:gen:Laplace} is given by
\begin{equation*}
\begin{aligned}
& C\exp\Big(-\dfrac{1}{2}(\beta^{\star}_E-\beta_E)^{\intercal} \Sigma_E^{-1} (\beta^{\star}_E-\beta_E) - \dfrac{1}{2} (\gamma^{\star}-A^{\star}\beta^{\star}_E -b^{\star})^{\intercal} (\Sigma^{\star})^{-1} (\gamma^{\star}-A^{\star}\beta^{\star}_E -b^{\star})\\
&\;\;\;\;\;\;\;\;\;\;\;\;\;\;\;\;\;\;\;\;\;\;\;\;\;\;\;\;\;\;\;\;\;\;\;\;\;\;\;\;\;\;\;\;\;\;\;\;\;\;\;\;\;\;\;\;\;\;\;\;\;\;\;\;\;\;\;\;\;\;\;\; - \barr (\gamma^{\star})+ \log J_{\phi}(\gamma^{\star}; \mathcal{U})\Big),
\end{aligned}
\end{equation*}
the usual Laplace approximation
where $\gamma^{\star}$ and $\beta^{\star}_E$ are obtained by solving
\begin{equation}\label{alt:opt:problem}
\begin{aligned}
&\text{minimize}_{\widetilde{\beta}_E, \widetilde{\gamma}}\ \Big\{\dfrac{1}{2}(\widetilde{\beta}_E-\beta_E)^{\intercal} \Sigma_E^{-1} (\widetilde{\beta}_E-\beta_E) + \dfrac{1}{2} (\widetilde{\gamma}-A^{\star}\widetilde{\beta}_E -b^{\star})^{\intercal} (\Sigma^{\star})^{-1} (\widetilde{\gamma}-A^{\star}\widetilde{\beta}_E -b^{\star}) \\
&\;\;\;\;\;\;\;\;\;\;\;\;\;\;\;\;\;\;\;\;\;\;\;\;\;\;\;\;\;\;\;\;\;\;\;\;\;\;\;\;\;\;\;\;\;\;\;\;\;\;\;\;\;\;\;\;\;\;\;\;\;\;\;\;\;\;\;  + \barr (\widetilde\gamma)- \log J_{\phi}(\widetilde{\gamma}; \mathcal{U})\Big\}.
\end{aligned}
\end{equation}
Problem \eqref{alt:opt:problem} deviates from our current formulation \eqref{optimizer:main} in terms of the part the (log) Jacobian term plays in determining the mode of the optimization. We opt specifically for a generalized formulation of the Laplace approximation to compute the Jacobian only once at the mode of \eqref{optimizer:main} for increased computational efficiency, obtaining our selection-informed posterior and the gradient associated with it.

\begin{proof}[Proof of Proposition \ref{adj:oglasso}.]
The proof of this Proposition follows by applying the change of variables map:
$$\omega^*\to (\widehat{\gamma}^*, \widehat{\mathcal{U}}^*, \widehat{\mathcal{Z}}^*) \ \text{ where } \  (\widehat{\gamma}^*, \widehat{\mathcal{U}}^*, \widehat{\mathcal{Z}}^*) =(\phi^*)^{-1}(\omega^*),$$
and conditioning upon: $\widehat{\mathcal{U}}^*= \mathcal{U}^*$ and $\widehat{\mathcal{Z}}^*= \mathcal{Z}^*$.
This leads to the below adjustment factor, the probability of the selection event under consideration,
        \begin{equation} 
        \begin{aligned}
       & \mathbb{P}(\mathcal{A}_{E^*} \  | \  \beta_E) =  \int \int J_{\phi^*}(\widehat\gamma^*; \mathcal{U}^*) \cdot \mathrm{p}(\beta_E, \Sigma_E ; \widehat{\beta}_E) \\
        &\times\exp \Big\{ - \frac{1}{2} (A\widehat{\beta}_E + B(\mathcal{U})\widehat{\gamma}^* +c(\mathcal{U}, \mathcal{Z}))^{\intercal}\Omega^{-1} (A\widehat{\beta}_E + B(\mathcal{U})\widehat{\gamma}^* +c(\mathcal{U}, \mathcal{Z})) \Big\}\cdot \mathbf{1}(\widehat{\gamma}^* > 0) \ d\widehat{\gamma}^*  d\widehat{\beta}_E;
        \end{aligned}
        \end{equation}
        $J_{\phi^*}(\widehat{\gamma}^*  ; \mathcal{U}^* )$ is the Jacobian associated with the change of variables derived from $\phi^*(\cdot)$.
        To complete the proof, we note that the value for $J_{\phi^*}(\widehat{\gamma}^*  ; \mathcal{U}^* )$ is obtained from \eqref{alt:exp:det} in the derivation of the adjustment factor when there are no overlaps in the groups, where we simply replace the original design with the augmented version.
\end{proof}

\begin{proof}[Proof of Proposition \ref{adj:stdglasso}.]
The proof is direct from using the change of variables map from inverting the stationary mapping we identify for the solver \eqref{std:glasso}.
Based upon the matrices we identify in \eqref{A:B:c:stdglasso}, the argument follows similar lines as Theorem \ref{thm1} yielding us the expression for the adjustment factor.
\end{proof}

\begin{proof}[Proof of Proposition \ref{adj:spglasso}.]
Modifying the proof of Theorem \ref{thm1} by replacing the stationary mapping with $\breve{\phi}(\cdot)$ defined in \eqref{kkt_sgl} results in the claim in this Proposition.
We thus omit further details of the proof here.
\end{proof}

\section{Proofs for large sample theory (Section \ref{sec:large-sample-theory}) }
\label{appendix:large-sample-theory}

We provide in this section the proofs of the large sample claims for our selection-informed posterior in Section \ref{sec:large-sample-theory}.
Supporting results for this theory, Lemma \ref{lemma_jacobian_derivs} and \ref{lemma_jacobian_contribution}, are included in Section \ref{appendix:supporting-large-sample-theory}.

\begin{proof}[Proof of Proposition~\ref{gen_laplace_convergence}.]
First, observe that we write our probability of selection as follows:
\begin{equation*}
	        \begin{aligned}
(b_n)^{-2} \log \mathbb{P} \left( 0< \sqrt{n} \gamma_n < b_n \bar{Q} \cdot 1_{\lvert E \rvert} + \bar{q}\right) &= (b_n)^{-2} \log \mathbb{E}\Big[\exp(\log J_{\phi} (\sqrt{n} \bar{Z}_n +b_n \bar{P} \bar{\beta}_E +  \bar{q};\mathcal{U})) \\
& \;\;\times \mathbf{1}( - b_n\bar{P} \bar{\beta}_E-\bar{q} < \sqrt{n} \bar{Z}_n<  b_n \bar{Q} \cdot 1_{\lvert E \rvert}- b_n\bar{P} \bar{\beta}_E )  \Big],
\end{aligned}
	\end{equation*}
(up to an additive constant), where $\sqrt{n} \bar{Z}_n$ is a centered gaussian random variable with covariance $\bar{\Sigma}$.
Define $\mathcal{C}_0 = \{ z: - \bar{P} \bar{\beta}_E < z < \bar{Q} \cdot 1_{\lvert E \rvert}  - \bar{P} \bar{\beta}_E\}$.
Using assumption \eqref{assump2}, we deduce
\begin{equation*}
	        \begin{aligned}
	        &\limsup\limits_{n\rightarrow\infty} \; (b_n)^{-2} \log \mathbb{P} \left( 0< \sqrt{n} \gamma_n < b_n \bar{Q} \cdot 1_{\lvert E \rvert} + \bar{q}\right)\\
	        &= \limsup\limits_{n\rightarrow\infty} \;  (b_n)^{-2} \log \mathbb{E}\Big[\exp(\log J_{\phi} (\sqrt{n} \bar{Z}_n +b_n \bar{P} \bar{\beta}_E +  \bar{q};\mathcal{U})) \cdot 1_{\mathcal{C}_0} (\sqrt{n} \bar{Z}_n/b_n)  \Big]\\
                &\leq \lim_{n\rightarrow\infty} \sup\textstyle_{z \in \mathcal{C}_0}\; (b_n)^{-2} |\log J_{\phi} (b_n z +b_n \bar{P} \bar{\beta}_E +  \bar{q};\mathcal{U}))| + \lim_{n \rightarrow \infty} (b_n)^{-2} \log \mathbb{P} \Big(\sqrt{n} \bar{Z}_n/b_n\in \mathcal{C}_0\Big)\\
                &= \lim_{n \rightarrow \infty} (b_n)^{-2} \log \mathbb{P} \Big( - \bar{P} \bar{\beta}_E <\sqrt{n} \bar{Z}_n/b_n< \bar{Q} \cdot 1_{\lvert E \rvert}  - \bar{P} \bar{\beta}_E\Big). \\
                	\end{aligned}
\end{equation*}

To justify that the limit of the term involving the Jacobian vanishes, note that for all $z \in \mathcal{C}_0$, we have
\begin{equation} \nonumber
	\bar{q} < b_n z +b_n \bar{P} \bar{\beta}_E +  \bar{q} < b_n \bar{Q}1_{\lvert E \rvert} + \bar{q}.
\end{equation}
Thus, $J_{\phi}(b_n z +b_n \bar{P} \bar{\beta}_E +  \bar{q};\mathcal{U})$ is the determinant of a matrix with entries uniformly bounded by $2 b_n \bar{Q} > 0$ for sufficiently large $n$. Then
\begin{equation} \nonumber
	\sup\textstyle_{z \in \mathcal{C}_0} \lvert J_{\phi}(b_n z +b_n \bar{P} \bar{\beta}_E +  \bar{q};\mathcal{U}) \rvert \leq \lvert E \rvert ! \left( 2 b_n \bar{Q} \right)^{\lvert E \rvert}.
\end{equation}
The Jacobian is also bounded away from zero, thus it follows that
\begin{equation} \nonumber
	\sup\textstyle_{z \in \mathcal{C}_0}\; (b_n)^{-2} |\log J_{\phi} (b_n z +b_n \bar{P} \bar{\beta}_E +  \bar{q};\mathcal{U}))| \leq (b_n)^{-2} \log \left( \lvert E \rvert ! \left( 2 b_n \bar{Q} \right)^{\lvert E \rvert} \right)
\end{equation}
which goes to 0 as $n \rightarrow \infty$. Using an argument along the same line,
\begin{equation*}
	        \begin{aligned}
	        &\liminf\limits_{n\rightarrow\infty} \;(b_n)^{-2} \log \mathbb{P} \left( 0< \sqrt{n} \gamma_n < b_n \bar{Q} \cdot 1_{\lvert E \rvert} + \bar{q}\right)\\
	        &\geq -\lim_{n\rightarrow\infty} \sup\textstyle_{z \in \mathcal{C}_0}\; (b_n)^{-2} |\log J_{\phi} (b_n z +b_n \bar{P} \bar{\beta}_E +  \bar{q};\mathcal{U}))| + \lim_{n \rightarrow \infty} (b_n)^{-2}  \log \mathbb{P} \Big(\sqrt{n} \bar{Z}_n/b_n\in \mathcal{C}_0\Big).
	        \end{aligned}
\end{equation*}
From the above limits, we have
$$ \lim_{n\rightarrow\infty}   (b_n)^{-2}\Big(  \log \mathbb{P} \left( 0< \sqrt{n} \gamma_n < b_n \bar{Q} \cdot 1_{\lvert E \rvert} + \bar{q}\right)-   \log \mathbb{P} \left(\sqrt{n} \bar{Z}_n/b_n\in \mathcal{C}_0\right)\Big)=0.$$
Using a moderate (large)-deviation type result \citep{de1992moderate} for the limiting value of the probability in the second term,
we have
\begin{equation*}
	        \begin{aligned}
	        &  \lim_{n\rightarrow\infty}  \; (b_n)^{-2} \log \mathbb{P} \left( 0< \sqrt{n} \gamma_n < b_n \bar{Q} \cdot 1_{\lvert E \rvert} + \bar{q}\right) +\textstyle\inf_{z\in \mathcal{C}_0 } z^{\tp} \bar{\Sigma}^{-1}z/2 =0.
	        	\end{aligned}
\end{equation*}
Lastly, we let $z+ \bar{P} \bar{\beta}_E + (b_n)^{-1} \bar{q} = \bar\gamma$.
Using the observation that the optimization $\inf_{z\in \mathcal{C}_0 } z^{\tp} \bar{\Sigma}^{-1}z$ has a unique minimum, and relying on the convexity of the objectives in the sequence of optimization problems defined below
$$\textstyle\inf_{\bar\gamma < \bar{Q} \cdot 1_{\lvert E \rvert} }  \Big\{\dfrac{1}{2} (\bar\gamma - \bar{P}\bar{\beta}_E - (b_n)^{-1}\bar{q})^{\tp} \bar{\Sigma}^{-1} (\bar\gamma - \bar{P}\bar{\beta}_E - (b_n)^{-1}\bar{q}) + (b_n)^{-2} \barr (b_n \bar\gamma)\Big\},$$
our claim in the Proposition is complete.
\end{proof}

Notice, we work with the below approximation for the adjustment factor in Theorem \ref{thm1}:
\begin{equation*}
\begin{aligned}
& \exp\Big(-\dfrac{b_n^2}{2} (\bar\gamma_n^{\star} - \bar{P}\bar{\beta}_E - (b_n)^{-1}\bar{q})^{\tp} \bar{\Sigma}^{-1} (\bar\gamma_n^{\star} - \bar{P}\bar{\beta}_E - (b_n)^{-1}\bar{q})\\
&\;\;\;\;\;\;\;\;\;\;\;\;\;\;\;\;\;\;\;\;\;\;\;\;\;\;\;\;\;\;\;\;\;\;\;\;\;\;\;\;\;\;- \barr (b_n \bar\gamma^{\star}) + \log J_{\phi} (b_n \bar\gamma^{\star}_n ; \mathcal{U}) \Big),
\end{aligned}
\end{equation*}
motivated by Proposition \ref{gen_laplace_convergence} where
\begin{equation}
	\label{optimizer:scaled}
	 \bar\gamma_n^{\star} = \operatorname{argmin} \frac{1}{2} (\bar\gamma - \bar{P}\bar{\beta}_E - (b_n)^{-1}\bar{q})^{\tp} \bar{\Sigma}^{-1} (\bar\gamma - \bar{P}\bar{\beta}_E - (b_n)^{-1}\bar{q}) + (b_n)^{-2} \barr (b_n \bar\gamma).
	\end{equation}

\begin{proof}[Proof of Proposition~\ref{bounds:lik}.]
Set $\sqrt{n} z_{n} = b_n \bar{z}$ and let $\bar\gamma^{\star}_n$ be the optimizer defined in \eqref{optimizer:scaled}. Then, the surrogate selection-informed (log) likelihood assumes the below form
\begin{equation}
\label{scaled:form:loglik}
\ell_{n,E}(z_{n} ;  \widehat{\beta}_{n, E} \ \lvert \ N_{n, E})= (\sqrt{n} \widehat{\beta}_{n, E})^{\intercal} \bar{\Theta}^{-1} \bar{R}( b_n\bar{z}) -b_n^2 C_n(\bar{z}).
\end{equation}
In the above representation, $C_n(\bar{z})$ equals
\begin{equation*}
\begin{aligned}
 &(\bar\gamma_n^{\star})^{\tp} (\bar{\Sigma})^{-1}(\bar{P}\bar{z} + (b_n)^{-1}\bar{q})-\dfrac{1}{2}(\bar\gamma_n^{\star})^{\tp}(\bar{\Sigma})^{-1} \bar\gamma_n^{\star} -(b_n)^{-2}\barr (b_n \bar\gamma_n^{\star}) \\
 &+ (b_n)^{-2} \log J_{\phi} (b_n \bar\gamma^{\star}_n ; \mathcal{U})+ \dfrac{1}{2}\bar{z}^{\tp} \bar{R}^{\intercal}(\bar{\Theta} + \bar{A}^{\intercal}\bar{\Omega}^{-1} \bar{A})^{-1} \bar{R} \bar{z} - (b_n)^{-1}\bar{z}^{\tp} \bar{P}^{\tp}\bar{\Sigma}^{-1} \bar{q} - \frac{1}{2}(b_n)^{-2} \bar{q}^{\intercal}\bar{\Sigma}^{-1}  \bar{q} ,
\end{aligned}
\end{equation*}
which we derive after plugging in the associated (log) approximation.

Next, we define the below constants: $C_1$ is the largest eigenvalue of $ \bar{R}^{\intercal}\bar{\Theta}^{-1} \bar{R}$ and $C_0$ is the smallest eigenvalue of $ \bar{R}^{\intercal}(\bar{\Theta} + \bar{A}^{\intercal}\bar{\Omega}^{-1} \bar{A})^{-1} \bar{R}$. Consistent with our parameterization, we denote $b_n \bar{\beta}^{\;\text{max}}_E= \sqrt{n}  \widehat{\beta}_{n, E}^{\;\text{max}}$.
It follows then from a Taylor series expansion of $C_n(\bar{z})$ around $\bar{\beta}^{\;\text{max}}_E$ that the difference of log-likelihoods
$$\ell_{n,E}(z_n ;  \widehat{\beta}_{n, E} \ \lvert \ N_{n, E})- \ell_{n,E}(\widehat{\beta}_{n, E}^{\;\text{max}} ;  \widehat{\beta}_{n, E} \ \lvert \ N_{n, E})$$
equals
\begin{equation*}
\begin{aligned}
& \sqrt{n} (\widehat{\beta}_{n, E})^{\tp} \bar{\Theta}^{-1} \bar{R} b_n(\bar{z}- \bar{\beta}^{\;\text{max}}_E)- b_n^2 \left\{C_n(\bar{z})- C_n(\bar{\beta}^{\;\text{max}}_E)\right\}\\
&= b_n (\bar{z}- \bar{\beta}^{\;\text{max}}_E)^{\tp} \bar{R}^{\intercal}\bar{\Theta}^{-1} \sqrt{n} \widehat{\beta}_{n, E}- b_n^2(\bar{z}- \bar{\beta}^{\;\text{max}}_E)^{\tp} \nabla  C_n(\bar{\beta}^{\;\text{max}}_E)\\
&\;\;\;\;\;\;\;\;\;\;\;\;\;\; - \frac{b_n^2}{2}(\bar{z}- \bar{\beta}^{\;\text{max}}_E)^{\tp} \nabla^2  C_n(R(\bar{\beta}^{\;\text{max}}_E; \bar{z})) (\bar{z}- \bar{\beta}^{\;\text{max}}_E)\\
&= - \frac{n}{2} (z_{n} - \widehat{\beta}_{n, E}^{\;\text{max}})^{\tp} \nabla^2  C_n(R(\bar{\beta}^{\;\text{max}}_E; \bar{z})) (z_{n} - \widehat{\beta}_{n, E}^{\;\text{max}}).
\end{aligned}
\end{equation*}
\smallskip

By Lemma~\ref{lemma_jacobian_contribution}, there exists $N \in \mathbb{N}$ such that the contribution of the Jacobian term towards the Hessian ($\nabla^2  C_n(\cdot)$),
$$\textstyle\sup_{\bar{z} \in \mathcal{C}} \lVert \nabla_{\bar{z}}^2 (b_n)^{-2} \log J_{\phi} (b_n \bar\gamma^*_n(\bar{z}) ; \mathcal{U}) \rVert_{\text{op}}$$
 is uniformly bounded in operator norm by $\epsilon_0$ for all $n \geq N$.
 Together with the observation that
$$(\bar\gamma_n^{\star})^{\tp} (\bar{\Sigma})^{-1}(\bar{P}\bar{z} + (b_n)^{-1}\bar{q})-\dfrac{1}{2}(\bar\gamma_n^{\star})^{\tp}(\bar{\Sigma})^{-1} \bar\gamma_n^{\star} -(b_n)^{-2}\barr (b_n \bar\gamma_n^{\star})$$
is a convex conjugate of the function $\frac{1}{2}(\bar\gamma)^{\tp}(\bar{\Sigma})^{-1} \bar\gamma + (b_n)^{-2}\barr (b_n \bar\gamma)$ evaluated at $\bar{\Sigma}^{-1} (\bar{P}\bar{z} + (b_n)^{-1}\bar{q})$,
we conclude
$$(C_0 - \epsilon_0) \cdot I \preccurlyeq \nabla^2C_n(\bar{z}) \preccurlyeq (C_1 + \epsilon_0) \cdot I$$
for all $\bar{z} \in \mathcal{C}$, where $I$ is the identity matrix of appropriate dimensions. This directly leads to our claim in the Proposition.
\end{proof}

\begin{proof}[Proof of Theorem~\ref{thm3}.]
Fix $0<a<1$ such that
$$ 4a^2\cdot (C_1+C_0/2) -(1-a)^2\cdot C_0/2<0,$$
where $C_0$ and $C_1$ are defined in Proposition~\ref{bounds:lik}. This follows by noting that the quadratic expression on the left-hand side has a root between $(0,1)$.
Denoting $\mathcal{C} \cap \mathcal{B}^c(\beta_{n,E},\delta_n) = \mathcal{B'}^c(\beta_{n,E},\delta_n)$, we observe that there exists $N$ such that for all $n\geq N$ such that
\begin{equation*}
\begin{aligned}
& \mathbb{P}_{n,E} \left(  \Pi_{n,E} \left( \mathcal{B}^c(\beta_{n,E},\delta_n) \ | \ \widehat{\beta}_{n,E} ; N_{n,E} \right) \leq \epsilon \right)\\
&\Scale[0.97]{=  \mathbb{P}_{n,E} \left( \int_{\mathcal{B'}^c(\beta_{n,E},\delta_n)} \pi_E(z_n) \cdot \exp(\ell_{n,E}(z_n ;  \widehat{\beta}_{n, E} \lvert N_{n, E}))\ dz_n \leq \epsilon \int \pi_E(z_n) \cdot \exp(\ell_{n,E}(z_n;  \widehat{\beta}_{n, E} \lvert N_{n, E})) \ dz_n\right)}\\
&\geq \mathbb{P}_{n,E} \Bigg( \int_{\mathcal{B'}^c(\beta_{n,E},\delta_n)} \pi_E(z_n) \cdot \exp(\ell_{n,E}(z_n ;  \widehat{\beta}_{n, E} \lvert N_{n, E})- \ell_{n,E}(\widehat{\beta}_{n, E}^{\;\text{max}} ;  \widehat{\beta}_{n, E} \lvert N_{n, E}))\ dz_n \\
& \;\;\;\;\;\;\;\;\;\;\;\;\; \leq \epsilon \int_{\mathcal{B}(\beta_{n,E},a \delta_n)} \pi_E(z_n) \cdot \exp(\ell_{n,E}(z_n;  \widehat{\beta}_{n, E} \lvert N_{n, E})- \ell_{n,E}(\widehat{\beta}_{n, E}^{\;\text{max}} ;  \widehat{\beta}_{n, E} \lvert N_{n, E})) \ dz_n\Bigg)\\
&\geq \mathbb{P}_{n,E} \Bigg( \int_{\mathcal{B'}^c(\beta_{n,E},\delta_n)} \pi_E(z_n) \cdot \exp(-nC_0\cdot   \|\widehat{\beta}_{n, E}^{\;\text{max}} -z_n\|^2/4)dz_n \\
& \;\;\;\;\;\;\;\;\;\;\;\;\; \;\;\;\;\;\;\;\;\;\;\;\;\; \;\;\;\;\;\;\;\;\;\;\;\;\;\leq \epsilon \int_{\mathcal{B}(\beta_{n,E},a \delta_n)} \pi_E(z_n) \cdot \exp(-n(C_1+C_0/2)\cdot  \|\widehat{\beta}_{n, E}^{\;\text{max}} -z_n\|^2/2)dz_n\Bigg).
\end{aligned}
\end{equation*}
The ultimate display follows by using the bounds in Proposition \ref{bounds:lik} where we set $\epsilon_0 = C_0/2$. This yields us the bound
\begin{equation*}
\begin{aligned}
& \mathbb{P}_{n,E} \left(  \Pi_{n,E} \left( \mathcal{B}^c(\beta_{n,E},\delta_n) \ | \ \widehat{\beta}_{n,E} ; N_{n,E} \right) \leq \epsilon \right)\\
&\geq \mathbb{P}_{n,E} \Bigg( \int_{\mathcal{B'}^c(\beta_{n,E},\delta_n)} \pi_E(z_n) \cdot \exp(-nC_0\cdot  \|\widehat{\beta}_{n, E}^{\;\text{max}} -z_n\|^2/4)dz_n \\
& \;\;\;\;\;\;\;\;\;\;\;\;\;\;\;\; \;\;\;\;\;\;\;\;\;\;\;\;  \leq \epsilon \int_{\mathcal{B}(\beta_{n,E},a \delta_n)} \pi_E(z_n) \cdot \exp(-n(C_1+ C_0/2)\cdot  \|\widehat{\beta}_{n, E}^{\;\text{max}} -z_n\|^2/2)dz_n,\\
&  \;\;\;\;\;\;\;\;\;\;\;\;\;\;\;\; \;\;\;\;\;\;\;\;\;\;\;\;  \|\widehat{\beta}_{n, E}^{\;\text{max}} -z_n\|\geq (1-a) \delta_n \text{ for all } z_n \in \mathcal{B'}^c(\beta_{n,E},\delta_n),  \\
& \;\;\;\;\;\;\;\;\;\;\;\;\;\;\;\; \;\;\;\;\;\;\;\;\;\;\;\;  \|\widehat{\beta}_{n, E}^{\;\text{max}} -z_n\|\leq 2a\delta_n \text{ for all } z_n \in \mathcal{B}(\beta_{n,E},a\delta_n)\Bigg)
\end{aligned}
\end{equation*}
\begin{equation*}
\begin{aligned}
&\geq \mathbb{P}_{n,E} \Bigg(  \exp(-C_0\cdot (1-a)^2 n\delta_n^2/4 ) \leq \epsilon\cdot  \Pi_{n,E}(\mathcal{B}(\beta_{n,E},a\delta_n)) \exp(-(C_1+C_0/2)\cdot 4a^2 \delta_n^2/2)\\
&  \;\;\;\;\;\;\;\;\;\;\;\;\;\;\;\; \;\;\;\;\;\;\;\;\;\;\;\;  \|\widehat{\beta}_{n, E}^{\;\text{max}} -z_n\|\geq (1-a) \delta_n \text{ for all } z_n \in \mathcal{B'}^c(\beta_{n,E},\delta_n),  \\
& \;\;\;\;\;\;\;\;\;\;\;\;\;\;\;\; \;\;\;\;\;\;\;\;\;\;\;\;  \|\widehat{\beta}_{n, E}^{\;\text{max}} -z_n\|\leq 2a\delta_n \text{ for all } z_n \in \mathcal{B}(\beta_{n,E},a\delta_n)\Bigg)\\
&\geq \mathbb{P}_{n,E} ( \|\widehat{\beta}_{n, E}^{\;\text{max}} - \beta_{n, E}\| \leq a\delta_n).
\end{aligned}
\end{equation*}
\smallskip

The argument in the last display follows from our assumptions on the selection-informed prior for sufficiently large $n$, coupled with the choice of $a\in (0,1)$.
We complete our proof by showing
$$\textstyle\lim_{n\to \infty} \mathbb{P}_{n,E} ( \|\widehat{\beta}_{n, E}^{\;\text{max}} - \beta_{n, E}\| > a\delta_n) = 0.$$
To this end, we note that the MLE estimating equation is given by:
$$\sqrt{n} \bar{R}^{\intercal}\bar\Theta^{-1} \widehat{\beta}_{n, E}= b_n \nabla C_n (\bar{\beta}^{\;\text{max}}_E)$$
from the surrogate selection-informed (log) likelihood in \eqref{scaled:form:loglik} (Proposition \ref{bounds:lik}) under the assumed parameters.
Further, observing that $C_n(\cdot)$ is strongly convex for sufficiently large $n$, we have
$$(L)^{-2}  \|\sqrt{n}\Sigma_E^{-1}\widehat{\beta}_{n, E}-b_n \nabla {C}_n(\bar{\beta}_E)\|^2 \geq \|\sqrt{n}\widehat{\beta}_{n, E}^{\;\text{max}} - \sqrt{n}\beta_{n, E}\|^2;$$
$L= C_0/2$.
Denoting the exact counterpart of $C_n (\cdot)$ (obtained upon using the exact probability of selection) by $\bar{C}_n(\cdot)$,
we conclude
\begin{equation*}
\begin{aligned}
\mathbb{P}_{n,E} ((b_n)^{-1} \sqrt{n}\|\widehat{\beta}_{n, E}^{\;\text{max}} - \beta_{n, E}\| > a\delta) &\leq (b_n a\delta L)^{-2}\cdot\mathbb{E}_{n,E}(\|\sqrt{n}\bar{R}^{\intercal}\bar\Theta^{-1}\widehat{\beta}_{n, E}-b_n \nabla C_n(\bar{\beta}_E)\|^2)\\
&\leq (b_n a\delta L )^{-2}\cdot\mathbb{E}_{n,E}(\|\sqrt{n}\bar{R}^{\intercal}\bar\Theta^{-1}\widehat{\beta}_{n, E}-b_n \nabla \bar{C}_n(\bar{\beta}_E)\|^2)\\
&+ (a\delta L)^{-2}\cdot \|\nabla C_n(\bar{\beta}_E) - \nabla \bar{C}_n(\bar{\beta}_E)\|^2.
\end{aligned}
\end{equation*}
The first term in the final display clearly converges to $0$ as $n \to \infty$.
The second term converges to $0$, using the result in Proposition \ref{gen_laplace_convergence} combined with the convexity and smoothness of the sequence $C_n(\cdot)$ for large enough $n$.
\end{proof}

\section{Supporting theory (Section \ref{sec:large-sample-theory})}
\label{appendix:supporting-large-sample-theory}

Below, we prove a result on the asymptotic orders of the gradient and Hessian of the (log) Jacobian; this in turn allows us to bound the contribution of the Jacobian term in the Hessian of the (log) likelihood in Proposition \ref{bounds:lik}.

\begin{lemma} \label{lemma_jacobian_derivs}
	For  $\eta > 0$, denote
	\begin{equation}
		\mathcal{K}_{\eta} = \{ x \in \mathbb{R}^{\lvert E \rvert} : \min_j x_j > \eta \}.
	\end{equation}
	We have then the following uniform bounds on the derivatives of the (log)Jacobian:
	\begin{equation*}
	\begin{aligned}
		& \textstyle\sup_{x \in \mathcal{K}_{\eta}} \lVert \nabla_{\gamma} \log J_{\phi}(\gamma ; \mathcal{U}) \big\vert_{b_n x} \rVert_{\infty} &= O(b_n^{-1}), \\
		&\textstyle\sup_{x \in \mathcal{K}_{\eta}} \lVert \nabla^2_{\gamma} \log J_{\phi}(\gamma ; \mathcal{U}) \big\vert_{b_n x} \rVert_{\text{op}} &= O(b_n^{-2}).
	\end{aligned}
	\end{equation*}
\end{lemma}

\begin{proof}
	We begin by deriving an expression for the Hessian $\nabla^2_{\gamma} \log J_{\phi}(\gamma ; \mathcal{U})$. Recall from Theorem~\ref{thm2}, for $g=1, \ldots ,|E|$,
	\begin{equation} \label{j_first}
		\frac{\partial}{\partial \gamma_g} \log J_{\phi}(\gamma ; \mathcal{U}) = \sum_{i \in M_g} \left[ ( \Gamma + C )^{-1} \right]_{ii},
	\end{equation}
	where for simplicity we denote $C = \bar{U}^{\intercal}(X_E^{\intercal}X_E)^{-1}\Lambda \bar{U}$, $\Gamma = \operatorname{diag} ((\gamma_g I_{|g|-1})_{g \in \mathcal{G}_E})$, and $M_g$ denotes the set of indices along the diagonal of $\Gamma$ corresponding to group $g$. Using matrix derivative identities similar to those applied in the proof of Theorem~\ref{thm2}, we derive for $g,h=1, \ldots ,|E|$,
	\begin{equation} \label{j_second}
		\frac{\partial^2}{\partial\gamma_g \partial\gamma_h} \log J_{\phi}(\gamma ; \mathcal{U}) = - \sum_{i \in M_g} \sum_{j \in M_h} [ (\Gamma + C)^{-1} ]_{ij} \cdot [ \{(\Gamma + C)^{-1}\}^{\tp} ]_{ij}
	\end{equation}

	Both our claims rely on a uniform bound on the entries of $(\Gamma + C)^{-1}$.
	Let the operators $s_{\max}$, $\lambda_{\max}$ and $\lambda_{\min}$ denote the largest singular value, largest eigenvalue, and smallest eigenvalue of a matrix respectively.
	Fix $x \in \mathcal{K}_{\eta}$; let $\Gamma(b_n x) = \operatorname{diag} ((b_n x_g I_{|g|-1})_{g \in \mathcal{G}_E})$.
	Denote its dimension by $q = \sum_{g \in E} (\lvert g \rvert - 1)$.
	Then
	\begin{equation*}
	\begin{aligned}
		\max_{1 \leq i,j \leq q} \lvert \left[ (\Gamma(b_n x) + C)^{-1} \right]_{ij} \rvert &\leq s_{\max}\left( (\Gamma(b_n x) + C)^{-1} \right) \\
		&= \lambda_{\max}^{1/2} \left( (\Gamma(b_n x) + C)^{-1} \left[(\Gamma(b_n x) + C)^{-1} \right]^{\tp} \right) \\
		&= \lambda_{\min}^{-1/2} \left( (\Gamma(b_n x) + C)^{\tp} (\Gamma(b_n x) + C) \right). \\
	\end{aligned}
	\end{equation*}
	Now note the following:
	\begin{align*}
		\lambda_{\min} \left( (\Gamma(b_n x) + C)^{\tp} (\Gamma(b_n x) + C) \right) &= \inf_{\lVert v \rVert_2=1} v^{\tp} \left( (\Gamma(b_n x) + C)^{\tp} (\Gamma(b_n x) + C) \right) v \\
		&\geq b_n^2 \min_{1 \leq j \leq \lvert E \rvert} x_j^2 + \lambda_{\min}(C^{\tp} C ) \\
		&\geq (b_n \eta)^2.
	\end{align*}
	Combining the previous two displays uniformly over $\mathcal{K}_{\eta}$, we have that
	\begin{equation}
		\sup_{x \in \mathcal{K}_{\eta}} \left( \max_{1 \leq i,j \leq q} \lvert \left[ (\Gamma(b_n x) + C)^{-1} \right]_{ij} \rvert \right) \leq \frac{1}{b_n \eta} = O(b_n^{-1}).
	\end{equation}
	By \eqref{j_first}, each entry of the gradient is the sum of up to $p-1$ entries of $(\Gamma(b_nx) + C)^{-1}$, and by \eqref{j_second}, each entry of the Hessian is a sum of up to $p^2$ products of entries of the same matrix.
Lastly, a bound on the $\ell_{\infty}$ norm of the gradient follows directly from this element-wise bound.
Further, a bound for the operator norm of the Hessian follows after noting that for an $r \times r$ square matrix $M$, $\lVert M \rVert_2 \leq r \max_{ij} \lvert [M]_{ij} \rvert$.
\end{proof}


With the previous result in hand, we prove the next Lemma used in Proposition~\ref{bounds:lik}.

\begin{lemma} \label{lemma_jacobian_contribution}
	Under the assumptions of Proposition~\ref{bounds:lik}, we have
	\begin{equation*} \label{result_jacobian_contribution}
		\lim_{n \rightarrow \infty} \left\{ \textstyle\sup_{\bar{z} \in \mathcal{C}} \lVert \nabla_{\bar{z}}^2 (b_n)^{-2} \log J_{\phi} (b_n \bar\gamma^{\star}_n(\bar{z}) ; \mathcal{U}) \rVert_{\text{op}} \right\} = 0.
	\end{equation*}
\end{lemma}

\begin{proof}
To proceed with the proof, we derive an expression for the Hessian of the (log) Jacobian.
Recall, $\bar\gamma^{\star}_n(\bar{z})$ is the optimizer of the convex conjugate of the function
$$(\bar\gamma)^{\tp}(\bar{\Sigma})^{-1} \bar\gamma/2 + (b_n)^{-2}\barr (b_n \bar\gamma)$$
evaluated at $\bar{\Sigma}^{-1}(\bar{P}\bar{z} + (b_n)^{-1}\bar{q})$.
By properties of the convex conjugate function and the chain rule,
\begin{equation} \nonumber
	\nabla_{\bar{z}} \bar\gamma^{\star}_n(\bar{z}) = \bar{P}^{\tp} \bar{\Sigma}^{-1} \mathcal{H}_n(\bar{z}),
\end{equation}
where $\mathcal{H}_n(\bar{z})$ denotes the inverse Hessian matrix $\left( \bar{\Sigma}^{-1} + \nabla^2 \barr (b_n \bar\gamma^{\star}_n(\bar{z})) \right)^{-1}$.
Here, we note that the (diagonal) Hessian of the barrier function is positive definite for all $\bar{z}$, which implies $\lVert \mathcal{H}_n(\bar{z}) \rVert_2 \leq \lVert \bar{\Sigma} \rVert_2$ uniformly over all $n$ and $\bar{z} \in \mathcal{C}$.
That is,
\begin{equation*}
        \begin{aligned}
	&  \nabla_{\bar{z}} (b_n)^{-2} \log J_{\phi} (b_n \bar\gamma^{\star}_n(\bar{z}) ; \mathcal{U})
	= (b_n)^{-1} \bar{P}^{\tp} \bar{\Sigma}^{-1} \mathcal{H}_n(\bar{z}) \left( \nabla_x \log J_{\phi}(x ; \mathcal{U}) \big\vert_{b_n \bar\gamma^{\star}_n(\bar{z})} \right).
	\end{aligned}
\end{equation*}
To compute the Hessian, we will use the following identity
\begin{equation} \label{mat_vec_diff}
\pdev{A(x)b(x)}{x} = \pdev{A(x)}{x} \times_2 b(x)^{\tp} + A(x) \pdev{b(x)}{x},
\end{equation}
where $A(x)b(x)$ denotes a matrix-vector product; the 3-dimensional tensor $\pdev{A(x)}{x}$ is summed across its second dimension in the first term of \eqref{mat_vec_diff}.
Observe that the element-wise derivatives of $\mathcal{H}_n$ with respect to the entries of $\bar\gamma^{\star}_n(\bar{z})$ are given by
\begin{equation}
	\frac{\partial}{\partial (\bar\gamma^{\star}_n)_{\ell}} \left[ \mathcal{H}_n(\bar{z}) \right]_{ij} = - b_n \nabla^3_{\ell\ell\ell} \barr(b_n \bar\gamma^{\star}_n(\bar{z})) \left[ \mathcal{H}_n(\bar{z}) \right]_{i\ell} \left[ \mathcal{H}_n(\bar{z}) \right]_{\ell j},
\end{equation}
involving the third derivatives of the barrier function. Then, plugging into the first term of \eqref{mat_vec_diff}, we get
\begin{align*}
    & - b_n \sum_j \nabla^3_{\ell\ell\ell} \barr(b_n \bar\gamma^{\star}_n(\bar{z})) \left[ \mathcal{H}_n(\bar{z}) \right]_{i\ell} \left[ \mathcal{H}_n(\bar{z}) \right]_{\ell j} \left[ \nabla J_{\phi}(b_n \bar\gamma^{\star}_n(\bar{z}); \mathcal{U}) \right]_j \\
    =& - b_n \nabla^3_{\ell\ell\ell} \barr(b_n \bar\gamma^{\star}_n(\bar{z})) \left[ \mathcal{H}_n(\bar{z}) \right]_{i\ell} \sum_j \left( \left[ \mathcal{H}_n(\bar{z}) \right]_{\ell j} \left[ \nabla J_{\phi}(b_n \bar\gamma^{\star}_n(\bar{z}); \mathcal{U}) \right]_j \right).
\end{align*}
Elementwise (row $i$ and column $\ell$), this matrix has the same entries as $\mathcal{H}_n(\bar{z})$, but with the $\ell$th column scaled by
$$
	- b_n \nabla^3_{\ell\ell\ell} \barr(b_n \bar\gamma^{\star}_n(\bar{z})) \left[ \mathcal{H}_n(\bar{z}) \nabla J_{\phi}(b_n \bar\gamma^{\star}_n(\bar{z}); \mathcal{U}) \right]_{\ell}.
$$
Let $b_n \mathcal{D}_n(\bar{z})$ be a diagonal matrix with these entries on its main diagonal.
Then the matrix in the first term of \eqref{mat_vec_diff} is given by $\mathcal{H}_n(\bar{z}) \mathcal{D}_n(\bar{z})$.
The second term of \eqref{mat_vec_diff} equals
\begin{equation*}
	b_n \mathcal{H}_n(\bar{z}) \left( \nabla^2_x \log J_{\phi}(x ; \mathcal{U}) \big\vert_{b_n \bar\gamma^{\star}_n(\bar{z})} \right).
\end{equation*}
Thus, $\nabla_{\bar{z}} (b_n)^{-2} \log J_{\phi} (b_n \bar\gamma^{\star}_n(\bar{z}) ; \mathcal{U})$ equals
\begin{equation*}
	\bar{P}^{\tp} \bar{\Sigma}^{-1} \mathcal{H}_n(\bar{z}) \left( \mathcal{D}_n(\bar{z}) + \nabla^2_x \log J_{\phi}(x ; \mathcal{U}) \big\vert_{b_n \bar\gamma^{\star}_n(\bar{z})} \right) \mathcal{H}_n(\bar{z}) \bar{\Sigma}^{-1} \bar{P}. \label{exp_jacobian_contribution}
\end{equation*}
Noting $\bar{P}^{\tp} \bar{\Sigma}^{-1} \mathcal{H}_n(\bar{z})=O(1)$ uniformly over $n$ and $\bar{z}$, bounding \eqref{exp_jacobian_contribution} in operator norm follows by uniformly bounding the largest diagonal element of $\mathcal{D}_n(\bar{z})$,
and the operator norm of $\nabla^2_x \log J_{\phi}(x ; \mathcal{U})\big\vert_{b_n \bar\gamma^{\star}_n(\bar{z})}$.

Bounds for both terms follow from an application of Lemma~\ref{lemma_jacobian_derivs}, along with the use of the observation that the third derivatives of the barrier function are decreasing.
To complete the proof, it therefore suffices to show uniformly over $\bar{z} \in \mathcal{C}$, and for sufficiently large $n$, all the entries of $\bar\gamma^{\star}_n(\bar{z})$ are bounded below by a constant $\eta$.
To this end, we define the limit of $\bar\gamma^{\star}_n(\bar{z})$ as $n \rightarrow \infty$:
\begin{equation}
	\bar\gamma^{\star}_{\infty}(\bar{z}) = \operatorname{argmin}_{\gamma > 0} \left\{ \frac{1}{2} \gamma^{\tp} \bar{\Sigma}^{-1} \gamma - \gamma^{\tp} \bar{\Sigma}^{-1} \bar{P} \bar{z} \right\}.
\end{equation}
The image of a compact set $\mathcal{C}$ under the continuous map $\bar\gamma^{\star}_{\infty}(\cdot)$ is compact and a subset of the positive orthant, which we call $\mathcal{C}'$.
Define
\begin{equation*}
	\eta = \frac{1}{2} \min \{ |x_j| : x \in \mathcal{C}' \} > 0.
\end{equation*}
Uniform convergence of $\bar\gamma^{\star}_n(\bar{z})$ to $\bar\gamma^{\star}_{\infty}(\bar{z})$ on a compact domain leads us to conclude
$$
\min_j | \left[ \bar\gamma^{\star}_n(\bar{z}) \right]_j | > \eta > 0,
$$
for all $\bar{z} \in \mathcal{C}$ and sufficiently large $n$.
\end{proof}

\section{Supplementary details (Section \ref{sec:empirical-analysis})}
\label{appendix:empirical-analysis}

We outline additional details involving the parameters in our numerical experiments below.
For the simulation instances we generate in the atomic and balanced case analyses, each active coefficient has a random sign with magnitude $\sqrt{2n^{-1} t  \log{p}}$.
We let $t = 0.2$ for the low SNR setting, $t = 0.5$ for the moderate SNR setting, and $t = 1.5$ for the high SNR setting.
In the heterogeneous scenario, the first predictor in the smallest active group has magnitude $\sqrt{2n^{-1} t\log{p}}/|T|$ and the last predictor in the largest group has magnitude $\sqrt{2n^{-1} t \log{p}}$ with magnitudes linearly interpolated for intermediate active coefficients; $T$ is the number of signal variables in the instance.
Each coefficient assumes a random sign and we set $t$ for the low, medium, and high SNR as our previous cases.

For the selection step, we set the grouped penalty weights:
$$\lambda_g = \lambda \rho \sigma \sqrt{2 \log{p} \dfrac{ \lvert g \rvert}{ \bar{g} }}$$
for solving the Group LASSO in both the randomized (``Selection-informed") and non-randomized formulations (``Naive" and "Split");
$\lvert g \rvert$ is the number of features in group $g$, and $\bar{g}$ is the floor of the average group size.
In the choice of the penalty weights, $\rho=r$ the proportion of data used for the query for ``Split" and takes the value $1$ for ``Selection-informed" and ``Naive" when we solve the Group LASSO and the overlapping Group LASSO and $\rho=\sqrt{r}$ for ``Split" and $1$ for the other methods when we solve the standardized Group LASSO.
Clearly, in our balanced settings, we impose a uniform penalty across all groups, while the penalties for the heterogeneous settings scale with the size of our groups.

Addressing selection-informed inference post the Group LASSO, the barrier function used in our optimization problem (see Theorem \ref{thm2}) is given by
$$\barr (\gamma) = \textstyle\sum_{g\in \mathcal{G}_{E}} \log(1+ (\gamma_g)^{-1}).$$
Observe, this choice of penalty assigns higher preference to optimizing variables away from the boundary of the selection region $[0,\infty)^{\mathbb{R}^{|\mathcal{G}_E|}}$.
For executing the sampler, we set the initial draw as follows
$$\beta^{(0)} = \widehat{\beta}_E, 
$$
the refitted least squares estimate in our setup.
Completing our specifications, the inferential results we report in Section \ref{sec:empirical-analysis} are based upon $1500$ draws of the Langevin sampler. We discard the first $100$ samples as burn-in retaining the remainder for uncertainty estimation.
The code for our experiments in the paper is available here: \url{https://github.com/snigdhagit/selective-inference/tree/group_LASSO/selection/randomized}.

\subsection{Supplementary details for HCP analysis}
\label{appendix:HCP}
The ``preprocessed'' version of the dataset used in our analysis, which had undergone the processing stream described in \citet{glasserMinimalPreprocessingPipelines2013},  was downloaded from the HCP's ConnectomeDB platform \citep{marcusInformaticsDataMining2011}.
The fMRI data comprises time courses at many ``voxels'' throughout the brain that are typically each a few millimeters cubed in volume.
This data was preprocessed as described in \citet{sripadaBasicUnitsInterIndividual2019}, excluding the steps that are specific to resting state processing.
While the HCP data includes a variety of imaging modalities, we utilize both behavioral and functional magnetic resonance imaging (fMRI) measurements recorded from a cognitive task, namely the ``N-back'' task \citep{barchFunctionHumanConnectome2013}.

In the the ``N-back'' task, participants are presented with a sequence of pictures about which they make judgments, and their accuracy and brain activity is recorded while they perform the task.
There are two different conditions of principle interest, each of which are presented in blocks.
In the 0-back condition, participants simply judge whether each item is the same as the item presented at the beginning of the block.
In the 2-back condition, participants judge whether each item is the same as the item presented two trials previous.
As may be intuitively clear, the 2-back condition is appreciably more demanding with respect to working memory.
A common approach for analyzing fMRI data involves the construction of ``contrasts.''
Measuring activity during the 2-back condition would likely indicate activity related to working memory, but it would also include activity indicating many other phenomena such as visual processing, motor activation in order to press buttons to indicate judgments, etc.
These phenomena are not of primary interest, so we consider a contrast formed by subtracting the activation during the 0-back condition from the activation during the 2-back condition.
This 2-back minus 0-back contrast is standard for the N-back task \citep{barchFunctionHumanConnectome2013}.
Contrasts were obtained using in-house processing scripts that use SPM12.
The standardized accuracy of each participant during this task will be our target of prediction $y$ and we will use the contrast as the predictor $X$.
As a preprocessing step, columns of the design matrix $X$ are adjusted to have mean $0$ and unit norm.

Using the contrast value from each voxel results in very high dimensional data, and analysis is sometimes instead performed at the level of ``regions of interest'' (ROIs).
This provides a means of effectively downsampling the data by aggregating information at each ROI, which is a spatially contiguous group of voxels.
These ROIs can be defined \emph{a priori} according to one of a variety of atlases, and this aids interpretability and enables comparisons of findings across studies that use the same atlas.
We use ROIs as defined by the ``Power Parcellation'' \citep{powerFunctionalNetworkOrganization2011}.
In addition to being a broadly popular atlas, the Power Parcellation is also noteworthy in that it assigns each of its $264$ ROIs to a ``brain system.''
The spatial coordinates of the ROIs, as well as their assignment to brain systems, are described in \citet{powerFunctionalNetworkOrganization2011}.
The MarsBar utility \citep{marsbar} was used to extract contrast values for each of these ROIs.
Of the $264$ ROIs, $236$ are assigned to one of 13 distinct, named brain systems while the remainder are simply labeled ``unknown'' and in our analysis we use only these $236$ positively labeled ROIs as predictors in our regression.
Because each of these brain systems is putatively believed to underlie a discrete set of functions (e.g., because they typically coactivate for a given type of task), we partition our predictors into groups by brain system label, and then use the Group LASSO to predict accuracy on the N-back task using data from these $236$ ROIs.
While inference may be performed at the level of individual ROIs, it is also useful to interrogate effects at a system-wide level.
Further averaging all of the ROIs within a single system may be too coarse and obscure useful signal, so the Group LASSO provides a means of allowing each ROI to make a distinct predictive contribution while still performing selection at the interpretable level of entire brain systems.
Because in this application $n > p$, we estimate $\hat{\sigma}^2 = \left( n - p \right)^{-1} \left\lVert y - X \left( X^{\intercal} X \right)^{-1} X^{\intercal} y \right\rVert_2^2$.
We use the same value for $\hat{\sigma}^2$ for the intervals obtained via data splitting.
We set $\lambda_g$ for each group as described in \ref{appendix:empirical-analysis} and set the randomization level $\tau$ to satisfy \eqref{rand:level} at varying levels of $r$.
Choosing $\lambda = 1$ (as we did for the simulation studies) yields a fully dense model, so we increase to $\lambda = 10$ which selects just a single group.

\subsection{Supplementary numerical comparison}
\label{appendix:new:num:comparison}

We conduct an additional numerical experiment to compare the methods of \citet{yang2016selective} and also \citet{loftus2015selective}.
    Specifically, we consider one instance of the simulation settings considered in \citet{yang2016selective} where we draw $X \in \mathbb{R}^{500 \times 500}$ with entries independently and identically distributed as $\mathcal{N} \left(0, \frac{1}{500}  \right)$. 
    The $p = 500$ features are arranged into 50 contiguous groups of 10 features each.
    The first 10 groups (i.e., first 50 features) are all active with associated coefficient $1.5$ and the remainder are inactive with associated coefficients $0$, i.e., $\beta = \begin{bmatrix} 1.5 \cdot 1_{50}^{\intercal} & 0_{450}^{\intercal} \end{bmatrix}^{\intercal}$.
    The response $y$ is then generated as $\mathcal{N} \left( \mu, 1 \right)$, where $\mu = X \beta$.
     We generate a single realization of the data in this setting and then apply the methods of \citet{yang2016selective}, \citet{loftus2015selective}, and our method conducted with $5000$ posterior samples (with $100$ samples discarded as burn-in).
   Findings in this instance gives us an opportunity to note the extent of agreement between all the three methods.

    For all methods, the first stage is automatically selecting groups using: (i) the Group LASSO (for \citet{yang2016selective}), (ii) the randomized Group LASSO (for our method) in \eqref{glasso}, or (iii) forward stage-wise selection (for \citet{loftus2015selective}).
    We use $\lambda = 4$ for the approach of \citet{yang2016selective} which yields the selection of 11 active groups (i.e., 110 features), and we then tune parameters for the other two methods to select the same number of active features.
    Once the model has been selected, we proceed to inference with $\alpha = 0.1$.
    
    For the methods by \citet{yang2016selective} and \citet{loftus2015selective}, the inferential target for each selected group $g$ in $E$ is an overall group effect $\mu_g$, which we review in more detail under Section \ref{sec:overall-effects}.
    We apply our methods to construct credible intervals for the individual components of the coefficient vector for each group; inference for individual effects in the selected groups is not addressed by the previous two methods.
    As described in Section \ref{sec:empirical-analysis}, sampling from the selection-informed posterior with a diffuse (non-informative) prior yields credible intervals for the individual effects with ``good" frequentist properties. 
    Note, we do not pursue inference for the overall group effect---a (non-linear) function of the selection-informed parameters $\beta_E$---using our Bayesian methods.
    This is because our focus is on the extent of agreement between all three methods in terms of their frequentist properties.
     Specially, ``good frequentist properties" for $\mu_g$ will also depend on a choice of prior for this parameter; in this case, a non-informative prior for $\beta_E$ might not be non-informative for $\mu_g$ for $g\in \mathcal{G}_E$.

    We summarize our results in Table \ref{tab:method-comparison}.
    Each of the three methods selects all of the 5 active groups and 6 additional inactive groups, although the identities of the selected inactive groups differ slightly across the methods due to differences in the query (i.e., Group LASSO vs randomized Group LASSO vs forward stage-wise selection).
The method of \citet{loftus2015selective} yields no significant p-values at $\alpha = 0.10$: it makes no Type I errors, but $5$ Type II errors.
The method of \citet{yang2016selective} correctly rejects the null for 4 of 5 active groups and only makes a Type I error for 1 of 6 six inactive groups.
For our method, we report the component-wise coverage of our marginal credible intervals in each group. 
Empirically, we appear to have coverage that does not appreciably deviate from nominal (i.e., $90\%$).
Coverage for active coefficients is slightly better at $88\%$ as opposed to coverage for inactive coefficients at $80\%$, although these may just be chance fluctuations.
 In summary, the test by \citet{loftus2015selective} seems more conservative than the remaining two methods. 
 Whereas, coverage for the individual variable effects in each active group by our method seem to be consistent with the lower bounds for the overall group effect by \citet{yang2016selective}. 
    \begin{table}[h]
      \centering
      \begin{tabular}{rrrrrr}
        Group \# & \(\mu_g\) & Loftus \$p\$-value & Yang LCB & Yang \$p\$-value & Ours (Coverage)\\
        \hline
        1 & 4.49 & 0.34 & 2.39 & 0.00 & 1.0\\
        2 & 4.01 & 0.14 & 1.11 & 0.03 & 0.9\\
        3 & 4.19 & 0.42 & 1.01 & 0.03 & 0.9\\
        4 & 4.07 & 0.32 & -2.68 & 0.33 & 0.8\\
        5 & 4.20 & 0.23 & 3.05 & 0.00 & 0.8\\
        8 & 0 & 0.55 &  &  & \\
        11 & 0 &  &  &  & 0.7\\
        14 & 0 &  &  &  & 0.7\\
        17 & 0 & 0.77 &  &  & \\
        18 & 0 &  & -9.99 & 0.81 & \\
        20 & 0 & 0.54 & -0.95 & 0.21 & 0.9\\
        28 & 0 &  & -9.54 & 0.66 & 1.0\\
        33 & 0 & 0.78 & 1.74 & 0.02 & \\
        36 & 0 & 0.83 & -4.08 & 0.45 & 0.7\\
        39 & 0 &  & -3.05 & 0.50 & 0.8\\
        46 & 0 & 0.83 &  &  & \\
      \end{tabular}
      \caption{
        Comparison of Inferential Results on Single Realization of Synthetic Data using Methods of \citet{loftus2015selective}, \citet{yang2016selective}, and the proposed method.
        Blanks indicate that the associated method did not select the variable group depicted in the corresponding row.
        LCB signifies lower confidence bound.
        Coverage (where applicable) was assessed using $90\%$ credible intervals.
        ``Ours'' refers to the Selection-informed method discussed in the manuscript.
      }
      \label{tab:method-comparison}
    \end{table}

\end{document}